\documentclass[aps, pra, superscriptaddress, 12pt, longbibliography, nofootinbib, tightenlines, amsfonts]{revtex4-2}

\usepackage{microtype}

\renewcommand{\thesection}{\arabic{section}}
\renewcommand{\thesubsection}{\thesection.\arabic{subsection}}
\renewcommand{\thesubsubsection}{\thesubsection.\arabic{subsubsection}}

\makeatletter
\renewcommand{\p@subsection}{}
\renewcommand{\p@subsubsection}{}
\makeatother

\usepackage{amsmath,amsfonts,amssymb,amsthm}
\usepackage{mathtools}
\usepackage{setspace}
\usepackage[normalem]{ulem} 

\usepackage{graphicx}
\usepackage[dvipsnames]{xcolor}
\usepackage{url}

\usepackage{verbatim}
\usepackage{alltt}
\usepackage{moreverb}
\usepackage{bm}
\usepackage{bbold}

\usepackage{xr}

\usepackage{tikz}
\usetikzlibrary{quantikz}

\newcommand{\listb}[1]{\{#1\}}

\newcommand{\ketbra}[1]{\KBop{#1}{#1}}

\usepackage[normalem]{ulem} 
\providecommand{\ignore}[1]{}
\providecommand{\auedit}[1]{#1}

\newif\ifcmnt
\cmnttrue
\ifdefined\cmntsoff\cmntfalse\fi

\ifcmnt
\providecommand{\aucmnt}[1]{#1}

\providecommand{\SGcolor}{\color{Magenta}}

\else
\providecommand{\aucmnt}{}
\renewcommand{\aucmnt}[1]{}
\renewcommand{\auedit}[1]{}

\providecommand{\SGcolor}{}
\renewcommand{\SGcolor}[1]{}

\fi

\usepackage[dvipsnames]{xcolor}
\usepackage[colorlinks=true, urlcolor=violet, linkcolor=blue, citecolor=red, hyperindex=true, linktocpage=true]{hyperref}
\usepackage{graphicx}
\usepackage{enumitem}
\usepackage{stmaryrd} 
\usepackage{xpatch} 

\usepackage{amsmath,amssymb,amsfonts,mathtools,dsfont,relsize,scalerel}
\usepackage{suffix} 

\usepackage{comment}

\newcommand{\vpd}[0]{\vphantom{\dagger}}

\newcommand{\vpp}[0]{\vphantom{\prime}}

\numberwithin{equation}{section}

\makeatletter
\patchcmd{\@ssect@ltx}
    {\addcontentsline{toc}{#1}{\protect\numberline{}#8}}
    {}
    {}
    {}
\makeatother

\renewcommand{\thesection}{\arabic{section}}
\renewcommand{\thesubsection}{\thesection.\arabic{subsection}}
\renewcommand{\thesubsubsection}{\thesubsection.\arabic{subsubsection}}
\makeatletter
\renewcommand{\p@subsection}{}
\renewcommand{\p@subsubsection}{}
\def\l@subsubsection#1#2{}
\makeatother

\makeatletter
\g@addto@macro\bfseries{\boldmath}
\makeatother

\usepackage{tikz}
\usetikzlibrary{quantikz}

\usepackage{amsthm}
\newtheorem{thm}{Theorem}
\newtheorem{cor}[thm]{Corollary}
\newtheorem{lem}[thm]{Lemma}
\newtheorem{prop}[thm]{Proposition}

\newtheorem{defn}[thm]{Definition}

\newtheorem{mainthm}{Theorem}

\DeclarePairedDelimiter{\norm}{\lVert}{\rVert}
\DeclarePairedDelimiter{\abs}{\lvert}{\rvert}
\DeclarePairedDelimiter{\expval}{\langle}{\rangle}

\newcommand{\ident}[0]{\mathds{1}}

\newcommand{\ii}[0]{\mathrm{i}}

\newcommand{\bvec}[1]{\boldsymbol{#1}}

\newcommand{\kron}[1]{\delta^{\vpp}_{#1}}\WithSuffix\newcommand\kron*[1]{\delta_{#1}}

\newcommand{\Reals}{\mathbb{R}}
\newcommand{\Comps}{\mathbb{C}}
\newcommand{\Ints}{\mathbb{Z}}
\newcommand{\Nats}{\mathbb{N}}

\newcommand{\hilbert}[0]{\mathcal{H}}

\newcommand{\KBop}[2]{\left| #1 \middle\rangle \hspace{-0.4mm} \middle\langle #2 \right|}

\newcommand{\inprod}[2]{ \left\langle #1 \middle| #2 \right\rangle}

\newcommand{\matel}[3]{\left\langle #1 \middle| #2 \middle| #3 \right\rangle}
\WithSuffix\newcommand\matel*[3]{\langle #1 | #2 | #3 \rangle}

\DeclareMathOperator*{\trace}{tr}
\newcommand{\tr}[1]{ \trace \left(  #1  \right)}\WithSuffix\newcommand\tr*[1]{\trace ( #1 )}
\newcommand{\Tr}[1]{ \trace \left[  #1  \right]}\WithSuffix\newcommand\Tr*[1]{\trace [ #1 ]}

\newcommand{\cliff}[0]{\mathbf{C}}
\newcommand{\Cliff}[1]{\cliff^{\vpp}_{\hspace{-0.1mm} #1}}\WithSuffix\newcommand\Cliff*[1]{\cliff_{\hspace{-0.3mm} #1}}

\newcommand{\pauliset}[0]{\mathbf{P}}
\newcommand{\PauliSet}[1]{\pauliset^{\vpp}_{\hspace{-0.7mm} #1}}\WithSuffix\newcommand\PauliSet*[1]{\pauliset_{\hspace{-0.7mm}  #1}}

\newcommand{\pauligroup}[0]{\mathbb{P}}
\newcommand{\PauliGroup}[1]{\pauligroup^{\vpp}_{\hspace{-0.7mm} #1}}\WithSuffix\newcommand\PauliGroup*[1]{\pauligroup_{\hspace{-0.7mm}  #1}}

\newcommand{\Pauli}[1]{P^{\vpp}_{\hspace{-0.2mm} #1}}\WithSuffix\newcommand\Pauli*[1]{P^{\,}_{\hspace{-0.2mm} #1}}

\newcommand{\PauliChan}[1]{\mathcal{P}^{\vpp}_{\hspace{-0.7mm} #1}}\WithSuffix\newcommand\PauliChan*[1]{\mathcal{P}_{\hspace{-0.7mm} #1}}

\newcommand{\Logicals}[0]{\mathcal{L}}

\newcommand{\Stabs}[0]{\mathcal{S}}

\newcommand{\stab}[0]{\mathsf{S}}

\newcommand{\nlog}[0]{k}

\newcommand{\nphys}[0]{n}

\newcommand{\POVM}[0]{\Pi}

\newcommand{\sinput}{s_{\text{in}}}
\newcommand{\soutput}{s_{\text{out}}}
\newcommand{\smeas}{s_{\text{meas}}}
\newcommand{\serr}{s_{\text{err}}}

\newcommand{\ssp}{syndrome decoding symmetry property}
\newcommand{\Ssp}{Syndrome decoding symmetry property}

\newcommand{\smip}{syndrome marginal independence property} 
\newcommand{\Smip}{Syndrome marginal independence property}
\newcommand{\ssig}{\sqrt{\Sigma}}

\newcommand{\prepvec}[0]{r_P}

\newcommand{\measvec}[0]{m_P}

\newcommand{\elt}[0]{i}

\newcommand{\ncyc}[0]{K} 

\newcommand{\cycind}[0]{k} 

\newcommand{\CFrame}[0]{U}

\newcommand{\Cor}[0]{\mathcal{C}}

\newcommand{\ErrChan}[0]{\mathcal{E}}

\newcommand{\zproj}[0]{\pi}
\newcommand{\ZProj}[1]{\zproj (#1)}\WithSuffix\newcommand\ZProj*[1]{\zproj \left( #1 \right)}

\newcommand{\pweight}[0]{c}

\newcommand{\PWeight}[1]{\pweight^{({#1})}_{\ell,\bvec{s}_{#1}}}\WithSuffix\newcommand\PWeight*[2]{\pweight^{({#1})}_{#2}}

\newcommand{\pprob}[0]{\lambda}

\newcommand{\PProb}[1]{\pprob^{({#1})}_{\Pauli*{\ell_{#1}} \otimes Q^{\,}_s \otimes Q^{\,}_a \otimes Q^{\,}_o}}\WithSuffix\newcommand\PProb*[2]{\pprob^{({#1})}_{#2}}

\newcommand{\cprob}[0]{p}
\newcommand{\ocprob}[0]{\overline{\cprob}}

\newcommand{\CProb}[2]{\cprob^{({#1})}_{#2}}\WithSuffix\newcommand\CProb*[1]{\cprob^{(#1)}_{\bvec{\ell}_{#1},\bvec{s}_{#1}}}

\newcommand{\OCProb}[2]{\ocprob^{({#1})}_{#2}}\WithSuffix\newcommand\OCProb*[1]{\ocprob^{(#1)}_{\bvec{\ell}_{#1},\bvec{s}_{#1}}}

\newcommand{\tmat}[0]{W}

\newcommand{\TMat}[1]{\tmat^{(#1)}_{\ell_{#1},\bvec{s}_{#1};\ell_{{#1} - 1},\bvec{s}_{{#1}-1}}}\WithSuffix\newcommand\TMat*[1]{\tmat^{(#1)}}
\newcommand{\TMatrix}[1]{\tmat^{\vpp}_{\ell_{#1},\bvec{s}_{#1};\ell_{{#1} - 1},\bvec{s}_{{#1}-1}}}\WithSuffix\newcommand\TMatrix*[1]{\tmat_{\ell_{#1},\bvec{s}_{#1};\ell_{{#1} - 1},\bvec{s}_{{#1}-1}}}

\newcommand{\stmat}[0]{\gamma}

\newcommand{\STMat}[1]{\stmat^{(#1)}_{\bvec{s}_{#1},\bvec{s}_{{#1}-1}}}\WithSuffix\newcommand\STMat*[1]{\stmat^{(#1)}}
\newcommand{\STMatrix}[1]{\stmat^{\vpp}_{\bvec{s}_{#1},\bvec{s}_{{#1}-1}}}\WithSuffix\newcommand\STMatrix*[1]{\stmat_{\bvec{s}_{#1},\bvec{s}_{{#1}-1}}}

\newcommand{\ltmat}[0]{\Lambda}

\newcommand{\LTMat}[1]{\ltmat^{(#1)}_{\bvec{s}_{#1},\bvec{s}_{{#1}-1}}}\WithSuffix\newcommand\LTMat*[1]{\ltmat^{(#1)}}
\newcommand{\LTMatrix}[1]{\ltmat^{\vpp}_{\bvec{s}_{#1},\bvec{s}_{{#1}-1}}}\WithSuffix\newcommand\LTMatrix*[1]{\ltmat_{\bvec{s}_{#1},\bvec{s}_{{#1}-1}}}

\begin{document}
\title{Constructing an approximate logical Markovian model of consecutive QEC cycles of a stabilizer code}
\author{Alex Kwiatkowski}
\affiliation{National Institute of Standards and Technology, Boulder, Colorado 80305, USA}

\author{Aaron J. Friedman}
\affiliation{National Institute of Standards and Technology, Boulder, Colorado 80305, USA}
\affiliation{Department of Electrical Engineering, University of Colorado, Denver, Colorado 80204, USA}

\author{Shawn Geller}
\affiliation{National Institute of Standards and Technology, Boulder, Colorado 80305, USA}

\author{Jalan A. \surname{Ziyad}}
\affiliation{Quantum Performance Laboratory, Sandia National Laboratories, Albuquerque, NM, 87185, USA}
\affiliation{Center for Quantum Information and Control, University of New Mexico, Albuquerque, NM, 87131, USA}
\affiliation{Department of Physics and Astronomy, University of New Mexico, Albuquerque, NM, 87131, USA}

\author{Scott Glancy}
\affiliation{National Institute of Standards and Technology, Boulder, Colorado 80305, USA}

\author{Emanuel Knill}
\affiliation{National Institute of Standards and Technology, Boulder, Colorado 80305, USA}
\affiliation{Center for Theory of Quantum Matter, University of Colorado, Boulder, Colorado 80309, USA}

\date{\today}

\begin{abstract}
  As quantum error correction (QEC) experiments continue to make rapid progress,
  there is increased interest in designing experiments with guarantees of logical performance.
  At present, one difficulty is the lack of a clear connection between logical performance 
  and the low-level error models.
  In this work, we take an important step toward addressing this issue by proving that 
  consecutive QEC cycles of a stabilizer code with Pauli stochastic noise and with
  a single-cycle infidelity $\epsilon_1 \leq 1/64$
  admit an approximate logical Markovian model, meaning that consecutive noisy QEC cycles
  can be modeled by a memoryless error process acting only on the logical subsystem.
  The approximate logical Markovian model can be computed from the low-level error model,
  and the deviations from the true behavior are exponentially suppressed in the number of QEC cycles.
  Consequently, we expect that the approximate logical Markovian model will be both a useful tool for logical characterization 
  and an aid for designing stabilizer-code implementations with guarantees of logical performance.
\end{abstract}

\maketitle

\section{Introduction}

Recent progress in the development of experimental quantum computing devices has centered around early 
demonstrations of quantum error-correction capabilities~\cite{acharyaQuantumErrorCorrection2025, bluvsteinLogicalQuantumProcessor2024, paetznickDemonstrationLogicalQubits2024,reichardtDemonstrationQuantumComputation2024,postler_monz_24_steane,lacroix2024scalinglogiccolorcode}.
Many of these experiments focus on demonstrating error correction in stabilizer codes~\cite{gottesman2016surviving, QC_book, gottesman1997stabilizercodesquantumerror}, 
including preparing encoded stabilizer states, performing logical Clifford operations, 
and making logical stabilizer measurements.
Other experiments focus on quantum memory, which involves applying many consecutive QEC cycles in order
to protect logical information for an extended period of time.
There is substantial interest both in predicting logical behavior in such experiments and 
in designing experiments that minimize the logical error rate,
given a low-level error model.
However,  the connection between low-level error models and logical performance is not always clear. 
For a single QEC cycle, logical performance depends on the complex interplay between the code, noise, and decoder~\cite{iyer2018small,rudinger2023probing, beale_quantum_2018, rahn2002exact}. 
For multiple QEC cycles, creating models of logical performance is further complicated by the phenomenon of emergent logical non-Markovianity~\cite{caesura2021non,Ziyad2025Emergent}. 
For example, forthcoming work announced in Ref.~\citenum{Ziyad2025Emergent} investigates a model of repeated QEC cycles in a stabilizer code, subject to Markovian noise on the qubits of the code, 
and finds that the performance of a given QEC cycle depends on which syndrome-measurement errors occurred in previous cycles.
Essentially, the syndrome degrees of freedom in an error-correcting code act as a memory for the logical register, with which they become correlated. 
Consequently, performance metrics derived from the behavior of a single QEC cycle may not quantitatively capture the behavior of many consecutive QEC cycles.

In this paper, we take an important step toward addressing these issues by studying 
the logical behavior of many consecutive QEC cycles of a stabilizer code. 
We show that,  under broad circumstances, 
the behavior of many consecutive QEC cycles is well approximated by modeling 
the errors in each individual QEC cycle  by a memoryless logical error process, 
which we refer to as an approximate logical Markovian model.
The precise statement of this result is given in Theorem~\ref{thm_main}, which we view as the main contribution of our work.
Several assumptions enter Theorem~\ref{thm_main}, 
which we discuss in detail in Sec.~\ref{main_thm_statement} (see Fig.~\ref{fig_assumptions})
and briefly summarize here.
In particular, we require (\emph{i}) Pauli-stochastic noise (see  Secs.~\ref{sec_Pauli_stochastic}--\ref{sec_Pauli_noise_meas});
(\emph{ii}) a particular property of the syndrome-extraction procedure (see Def.~\ref{def_ssp}) 
or the ability to randomize the syndrome subsystem prior to the QEC cycle (see Def.~\ref{defn_rand_synd}); 
and 
(\emph{iii}) that the single-cycle logical entanglement fidelity (see Def.~\ref{defn_log_fid}) is sufficiently high ($f_1 \geq 1 - 1/64$).
The first two assumptions can be enforced via Pauli twirling, 
and we expect the third assumption to hold for any reasonably high-quality QEC experiment.
Altogether, we expect the vast majority of quantum-memory experiments to satisfy these assumptions.
Furthermore, the proof of Theorem~\ref{thm_main} is constructive, 
in that it provides a method to calculate the approximate logical Markovian model given a particular low-level error model.
In addition, our definition of approximate logical Markovian model requires that the deviations between true probabilities of outcomes 
in logical experiments and the probabilities of outcomes predicted by the approximate logical Markovian model are exponentially suppressed as $\ncyc$ increases.
Thus, if an approximate logical Markovian can be constructed, it has a strong guarantee of accuracy.

In total, we expect that the approximate logical Markovian model provides a practical and informative 
benchmark for quantum-memory experiments that is independent of 
the particular platform, the stabilizer code, or the implementation thereof. 
Furthermore, the approximate logical Markovian model can be used as a metric for designing 
stabilizer codes and fault-tolerant implementations, 
based on the relevant low-level error model.

The rest of the paper is organized as follows.
In Sec.~\ref{sec_prelim}, we provide a preliminary discussion of notation and basic concepts, 
including a brief overview of stabilizer codes.
In Sec.~\ref{sec_prelim_noiseless}, we provide a standardized form of a QEC cycle in a stabilizer code which we use for our analysis.
In Secs.~\ref{sec_Pauli_stochastic}--\ref{sec_Pauli_noise_meas}, 
we define Pauli-stochastic noise, along with particular noise models in the contexts of logical state preparation, QEC cycles, and logical measurement.
In Sec.~\ref{sec_approx_model}, we define a notion of total logical error channel for any QEC experiment,
which depends on the number of QEC cycles $\ncyc$.
We then introduce the concept of an approximate logical Markovian model, which is \emph{independent} of $\ncyc$, but approximates 
the behavior of the total logical error channel. 
In Sec.~\ref{sec_Pauli_eigs}, we show that the ability to approximate all Pauli eigenvalues of the total logical error channel 
by exponentially decaying functions of $\ncyc$ implies an approximate logical Markovian model.
In Sec.~\ref{sec_smip}, we introduce the \smip{}, which is a property of the syndrome-extraction procedure and is a key requirement of Theorem~\ref{thm_main}.
We explain that this property holds for many realizations of stabilizer codes and can also be enforced via Pauli twirling if needed.
In Sec.~\ref{sec_eff_markovian}, we introduce the notion of a ``high quality'' QEC cycle, 
which is formally captured by the single-cycle logical entanglement fidelity $f_1$ (see Def.~\ref{defn_log_fid}) and is
another key requirement for Theorem~\ref{thm_main}. We then state and prove Theorem~\ref{thm_main}.

\section{Preliminaries}
\label{sec_prelim}

In this work, we assume 
familiarity with standard concepts related to quantum error correction with qubits 
(see Ref.~\citenum{gottesman2016surviving} for an introduction).
We now introduce the notational and terminological conventions used in the remainder.

\subsection{Pauli operators and related notation}
\label{sec_prelim_Pauli}

Consider a system $A$ with $\nphys$ qubits. For each $j \in \left\{ 1, \ldots, \nphys \right\}$, 
we denote by $\PauliSet*{j} = \{\ident_j, X_j, Y_j, Z_j\}$ the set of single-qubit Pauli operators on qubit $j$. 
A basis for the operators on $A$ is given by the $4^{\nphys}$ Pauli operators, 
which make up the set
\begin{equation}
    \label{eq:Pauli set}
    \PauliSet{A} = \left\{ \, Q_1 \otimes \cdots \otimes Q_\nphys \, \middle|~Q_j \in \PauliSet*{j} \, \right\} \, , ~~ 
\end{equation}
which are tensor products of Pauli matrices on each of the $\nphys$ qubits in $A$.
Relatedly, the Pauli group $\PauliGroup*{A}$ is the group generated by the Pauli operators~\eqref{eq:Pauli set} under multiplication,
\begin{align}
  \label{eq:Pauli group}
  \PauliGroup{A} = \expval{ \PauliSet{A} } &= \left\{\, \ii^\nu Q_1\otimes\cdots\otimes Q_{\nphys}\,\middle|\;\nu \in \left\{ 0, 1,2,3 \right\},\,Q_j \in \mathbf{P}_j \right\} \, ,~
\end{align}
where $\expval{S}$ denotes the subgroup of the 
matrix group $\mathrm{GL}_{2^{\nphys}}(\Comps)$ generated by a subset $S$.

We frequently work in the computational basis and we label computational basis elements by bitstrings. 
For a system $A$ with $\nphys$ qubits, we define a ``bitstring of $A$'' as an element of $\Ints_2^{\nphys}$ 
(i.e., $\nphys$-digit bitstrings). These bitstrings label computational basis states, 
which are simultaneous eigenstates of all Pauli $Z$ operators. For a bitstring $s \in \Ints_2^{\nphys}$,
we have that $Z_j \ket{s} = (-1)^{s_j} \ket{s}$ with $s_j$ the $j$th bit of $s$.
We denote by $\ZProj{s} = \KBop{s}{s}$ the projector onto the computational basis element corresponding to the bitstring $s \in \Ints_2^\nphys$.

We often consider dephased qubits, 
for which the density matrix can be expressed as a convex combination of projectors $\ZProj{s}$.
We also denote bit flips acting on bitstrings $s \in \Ints_2^\nphys$ using $+$, which acts as bitwise XOR. 
Using the shorthand $X^{s} = \prod_{j=1}^{\nphys} X_j^{s_j}$, e.g., we write
\begin{equation}
  \label{eq:bit flips}
  X^s \, \ZProj{s'} X^s = \ZProj{s + s'}  = \KBop{s+s'}{s+s'}\, .
\end{equation}
We also define $\PauliChan*{i}$ as the Pauli superoperator associated with the $i$th Pauli operator, i.e.,
\begin{equation}
  \label{eq:Pauli superoperator}
  \PauliChan{i} (\rho) = \Pauli{i} \rho \Pauli{i} \, . ~~
\end{equation}
so that general Pauli channels are a mixture of such Pauli superoperators. 
In particular, if $\Lambda$ is a Pauli channel with
mixture probability $p_i$ for the $i$th Pauli operator, then $\Lambda$ acts as
\begin{equation}
  \label{eq:general Pauli channel}
  \Lambda (\rho) = \sum_i p^{\,}_i \, \PauliChan{i} (\rho) = \sum_i p^{\,}_i \, \Pauli{i} \rho \Pauli{i} \, . ~~
\end{equation}

\subsection{Stabilizer error-correcting codes}
\label{sec:stabilizer qec}

We restrict our analysis to stabilizer codes~\cite{CalderbankGood, GottesmanIntro, SteaneQEC1, SteaneQEC2, Shor_1995, lidar_quantum_2013}. 
We consider arbitrary stabilizer codes on $\nphys$ qubits that encode $\nlog$ logical qubits. 
A stabilizer code is specified by a set $\Stabs \subseteq \PauliSet*{\nphys}$ of $\nphys-\nlog$ multiplicatively independent 
commuting Pauli operators such that the group generated by $\Stabs$ does not contain $-\ident$.
The elements of $\Stabs$ are referred to as the stabilizer generators, 
while generic elements of $\expval{\Stabs}$ are referred to as stabilizers.

We also define the group of logical operators $\Logicals = \PauliGroup*{\nphys} / \expval{\Stabs}$. 
To define the \emph{algebra} of logical operators, 
we first fix a set $\Logicals_*$ of coset representatives of $\Logicals$ called the \emph{bare} logical operators,
\begin{equation}
\label{eq:bare logicals}
\Logicals_* = \{  X_{\text{L}}^{(1)} , Z_{\text{L}}^{(1)}, \dots ,  X_{\text{L}}^{(\nlog)} , Z_{\text{L}}^{(\nlog)} \}\subset \mathbb{P}_n \, ,~~
\end{equation}
which satisfy the Pauli group relations on $\nlog$ qubits.
An element of the algebra generated by $\Logicals_*$ is called a logical operator.
This defines a $\nlog$-qubit Hilbert space $\hilbert_L$, 
and the full Hilbert space $\hilbert$ decomposes as $\hilbert = \hilbert_L \otimes \hilbert_S$, 
where $\hilbert_S$ is a Hilbert space of $\nphys-\nlog$ qubits, corresponding to the stabilizer generators.
We note that measuring the bare logical operators~\eqref{eq:bare logicals} is insufficient to extract logical information;
one must also measure the stabilizer generators to interpret the outcomes of logical measurements~\cite{GottesmanIntro}.

Indeed, the importance of the stabilizers $\stab_i\in\Stabs$ is their role in error detection and correction. 
A stabilizer code is ``initialized'' in a state that is a simultaneous eigenstate of all $\nphys-\nlog$ stabilizer generators in $\Stabs$, 
where the bitstring $s_* = s_{*,1} \cdots s_{*,\nphys-\nlog}$ prescribes the ``default'' eigenvalue, 
so that $\stab_i \ket{s_*} = (-1)^{s_{*,i}} \ket{s_*}$. 
Simultaneously, a $\nlog$-qubit logical state  $\rho_L$ is encoded in the logical subspace, 
in that the logical Pauli operators $X_{\text{L}}^{(i)}$ and $Z_{\text{L}}^{(i)}$ act on the encoded state in the same manner that 
the operators $X^{\,}_i$ and $Z^{\,}_i$ would act on an unencoded $\nlog$-qubit state $\rho_L$. 
The utility of the stabilizer representation is that  generic errors result in a nontrivial \emph{syndrome} $s \neq s_*$---i.e., 
a deviation in the stabilizer eigenvalues from the default $s_*$. 
Measuring the stabilizer generators $\Stabs$ results in a string $\smeas$, which is then compared to $s_*$.
The set of stabilizers that differ from the default value is indicated by the bitstring
\begin{equation}
  \label{eq_s_err}
  s_{\text{err}} = \smeas + s_* \, ,~~
\end{equation}
which is called the error syndrome, and is used to determine what error occurred. 
In general, multiple error processes give rise to the same syndrome,
with the distinct errors being related to one another by logical operators (which commute with the stabilizers). 
In practice, one associates each nontrivial syndrome $s \neq s_*$ with a particular error process $E_s$, 
so that when $s \neq s_*$ is observed, one applies a correction $E_s^{-1}$ to undo the error. 
The error $E^{\,}_s$ associated with a syndrome $s$ is said to be ``correctable,'' 
while other error processes that give rise to the same syndrome $s$ are said to be ``uncorrectable.''

\subsection{The Clifford frame}
\label{sec:clifford frame}

In the ``physical frame,'' the stabilizers and logical operators~\eqref{eq:bare logicals} take the form of 
Pauli operators~\eqref{eq:Pauli set} with high weight. 
However, one can define a Clifford change of basis $\CFrame$ in which the factorization of the system 
into logical subsystem $L$ and syndrome subsystem $S$ is manifest~\cite{combes2017logicalrandomizedbenchmarking,caesura2021non}.
We call this new basis the ``Clifford frame.'' The description and analysis of quantum error correction is more straightforward in this basis, 
since the $\nlog$ logical operators and $\nphys-\nlog$ stabilizers act on independent, individual qubits. 
In general, this assignment of qubits is arbitrary; for concreteness---and without loss of generality---we define the Clifford frame such that 
the first $\nlog$ qubits correspond to the logical subspace $L$~\eqref{eq:bare logicals}, 
while the remaining $\nphys - \nlog$ registers correspond to the stabilizer subspace $S$.
The Clifford unitary $\CFrame$ that realizes this change of basis 
acts on the stabilizers and logical operators~\eqref{eq:bare logicals} as
\begin{equation}
\label{eq:CliffordFrame}
    \CFrame^\dagger X_{\text{L}}^{(i)} \CFrame ~,~ \CFrame^\dagger Z_{\text{L}}^{(i)} \CFrame ~~ \mapsto ~~ X^{\vpp}_{i} ~,~Z^{\vpp}_i 
    ~~~~\text{and}~~~~
    \CFrame^\dagger \, \stab^{\vpp}_j \, \CFrame ~~ \mapsto ~~ Z^{\vpp}_{j+k} \, ,~~
\end{equation}
so that the logical operators~\eqref{eq:bare logicals} realize the Pauli algebra on $L$, 
and the stabilizers $\stab$ map to Pauli $Z$ operators (see also Diagram~\ref{circ_qec_step_perfect}).
We comment that any Clifford unitary 
consistent with Eq.~\ref{eq:CliffordFrame} defines a suitable Clifford frame.

Additionally, the unitary $\CFrame$ is related to the \emph{encoding} of a target logical state $\sigma_L$.
The stabilizer code is initialized by preparing an eigenstate of the stabilizers $\stab_i$ 
with eigenvalues $s_*$ and encoding a target logical state $\sigma_L$. 
If one prepares $\sigma_L$ on the first $\nlog$ physical qubits (with all other qubits in the state $\ket{0}$),
then $\CFrame^\dagger$~\eqref{eq:CliffordFrame} encodes $\sigma_L$ in the logical subspace of the stabilizer code. While the resulting state is generally complicated in the physical frame, in the Clifford frame~\eqref{eq:CliffordFrame}, the encoded state is simply
\begin{equation}
    \label{eq:encoded state Clifford frame}
    \rho_{\text{enc}} =  \sigma_L \otimes \ZProj{s^{\vpp}_*}  = \sigma_L \otimes \bigotimes_{i=1}^{\nphys-\nlog} \KBop{s_{*,i}}{s_{*,i}} \, , ~~
\end{equation}
where $\ZProj{s} = \KBop{s}{s}$ denotes a computational-basis projector, 
$s_*$ is the ``default'' configuration of the $\nphys-\nlog$ stabilizers, 
and $s_{*,i}$ the default eigenvalue of the stabilizer $\stab_i$. 

\section{Model of an ideal QEC cycle}
\label{sec_prelim_noiseless}

Here we specify the particular model of an ideal (i.e., noiseless) QEC cycle in an arbitrary stabilizer code.
The model consists of four steps, which are depicted in Diagram~\ref{circ_qec_step_perfect}. In the remainder of this section, 
we describe those four steps in detail. 

The QEC cycle is modeled in the Clifford frame introduced in Sec.~\ref{sec:clifford frame}. 
In addition to the logical register $L$ and syndrome register $S$---which contain the qubits that define the code---we 
include ``ancilla'' qubits in a register $A$, and 
an ``outcome'' register $O$. 
The ancilla qubits are often used in fault-tolerant implementations of QEC cycles, and the outcome register records the outcomes of measurements performed as part of the QEC cycle.
For convenience, we model these measurements with a Clifford unitary $C$ that realizes a ``Stinespring representation'' of the measurements
of $S$ on the dilated Hilbert space associated with $S$, $A$, and $O$~\cite{Stinespring, KrausMeas1969, KrausBook, SpeedLimit, DiegoMeasOverview}.
Because the system $O$ effectively stores classical information, we model $O$
as a \emph{dephased} quantum register, 
meaning that the state of $O$ is always a mixture of computational-basis states, 
which may be correlated with the states of other subsystems. 
This convention simplifies the treatment of Pauli-stochastic noise in, e.g., Sec.~\ref{sec_Pauli_stochastic}.

The model of an ideal QEC cycle that we consider is captured by the following diagram, 
which should be read sequentially from left to right.
\begin{equation}
  \begin{quantikz}
    \lstick{$L$} & \qw  & \qw & \qw & \gate{\text{EC}} & \qw \rstick{} \\
    \lstick{$S$} & \gate[3]{C} &  \qw & \gate{\text{Reset}} & \qw & \qw \rstick{$\ZProj{s_*}$}\\ 
    \lstick{$A$} & \ghost{C} &  \qw & \qw & \qw & \qw \rstick{Discard}\\
    \lstick{$O$} & \ghost{C} & 
    \gate{\includegraphics[width=0.6cm]{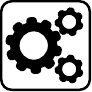}} & \phase\qw\vcw{-2} & \phase\qw\vcw{-3} & \qw\rstick{Discard}
    \end{quantikz}
    \label{circ_qec_step_perfect} 
\end{equation}
The four steps of the QEC cycle correspond to the four boxed operations above: 
\eqref{qec_step_C} syndrome information readout, corresponding to the Clifford $C$ \eqref{circ_cliff_readout};
\eqref{qec_step_proc} syndrome decoding, corresponding to the gearbox icon;
\eqref{qec_step_reset} syndrome reset, denoted ``Reset'';
\eqref{qec_step_EC} logical recovery, denoted ``EC'' for ``error correction.'' 
We now describe each of these steps in detail.


\begin{enumerate}
\item \textbf{Syndrome information readout.} \label{qec_step_C}
  The Clifford $C$~\eqref{circ_qec_step_perfect} implements the measurements used to obtain the information necessary to infer the error syndrome $s_{\text{err}}$ \eqref{eq_s_err}.
  The Clifford $C$ is specified as part of the definition of a QEC cycle, and has a particular action on input product states, as shown in the following diagram:
    \begin{equation}
      \begin{quantikz}
        &\lstick{$L$} \quad \quad &\lstick{$\rho_L$} & \qw & \qw \rstick{$\rho_L$} \\
        &\lstick{$S$} \quad \quad &\lstick{$\ZProj{\sinput}$} & \gate[3]{C} & \qw \rstick{$\ZProj{\sinput}$} \\ 
        &\lstick{$A$} \quad \quad &\lstick{$\ZProj{a}$} & \qw & \qw\rstick{Discard} \\
        &\lstick{$O$} \quad \quad &\lstick{$\ZProj{0}$} & \ghost{C}& \qw\rstick{$\ZProj{\theta} = \ZProj{E(\sinput)}$} 
        \end{quantikz}
        \label{circ_cliff_readout} 
    \end{equation}
  where $\sinput$ is the bitstring that encodes the input syndrome state, 
  $C$ preserves the input states $\rho_L$ and $\ZProj{\sinput}$ of the systems $L$ and $S$, respectively, and
  $\ZProj{a}$ and $\ZProj{0}$ are arbitrary computational-basis states of the additional systems $A$ and $O$, respectively. 
  We note that $C$ is always applied to product input states of the form above, and the resulting state also factorizes over the four subsystems.

  The output state of system $O$ is $\ZProj{\theta}$, where the bitstring $\theta$ is determined by $\sinput$ according to the specifics of the Clifford $C$.
  To model this, we introduce the ``syndrome encoding function'' $E$, which is part of the specification of $C$, where $\theta = E(\sinput)$.
  The name of the function $E$ is motivated by the fact that action of $C$ typically encodes the input syndrome $\sinput$ in a classical error correcting code.
  We emphasize that the Clifford $C$ implements the entire fault-tolerant syndrome extraction needed for the QEC cycle, 
  which often includes multiple rounds of syndrome measurements~\cite{gottesman2016surviving}.
  In such situations, after the action of $C$, the register $O$ contains the outcomes from all measurement rounds, 
  which are used jointly to extract the error syndrome $s_{\text{err}}$.

\item \textbf{Syndrome decoding.} \label{qec_step_proc} 
  The bitstring $\theta$ containing measurement outcomes obtained in Step~\ref{qec_step_C} is classically processed to extract the error syndrome $s_{\text{err}}$.
  This step only involves the register $O$, and is represented by the box with the gear icon in Diagram~\ref{circ_qec_step_perfect}.
  We assume it is performed perfectly (i.e., without errors). 
  First, the measured syndrome is obtained according to $\smeas = D(\theta)$, 
  where the function $D$ is specified as part of the definition of the QEC cycle.
  Typically, $D$ is the decoding function of a classical error-correcting code, 
  so that $D\big(E(s) \big) = s$ for any bitstring $s$ of system $S$.
  Accordingly, we refer to $D$ as the ``syndrome decoding function.''
  Finally, the error syndrome is obtained via $s_{\text{err}} = \smeas + s_*$, 
  where $s_*$ is the 
  default configuration of $S$.
  \item \textbf{Syndrome reset.} \label{qec_step_reset}
  Based on the inferred error syndrome $s_{\text{err}}$ from Step~\ref{qec_step_proc}, 
  the syndrome subsystem $S$ is returned to the default configuration $s_*$. 
  The ``Reset'' operation in Diagram~\ref{circ_qec_step_perfect} is realized by
  a Pauli superoperator $\PauliChan*{\text{Reset}}$ applied to $S$, where
  \begin{equation}
    \label{eq_reset_op}
    \PauliChan{\text{Reset}}(\ZProj{\smeas}) = \ZProj{\smeas+ s_{\text{err}}} = \ZProj{s_*} \, , ~~
  \end{equation}
  so that $S$ is reset to $s_*$. 
  Concretely, this operation is implemented in the Clifford frame by 
  applying Pauli $X$ to the $j$th qubit in $S$
  if the $j$th bit of $s_{\text{err}}$ is one.

\item \textbf{Logical recovery.} \label{qec_step_EC}
  A Pauli superoperator determined by the error syndrome $s_{\text{err}}$ is applied to the logical system $L$, 
  indicated by the box labeled ``EC'' in Diagram~\ref{circ_qec_step_perfect}.
  Concretely, this operation acts as
  \begin{equation}
  \label{eq:EC step}
      \operatorname{EC} (\rho_L \otimes \ZProj{s} \otimes \ZProj{\theta}) = \Cor_{s_\text{err}} (\rho_L) \otimes \ZProj{s} \otimes \ZProj{\theta}\, ,~~  
  \end{equation}
  where $\Cor$ is a function (typically called the decoder) that maps syndrome bitstrings $s_{\text{err}}$ to the appropriate correction Pauli superoperator 
  $\Cor_{s_{\text{err}}}$. 
  We note that the error-correction operation EC acts nontrivially only on the logical register $L$, 
  and only when the syndrome is nontrivial; when $\smeas = s_*$ is trivial, we have that $\Cor_{s_*} (\rho_L) = \rho_L$. 

\end{enumerate}

To summarize, we require that a QEC cycle can be modeled as shown in Diagram~\ref{circ_qec_step_perfect} and discussed 
above.
The specification of a QEC cycle consists of the Clifford unitary $C$---which implicitly includes the function $E$---as well as the function $D$ and the decoder $\mathcal{C}$.
The sizes of the systems $L$, $S$, $A$, and $O$ are also implicitly specified as part of the QEC cycle.

Strictly speaking, there 
exist implementations of QEC cycles in stabilizer codes that do \emph{not} fit into the
foregoing prescription. 
One example is flagged fault-tolerant schemes \cite{PhysRevLett.121.050502}, 
in which intermediate measurement outcomes may affect the choice of subsequent Clifford operations.
We expect that such QEC cycles can also be treated using our formalism, but we leave such an analysis for future work.
As another example, some stabilizer QEC schemes, like surface codes, use the measurement outcomes from previous QEC cycles in order to determine correction operations \cite{PhysRevA.86.032324}.
Our definition of a QEC cycles is not compatible with this case, and extending our results to cover it is an important aspect of future work.

We also note that the syndrome reset operation in Diagram~\ref{circ_qec_step_perfect} need not be physically implemented, and instead can be frame-tracked by using knowledge of the reset operation to inform the next QEC cycle.
Similarly, as long as there are no non-Clifford logical operations the logical recovery operation in Diagram~\ref{circ_qec_step_perfect} can be deferred until to the end of the experiment by keeping track of an appropriate logical Pauli frame~\cite{Knill2005,Chamberland_2018}.

\section{Modeling noisy QEC experiments}
\label{sec_noisy_qec}

Here we consider QEC experiments involving state preparation, $\ncyc$ consecutive QEC cycles, and logical measurement 
in the presence of Pauli-stochastic noise.
In particular, we define an effective Markovian model of a family of QEC experiments, 
which acts on the logical register $L$ alone.
The model depends on the particulars of the experiment, which we consider herein.

The remainder of this section is organized as follows.
In Sec.~\ref{sec_Pauli_stochastic}, we define the notion of Pauli-stochastic noise in general; 
precise definitions are given in the specific contexts of 
state preparation, individual QEC cycles, and logical measurements in the respective subsections. 
In particular, we describe noisy state preparation in Sec.~\ref{sec_Pauli_noise_state_prep}, 
noisy QEC cycles in Sec.~\ref{sec_noisy_qec_cycles},
and noisy measurements in Sec.~\ref{sec_Pauli_noise_meas}. 
In Sec.~\ref{sec_approx_model}, we define QEC experiments and 
the total logical error channel $\Lambda_{\text{tot}}$~\eqref{eq_lambda_total} associated 
with state preparation, $\ncyc$ QEC cycles, and logical measurement. 
We also define the notion of a Markovian model on $L$ alone that approximates the total logical error channel 
$\Lambda_{\text{tot}}$ (see Def.~\ref{def_logical_markov_model}).
Finally, in Sec.~\ref{sec_Pauli_eigs}, we show that such a model $\Phi_{\text{tot}}$ exists 
if the Pauli eigenvalues~\eqref{eq_qec_lambda_total_eig} of the total logical error channel $\Lambda_{\text{tot}}$
decay exponentially with the total number of QEC cycles $\ncyc$.

\subsection{Pauli-stochastic noise}
\label{sec_Pauli_stochastic}

Our analysis assumes that all errors can be modeled as Pauli-stochastic noise. Loosely speaking, 
given a target quantum operation $\Phi$, 
the noisy implementation $\Phi'$ realizes Pauli-stochastic noise if $\Phi'$ can be written as the composition of the ideal implementation $\Phi$ 
and a Pauli channel $\ErrChan$~\eqref{eq:general Pauli channel}. 
We allow the Pauli noise channel $\ErrChan$ to act before or after the ideal implementation $\Phi$.  
In the case that $\ErrChan$ does not realize a Pauli channel at the physical level, 
Pauli stochasticity may be enforced via Pauli twirling or Pauli randomization. 
We explicitly define Pauli-stochastic noise in the contexts of state preparation, QEC cycles, and 
measurements in Secs.~\ref{sec_Pauli_noise_state_prep}, \ref{sec_noisy_qec_cycles}, and \ref{sec_Pauli_noise_meas}, respectively. 
Importantly, the Pauli-stochastic noise channels associated with each process depends only on the 
details of that process, and is independent of any aspect of the other two processes, including the associated noise channels.

\subsection{Noisy state preparation}
\label{sec_Pauli_noise_state_prep}

We define Pauli-stochastic noise for state preparation to mean that the noisy initial state $\rho_0$ 
can be expressed as a Pauli channel~\eqref{eq:general Pauli channel} $\ErrChan_{\text{prep}}$ applied to the 
ideal encoded state $\rho_{\text{enc}} = \sigma_L \otimes \ZProj{s_*}$~\eqref{eq:encoded state Clifford frame} on $L$ and $S$. 
Suppose that ideal state preparation is implemented by a process ${\cal U}$ on $L$ and $S$, 
where $\cal{U}$ outputs 
the nominal encoded state $\rho_{\text{enc}} = \sigma_L \otimes \ZProj{s_*}$ regardless of the input state $\sigma_L$. 
The noisy implementation of state preparation ${\cal U}'$ acts as

\begin{equation}
  \label{eq_circ_state_prep}
  \begin{quantikz}
    &\lstick{$L$} & \qw & \gate[2]{{\cal U}_{\vphantom{prep}}} &  \gate[2]{\ErrChan_{\text{prep}}} & \qw \\
    &\lstick{$S$} & \qw & \ghost{{\cal U}_{\vphantom{prep}}} & \ghost{\ErrChan_{\text{prep}}} & \qw  
    \end{quantikz}
\end{equation}
We show in App.~\ref{app_noisy_state_prep} that, if $\ErrChan_{\text{prep}}$ is a Pauli channel then the resulting
state $\rho_0$ on $L$ and $S$
can be expressed in terms of a probability distribution $\gamma_{s_0}^{(\text{prep})}$ over the true bitstring $s_0$ of system $S$ 
and a conditional Pauli channel $\Lambda^{(\text{prep})}_{s_0}$ that depends on $s_0$, as follows
\begin{equation}
 \label{eq_noisy_prep}
  \rho^{\vpp}_0 = \sum_{s_0}
  \gamma_{s_0}^{(\text{prep})} \,  \Lambda_{s_0}^{(\text{prep})} (\sigma_L) \otimes \ZProj{s_0} \, ,~~
\end{equation}
where $\gamma_{s_0}^{(\text{prep})}$ and $\Lambda^{(\text{prep})}_{s_0}$ are 
related to $\ErrChan_{\text{prep}}$ via Eq.~\ref{eq_state_prep_probs} of App.~\ref{app_noisy_state_prep}.

We comment that, if $\sigma_L$ is a stabilizer state, then one can enforce Pauli-stochastic noise by performing Pauli randomization.
Specifically, one applies a uniformly random element of the stabilizer group $\Stabs$ of the ideal encoded state $\rho_{\text{enc}}$~\eqref{eq:encoded state Clifford frame} 
after the noisy unitary $\mathcal{U}' = \ErrChan_{\text{prep}} \circ \mathcal{U}$.

\subsection{Noisy QEC cycles}
\label{sec_noisy_qec_cycles}

We define Pauli-stochastic noise for a QEC cycle to be noise that can be modeled with a Pauli channel $\ErrChan$
that acts immediately after the Clifford $C$
in Diagram~\ref{circ_qec_step_perfect} (which depicts a noiseless QEC cycle). 
In general, $\ErrChan$ acts on all four registers $L$, $S$, $A$, and $O$, and the noisy implementation of a QEC cycle is captured by the following diagram:
\begin{equation}
  \begin{quantikz}
    &\lstick{$L$} ~ & \qw & \gate[4]{\ErrChan} &\qw & \qw & \gate{\text{EC}} & \qw \\
    &\lstick{$S$} ~ & \gate[3]{C} & \ghost{\ErrChan}& \qw & \gate{\text{Reset}} & \qw & \qw \\ 
    &\lstick{$A$} ~ & \ghost{C} & \ghost{\ErrChan} & \qw & \qw & \qw & \qw \rstick{Discard}\\
    &\lstick{$O$} ~ 
    & \ghost{\ErrChan}& \ghost{C} &\gate{\includegraphics[width=0.6cm]{sideproc.pdf}} &  \phase\qw\vcw{-2} & \phase\qw\vcw{-3} & \qw\rstick{Discard}
    \end{quantikz}
    \label{circ_qec_step_errors}
\end{equation}
The Pauli channel $\ErrChan$ captures errors associated with syndrome extraction (Step~\ref{qec_step_C}).
As a reminder, the other QEC steps can all be implemented in classical side processing, 
so we assume they are error free (see Sec.~\ref{sec_prelim_noiseless}).
Additionally, any errors in the initialization of the registers $A$ and $O$ can be absorbed into $\ErrChan$. 
Finally, we note that even if the error channel $\ErrChan$ 
is not a Pauli channel~\eqref{eq:general Pauli channel} at the physical level,
it can be converted into an effective Pauli channel via Pauli twirling of $C$.

We now consider the update to the joint state of $L$ and $S$ corresponding to Diagram~\ref{circ_qec_step_errors}. 
For simplicity, we consider a product input state on $L$ and $S$ of the form
\begin{equation}
\label{eq:QEC input state}
  \rho_{\text{in}} =\rho_L \otimes \ZProj{\sinput} \, , ~~
\end{equation}
and we note that generic input states to a QEC cycle are convex linear combinations of such product states. 
As a reminder, the QEC cycle involves an ancillary system $A$ prepared in the default state $\ket{a}$ 
and a (dephased) outcome register $O$ prepared in the default state $\ket{0}$.

Step~\ref{qec_step_C} of the QEC cycle is the noisy application of the Clifford $C$, 
which we denote by $C' = \ErrChan \circ C$. 
Next, the ancilla system $A$ is discarded; the resulting state on $L$, $S$, and $O$ is
\begin{equation}
    C'\left(\rho_L \otimes \ZProj{\sinput} \otimes \ZProj{s_*} \right)
    = \sum_{\ell', s', \theta'} \overline{p}^{\vpp}_{\ell', s', \theta'} 
    ~ \PauliChan{\ell'} (\rho_L) 
  \otimes \ZProj{\sinput+ s'} 
  \otimes \ZProj{E(\sinput) + \theta'} \, ,~ ~
  \label{eq_after_errchan} 
  \end{equation}
  for some probabilities $\overline{p}^{\vpp}_{\ell', s', \theta'}$, where 
$s'$ and $\theta'$ are bitstrings of the systems $S$ and $O$, respectively. 
A derivation is given in App.~\ref{app_QEC_update_C}, including an expression for $\overline{p}^{\vpp}_{\ell', s', \theta'}$ in terms of $\ErrChan$ in Eq.~\ref{eq:QEC intermediate probabilities}.

We now account for Steps~\ref{qec_step_proc}, \ref{qec_step_reset}, and \ref{qec_step_EC}, concluding the QEC cycle.
With the cycle complete, we now discard (i.e., trace out and reset) the outcome register $O$.
The joint state of the subsystems $L$ and $S$ output 
by the QEC cycle is
\begin{align}
    \rho_{\text{out}} &=  \sum_{\ell', s', \theta'} 
    \overline{p}^{\vpp}_{\ell', s', \theta'} ~
    \Cor_{\serr} \circ \PauliChan{\ell'} (\rho_L) 
    \otimes 
    \ZProj{\sinput + s' + \serr} 
     \label{eq:rho out naive} \, ,~~
\end{align}
where $\serr = D(E(\sinput) + \theta') + s_*$ depends on $\sinput$ and $\theta'$
and the decoder $\Cor_{\serr}$ is defined in Eq.~\ref{eq:EC step}. 
We simplify the expression above by introducing a new summation variable $\soutput = \sinput + \serr + s'$ 
and a probability distribution $p^{\,}_{\ell,\soutput}$, finding
\begin{align}
  \rho_{\text{out}} &= \sum\limits_{\ell,\soutput} p^{\vpp}_{\ell,\soutput} \PauliChan{\ell} (\rho_L) \otimes \ZProj{\soutput} \, ,~~
  \label{eq:rho out nice}
\end{align}
where we combined the correction channel $\Cor_{\serr}$ and logical error channel $\PauliChan*{\ell'}$ into 
a single logical Pauli superoperator $\PauliChan*{\ell}$. The probabilities $p^{\,}_{\ell,\soutput}$
encode the correlations between  $L$ and $S$ that result from the QEC cycle. 
A derivation appears in App.~\ref{app_QEC_update_2}, 
with $p^{\,}_{\ell,\soutput}$ defined in Eq.~\ref{eq:QEC nice probabilities}. 
We stress that $\soutput$ and the probabilities $p^{\,}_{\ell,\soutput}$ depend implicitly on $\sinput$.

We then observe that the transformation of the joint state of $L$ and $S$ from $\rho_{\text{in}}$~\eqref{eq:QEC input state}
to $\rho_{\text{out}}$~\eqref{eq:rho out nice} is independent of $\rho_L$. 
As a result, we recover a convenient expression for the state $\rho^{\,}_{\cycind}$ output by the $\cycind$th QEC cycle
in terms of a classical Markov process $\gamma^{\,}_{\soutput, \sinput}$ acting on the system $S$ 
and a family of Pauli channels $\Lambda^{\,}_{\soutput,\sinput}$ acting on the system $L$, 
conditioned on the bitstrings $\sinput$ and $\soutput$ on $S$. 
This formulation
is formalized in Prop.~\ref{prop_qec_final} below.

\begin{prop}
  \label{prop_qec_final}
  Consider a QEC cycle as described in Sec.~\ref{sec_prelim_noiseless} with Pauli-stochastic noise as described above.
  There exists a Markov process $\gamma^{\,}_{\soutput,\sinput}$ on the system $S$ 
  and logical Pauli channels $\{\Lambda^{\,}_{\soutput,\sinput}\}$ such that, 
  if the joint state on $L$ and $S$ prior to a QEC cycle is $\rho_{\mathrm{in}} = \rho_L \otimes \ZProj{\sinput}$~\eqref{eq:QEC input state},
  the joint state after the QEC cycle is given by
  \begin{align}
    \label{eq_qec_gamma_lambda}
    \rho_{out} = \sum\limits_{\soutput} \gamma^{\vpp}_{\soutput,\sinput} \, 
    \Lambda^{\vpp}_{\soutput,\sinput} ( \rho_L) \otimes \ZProj{\soutput} \, .~~
  \end{align}
  Additionally, the joint state on $L$ and $S$ that results from $\cycind$ (possibly distinct) QEC cycles 
  applied to the noisy initial state $\rho_0$~\eqref{eq_noisy_prep} can be written
  \begin{align}
    \label{eq_qec_final}
    \rho^{\vpp}_{\cycind} = \sum\limits_{s_0,\dots,s_{\cycind}} 
    \gamma^{(\cycind)}_{s_{\cycind},s_{\cycind-1}} \dots \gamma^{(1)}_{s_{1},s_{0}}
     \gamma^{(\mathrm{prep})}_{s_0} ~
    \Lambda^{(\cycind)}_{s_{\cycind},s_{\cycind-1}} \circ \dots \circ \Lambda^{(1)}_{s_{1},s_{0}} 
    \circ \Lambda^{(\mathrm{prep})}_{s_0} (\sigma_L) \otimes \ZProj{s_{\cycind}} \, , ~~
  \end{align}
  where explicit formulae for $\gamma^{(\cycind)}_{\soutput,\sinput}$ and $\Lambda^{(\cycind)}_{\soutput,\sinput}$ 
  appear in Eqs.~\ref{eq:QEC gamma def} and \ref{eq:QEC Lambda def} of App.~\ref{app_gamma_and_gamma}. 
\end{prop}

The proof of Prop.~\ref{prop_qec_final} appears in App.~\ref{app_gamma_and_gamma}. 
Going forward, we assume both $\gamma^{(\cycind)}_{\soutput,\sinput}$ and $\Lambda^{(\cycind)}_{\soutput,\sinput}$ to be identical 
for all QEC cycles $\cycind$.

\subsection{Logical measurement}
\label{sec_Pauli_noise_meas}

The final stage of any QEC experiment is the extraction of logical information via measurement, 
which we model as shown in Diagram~\ref{circ_qec_meas_full}.
Extraction of the logical information requires the concurrent measurement of the syndrome in order to correct residual logical errors.
We model this with a Clifford $C_{\text{meas}}$, much like the Clifford $C$~\eqref{circ_cliff_readout} that defines the QEC cycle.
As in Step~\ref{qec_step_C} of a QEC cycle, the Clifford $C_{\text{meas}}'$ acts on the systems $L,S,A,$ and $O$ 
in a way that writes the relevant measurement data to a dephased quantum register $O$, 
which is then classally processed to recover the measurement outcome. 
We define Pauli-stochastic noise
for logical measurements as noise that can be modeled by a Pauli channel $\ErrChan_{\text{meas}}$ acting on all registers immediately 
after $C_{\text{meas}}$, as depicted in the following diagram.

\begin{equation}
  \begin{quantikz}
    \lstick{$L$} & \gate[4]{C_{\text{meas}}} & \gate[4]{\ErrChan_{\text{meas}}} &\qw  & \qw\rstick{Discard} \\
    \lstick{$S$} & \ghost{C_{\text{meas}}} & \ghost{\ErrChan_{\text{meas}}} &\qw  & \qw \rstick{Discard} \\ 
    \lstick{$A$} & \ghost{C_{\text{meas}}} & \ghost{\ErrChan_{\text{meas}}} &\qw & \qw \rstick{Discard}\\
    \lstick{$O$} & \ghost{C_{\text{meas}}}& \ghost{\ErrChan_{\text{meas}}} & \gate{\includegraphics[width=0.6cm]{sideproc.pdf}}  & \qw\rstick{Outcomes}
    \end{quantikz}
    \label{circ_qec_meas_full}
\end{equation}

For convenience, 
we restrict to destructive logical measurements $\POVM$ in the computational basis,
which output a bitstring indicating the observed outcome. 
Consequently, the information on $L$ and $S$ is discarded following the final measurement. 
Measurements in other bases are realized by absorbing additional logical Clifford gates into $C_{\text{meas}}$.

It is convenient to model a measurement with Pauli-stochastic noise in terms of a family of logical channels $\Lambda_{\soutput}^{(\text{meas})}$,
which depend on the bitstring $\soutput$ of system $S$. Concretely, 
if the nominal logical measurement $\POVM$ has POVM elements labeled $\POVM_{\elt}$
and the joint state of $L$ and $S$ prior to measurement is 
$\rho_L \otimes \ZProj{\soutput}$, 
then the probability of outcome $\elt$ is
\begin{equation}
  \operatorname{Prob}(\elt) = 
  \tr{\POVM_{\elt} \, \Lambda_{\soutput}^{(\text{meas})}(\rho_L) } \, , ~~
  \label{eq_meas_err_state}
\end{equation}
which extends straightforwardly to mixtures of states of the form above.

In App.~\ref{app_logical_meas}, we compute $\operatorname{Prob}(\elt)$ on the left-hand side of Eq.~\ref{eq_meas_err_state} 
using the noisy channel $C_{\text{meas}}' = \ErrChan^{\vpp}_{\text{meas}} \circ C_{\text{meas}}^{\vpp}$ depicted above. 
We then extract the logical channel on the right-hand side of Eq.~\ref{eq_meas_err_state},
which captures the noisy measurement. In particular, we define
\begin{equation}
  \label{eq_Lambda_meas}
  \Lambda^{(\text{meas})}_{\soutput} (\rho_L) = \sum\limits_{\Pauli*{\ell} \in \PauliSet*{L}} 
  p^{(\text{meas})}_{\ell,\soutput} ~\PauliChan{\ell} ( \rho_L) \, ,~~
\end{equation}
where the probabilities $p^{(\text{meas})}_{\ell,\soutput}$ are defined in Eq.~\ref{eq_meas_Lambda_probs} of App.~\ref{app_logical_meas} and depend on $\ErrChan_{\text{meas}}$.

\subsection{Approximate logical Markovian model of QEC experiments}
\label{sec_approx_model}

We now combine the preceding results to identify 
(\emph{i}) a total logical error channel $\Lambda_{\text{tot}}^{(K)}$ that models an entire QEC experiment and
(\emph{ii}) the notion of a Markovian model on $L$ that approximates the true model. We first define a QEC experiment.
\begin{defn}[QEC experiment]
  A \emph{QEC experiment} that uses a  QEC cycle (see Sec.~\ref{sec_noisy_qec_cycles}) is an experiment that consists 
  of a state preparation 
  (see Sec.~\ref{sec_Pauli_noise_state_prep}),
  followed by $\ncyc \geq 1$ repetitions of the QEC cycle, followed by
   a measurement 
  (see Sec.~\ref{sec_Pauli_noise_meas}). 
  \label{defn_QEC_experiment}
\end{defn}
We are interested in constructing a description of the logical errors that occur in a QEC experiment with Pauli stochastic noise, as defined in Secs.~\ref{sec_Pauli_noise_state_prep}-\ref{sec_Pauli_noise_meas}.
Importantly, we demand that this description involves only the $L$ subsystem.
Accordingly, we define
\begin{defn}[Logical model of a QEC experiment]
  A \emph{logical model of a QEC experiment} with $\ncyc$ cycles consists of a nominal logical state $\sigma_L$, 
  a total logical error channel $\Lambda_{tot}^{(\ncyc)}$, and a nominal logical POVM $\POVM$. The logical model predicts 
  the probabilities of the measurement outcomes of a QEC experiment according to
  \begin{equation}
    \label{eq_meas_prob_tot}
    \operatorname{Prob}(\elt) = \trace \big( \POVM_{\elt} \, \Lambda^{(\ncyc)}_{\mathrm{tot}}(\sigma_L) \big) \, ,~
  \end{equation}
  where the POVM element $\POVM_\elt$ is asssociated with the $\elt$th measurement outcome (see Sec.~\ref{sec_Pauli_noise_meas}).
  We occasionally overload notation by referring to the probabilities $\operatorname{Prob}(\elt)$ as the logical model.
  \label{def_total_model}
\end{defn}
When the QEC experiment has Pauli-stochastic noise, the true distribution over measurement outcomes admits a logical model, as shown in the following corollary to Prop.~\ref{prop_qec_final}.
\begin{cor}[Total logical error channel]
  Given a QEC experiment as specified in Def.~\ref{defn_QEC_experiment} with Pauli-stochastic noise as defined in Secs.~\ref{sec_Pauli_stochastic}--\ref{sec_Pauli_noise_meas},
  the \emph{total logical error channel}  is
  \begin{equation}
    \Lambda^{(\ncyc)}_{\mathrm{tot}} = \hspace{-2mm} \sum_{s_0,\dots,s_\ncyc} \hspace{-0.7mm}
    \gamma^{\vpp}_{s_{\ncyc},s_{\ncyc-1}} \hspace{-0.8mm} \cdots \gamma^{\vpp}_{s_1,s_0} \gamma^{\mathrm{(prep)}}_{s_0}~
    \Lambda^{\mathrm{(meas)}}_{s_\ncyc} \hspace{-0.3mm} \circ  \Lambda^{\vpp}_{s_{\ncyc},s_{\ncyc-1}} \hspace{-0.8mm} \circ \cdots \circ \Lambda^{\vpp}_{s_1,s_0}\hspace{-0.5mm}
    \circ  \Lambda^{\mathrm{(prep)}}_{s_0} 
    \label{eq_lambda_total} \, ,~~
  \end{equation}
  where $\gamma^{\mathrm{(prep)}}_{s_0}$ and $\Lambda^{\mathrm{(prep)}}_{s_0}$ depend on the preparation procedure,
  $\Lambda^{\mathrm{(meas)}}_{s_\ncyc}$ depends on 
  the measurement procedure,
  and $\gamma^{\vpp}_{s_{\cycind},s_{\cycind-1}}$ and $\Lambda^{\vpp}_{s_{\cycind},s_{\cycind-1}}$ depend on
  $\ErrChan$ (see Prop.~\ref{prop_qec_final}).  
  \label{cor_lambda_tot}
\end{cor}
\begin{proof}
  By Eq.~\ref{eq_qec_final} of Prop.~\ref{prop_qec_final}, the joint state of $L$ and $S$ 
  following the $\ncyc$th and final QEC cycle is given by $\sum_{s_{\ncyc}} \, \rho_{\ncyc} \otimes \ZProj{s_{\ncyc}}$.
  Applying Eq.~\ref{eq_meas_err_state} 
  leads to Eq.~\ref{eq_lambda_total}.
\end{proof}

It is useful to define a set of QEC experimental settings, which we use to define logical Markovian models.
We note that a QEC experiment (see Def.~\ref{defn_QEC_experiment}) is specified by a QEC experimental setting and the number of QEC cycles, $\ncyc > 0$.

\begin{defn}[Set of QEC experimental settings]
  Let $\mathsf{Prep}$ be a set of state-preparation procedures with Pauli-stochastic noise (see Sec.~\ref{sec_Pauli_noise_state_prep}), 
  each of which consists of a target logical state $\sigma_L$ and the associated error model $\ErrChan_{\mathrm{prep}}$ 
  (see Diagram~\ref{eq_circ_state_prep}). 
  Let $\mathsf{Meas}$ be a set of measurements with Pauli-stochastic noise (see Sec.~\ref{sec_Pauli_noise_meas}),
  each of which consists of a target logical POVM $\POVM$ and the associated error model $\ErrChan_{\mathrm{meas}}$ 
  (see Diagram~\ref{circ_qec_meas_full}).
  Then a \emph{set of QEC experimental settings} is given by a subset $\mathsf{Exp} \subseteq \mathsf{Prep} \times \mathsf{Meas}$,
  which includes all combinations of preparations and measurements of interest. 
  \label{defn_exp_settings}
\end{defn}

We comment that a total logical error channel $\Lambda^{(\ncyc)}_{\text{tot}}$~\eqref{eq_lambda_total} 
is defined for each experimental setting $(p,m) \in \mathsf{Exp}$ and for each $\ncyc \geq 1$. Going forward, we write $\Lambda^{(\ncyc)}_{p,m}$ in lieu 
of $\Lambda^{(\ncyc)}_{\text{tot}}$ to make the relation to the experimental setting explicit. We now define the notion of a Markovian model $\mathcal{M}$ on the logical subsystem $L$ for a set 
$\mathsf{Exp}$ of QEC experimental settings associated with a fixed QEC cycle (see Def.~\ref{defn_exp_settings}).

\begin{defn}[Logical Markovian model]
  Consider a set $\mathsf{Exp} \subseteq \mathsf{Prep} \times \mathsf{Meas}$ of experimental settings, as in Def.~\ref{defn_exp_settings}. 
  The \emph{logical Markovian model} for $\mathsf{Exp}$ is the triple
  \begin{equation}
    \label{eq_Markov_triple}
    \mathcal{M} = \left( \mathcal{F}^{\,}_{\mathrm{prep}} , \Phi^{\,}_{\mathrm{QEC}} , \mathcal{F}^{\,}_{\mathrm{meas}} \right) \, ,~~
  \end{equation}
  where $\Phi_{\mathrm{QEC}}$ is a logical superoperator that is independent of the settings in $\mathsf{Exp}$, 
  and we define the following families of logical superoperators
  \begin{align}
    \mathcal{F}^{\,}_{\mathrm{prep}} = \left\{ \Phi^{\,}_{\mathrm{prep},p} \right\}_{p \in \mathsf{Prep}}
    ~~~ \text{and} ~~~
    \mathcal{F}^{\,}_{\mathrm{meas}} = \left\{ \Phi^{\,}_{\mathrm{meas},m} \right\}_{m \in \mathsf{Meas}}
    \, .~~ \label{eq_Phi_families}
  \end{align}
  \label{def_logical_markov_model}
\end{defn}
Importantly, for any experimental setting $(p,m) \in \mathsf{Exp}$ (see Def.~\ref{defn_exp_settings}) and a fixed 
number of QEC cycles $\ncyc \geq 1$, the Markov model $\mathcal{M}$ gives rise to an approximation of
the model $\operatorname{Prob}(\elt)$~\eqref{eq_meas_prob_tot} of the corresponding QEC experiment 
(see Defs.~\ref{defn_QEC_experiment} and \ref{def_total_model}).
\begin{defn}[Approximate probability distribution]
  Let $\mathcal{M}$ be a logical Markovian model for a set of QEC experimental settings $\mathsf{Exp}$, 
  as in Def.~\ref{def_logical_markov_model}.
  For each experimental setting $(p,m) \in \mathsf{Exp}$ and any $\ncyc >0$, $\mathcal{M}$ 
  gives rise to the approximation
  \begin{equation}
    \label{eq_Markov_prob_requirement}
    \mu^{(\ncyc)}_{p,m} (\elt) = \trace \big(  \POVM_{\elt} \, \Phi^{\vpp}_{\mathrm{meas}} \circ \Phi_{\mathrm{QEC}}^{\ncyc} 
    \circ \Phi^{\vpp}_{\mathrm{prep}}(\sigma_L) \big) \, ,~~
  \end{equation}
  %
  %
  for the logical probability distribution $\operatorname{Prob}(\elt)$~\eqref{eq_meas_prob_tot} 
  that models any QEC experiment associated with the setting $(p,m)$ in $\mathsf{Exp}$ and involving $\ncyc$ QEC cycles (see Def.~\ref{def_total_model}). 
\label{def_approx_probs}
\end{defn}

We note that the superoperators in Def.~\ref{def_logical_markov_model} are \emph{not} required to be completely positive
nor trace preserving. Accordingly, $\mu^{(\ncyc)}_{p,m} (\elt)$~\eqref{eq_Markov_prob_requirement} does not necessarily realize
a probability distribution, and in general may lead to a poor approximation of $\operatorname{Prob}(\elt)$~\eqref{eq_Markov_prob_requirement}. 
We now define the more useful notion of an \emph{approximate} logical Markovian model, which we consider in the remainder of the text. 
We do so by imposing a bound on the deviation of the approximations $\mu^{(\ncyc)}_{p,m} (\elt)$ of the Markovian model $\mathcal{M}$ from the true values $\operatorname{Prob}(\elt)$.

\begin{defn}[$(G,\epsilon)$-approximate logical Markovian model]
  Consider a set of experimental settings $\mathsf{Exp}$ as in Def.~\ref{defn_exp_settings}.
  Let $G$ and $\epsilon$ be positive, setting-independent constants with $\epsilon < 1$.
  Then, a logical Markovian model $\mathcal{M}$ for $\mathsf{Exp}$ (see Def.~\ref{def_logical_markov_model}) 
  is a \emph{$(G,\epsilon)$-approximate logical Markovian model} for $\mathsf{Exp}$ if,
  for all $(p,m)$ in $\mathsf{Exp}$, all $\ncyc > 0$, and all $\POVM_{\elt}$, 
  \begin{equation}
    \Big\lvert\trace \big( \POVM_{\elt}  \, \Lambda^{(\ncyc)}_{p,m} (\sigma_L) \big) - \trace ( \POVM_{\elt} \, \Phi_{\mathrm{meas}}^{\vpp} \Phi_{\mathrm{QEC}}^{\ncyc} 
    \Phi^{\vpp}_{\mathrm{prep}}(\sigma_L) ) \Big\rvert \leq G \epsilon^{\ncyc} \, .~
    \label{eq_approx_markov_equiv}
  \end{equation}
  \label{def_markov_approx}
\end{defn}

\subsection{Pauli eigenvalues of the total logical error channel}
\label{sec_Pauli_eigs}

It is convenient to consider the total logical error channel $\Lambda^{(\ncyc)}_{p,m}$~\eqref{eq_lambda_total}
in terms of its action on logical Pauli operators $P \in \PauliSet*{L}$. 
Because $\Lambda^{(\ncyc)}_{p,m}$ is a Pauli channel~\eqref{eq:general Pauli channel} 
for all experimental settings in $\mathsf{Exp}$ (see Def.~\ref{def_logical_markov_model}) and all $\ncyc \geq 1$, 
it acts on every logical Pauli $P$ as
\begin{equation}
  \Lambda^{(\ncyc)}_{p,m} (P) = \lambda^{(\ncyc)}_{p,m,P} \, P  = \frac{1}{2^{\nlog}} \trace \left( P \, \Lambda^{(\ncyc)}_{p,m} (P) \right) \,  , ~~
  \label{eq_qec_lambda_total_eig}
\end{equation}
where $\lambda^{(\ncyc)}_{p,m,P}$ is the \emph{Pauli eigenvalue} of the channel $\Lambda^{(\ncyc)}_{p,m}$ 
for the eigenoperator $P$.

A key observation is that an approximate logical Markovian model $\mathcal{M}$ 
(see Defs.~\ref{def_logical_markov_model} and \ref{def_markov_approx})
of $\Lambda^{(\ncyc)}_{p,m}$ exists if, for every logical Pauli $P$, 
the associated Pauli eigenvalue $\lambda^{(\ncyc)}_{p,m,P}$ is well approximated by a function that 
decays exponentially in the number of QEC cycles, $\ncyc$. We formalize this observation in the following proposition.

\begin{prop}
  Let $\mathsf{Exp}$ be a set of experimental settings, as in Def.~\ref{defn_exp_settings}. 
  For each $(p,m) \in \mathsf{Exp}$ and each $\ncyc \geq 1$, let $\Lambda^{(\ncyc)}_{p,m}$ be the 
  associated total logical error channel (see Eq.~\ref{eq_lambda_total} of Def.~\ref{def_total_model}).
  Denote by $\lambda^{(\ncyc)}_{p,m,P}$ the Pauli eigenvalue of $\Lambda^{(\ncyc)}_{p,m}$ for the logical Pauli $P \in \PauliSet*{L}$, 
  as in Eq.~\ref{eq_qec_lambda_total_eig}.

  Let $\epsilon$ and $G'$ be positive.
  For each logical Pauli $P \in \PauliSet*{L}$, let $\chi^{\,}_{P} \in \Reals$; 
  for each $P \in \PauliSet*{L}$ and each experimental setting $(p, m) \in \mathsf{Exp}$, 
  let $C_{p,P}^{(\mathrm{prep})}\in \Reals$ and $C_{m,P}^{(\mathrm{meas})} \in \Reals$. 
  Suppose that, for each $(p,m) \in \mathsf{Exp}$, each $\ncyc \geq 1$, and each $P \in \PauliSet*{L}$,
  \begin{equation}
    \Big\lvert \lambda^{(\ncyc)}_{p,m,P} - C_{p,P}^{(\mathrm{prep})} C_{m,P}^{(\mathrm{meas})} \chi_P^{\ncyc} \Big\rvert \leq G' \epsilon^{\ncyc} \, .~~
  \end{equation}
  Let $G = G' \, \sqrt{D}$, where $D = 2^{\nlog}$ is the dimension of $L$. 
  Then there exists a $(G,\epsilon)$-approximate logical Markovian model for $\mathsf{Exp}$ (see Def.~\ref{def_markov_approx}).
  \label{prop_exp_decay_approx_model}
\end{prop}
\begin{proof}
  Define a fixed logical superoperator $\Phi^{\,}_{\text{QEC}}$ and logical superoperators 
  $\Phi^{(\text{prep})}_{p}$ and $\Phi^{(\text{meas})}_{m}$ for each preparation $p \in \mathsf{Prep}$ 
  and measurement $m \in \mathsf{Meas}$, respectively, such that
  \begin{subequations}
    \label{eq_Phis_from_Paulis}
    \begin{align}
      \Phi^{(\text{prep})}_p (P) &= C_{p,P}^{(\text{prep})} \, P \label{eq_Phi_prep_from_Pauli} \\
      \Phi^{\vpp}_{\text{QEC}} (P) &= \chi^{\vpp}_P \, P \label{eq_Phi_QEC_from_Pauli} \\
      \Phi^{(\text{meas})}_m (P) &= C_{m,P}^{(\text{meas})} \, P \label{eq_Phi_meas_from_Pauli} 
    \end{align}
  \end{subequations}
  for all logical Paulis $P \in \PauliSet*{L}$. We note that each logical superoperator above is fully specified by knowledge of its action 
  on every Pauli $P \in \PauliSet*{L}$.

  The families of superoperators in Eq.~\ref{eq_Phi_families} are given by $\mathcal{F}^{\,}_{\text{prep}} = \{ \Phi^{(\text{prep})}_p \}_{p}$ 
  and $\mathcal{F}^{\,}_{\text{meas}} = \{ \Phi^{(\text{meas})}_m \}_{m}$. Then, the corresponding logical Markovian model is
  $\mathcal{M} = (\mathcal{F}^{\,}_{\text{prep}},\Phi^{\,}_{\text{QEC}},\mathcal{F}^{\,}_{\text{meas}})$, as in Eq.~\ref{eq_Markov_triple}.
  By Lemma~\ref{lem_approx_model_eigenvalues}, for each setting $(p,m) \in \mathsf{Exp}$ and for each $\ncyc \geq 1$,
  a uniform bound of $G' \, \epsilon^{\ncyc}$ on the difference between the true Pauli eigenvalue 
  $\lambda^{(\ncyc)}_{p,m,P}$~\eqref{eq_qec_lambda_total_eig} and the eigenvalue 
  $C_{p,P}^{(\text{prep})} C_{m,P}^{(\text{meas})} \chi_P^{\ncyc}$ of the approximate model implies a bound of 
  $G \, \epsilon^{\ncyc}$ on the difference between the approximate probabilities 
  $\mu^{(\ncyc)}_{p,m} (\elt)$~\eqref{eq_Markov_prob_requirement} and the true probabilities 
  $\operatorname{Prob}(\elt)$~\eqref{eq_meas_prob_tot}. Hence, 
  $\mathcal{M}$ realizes a $(G,\epsilon)$-approximate logical Markovian model for $\mathsf{Exp}$. 
\end{proof}

\section{The \smip{}}
\label{sec_smip}

Here we describe a property of QEC cycles that we call the \smip{}, 
which simplifies the model of a single QEC cycle in Eq.~\ref{eq_qec_gamma_lambda}. 
We define this property below, and discuss two means of ensuring that it holds.

\begin{defn}[\Smip{}]
  A QEC cycle with Pauli-stochastic noise (see Secs.~\ref{sec_prelim_noiseless} and \ref{sec_noisy_qec_cycles})
  has the \emph{\smip} if the transition probabilities $\gamma_{\soutput,\sinput}$ in Eq.~\ref{eq_qec_final} 
  are independent of $\sinput$.
  \label{def_smip}
\end{defn}
When the \smip{} holds, we replace the transition probabilities $\gamma_{\soutput,\sinput}$ with $\gamma_{\soutput}$;
we refer to the latter as the \emph{syndrome marginal distribution}. Importantly, if the QEC cycle has the \smip{},
then the Markov process $\gamma_{\soutput,\sinput}$ in each time step of 
Eq.~\ref{eq_qec_final} can be replaced with independent samples from the syndrome marginal distribution $\gamma_{\soutput}$.

We now discuss two means of guaranteeing  that the \smip{} of Def.~\ref{def_smip} holds. 
First, we show that the \smip{} holds if the QEC cycle defined in Sec.~\ref{sec_prelim_noiseless} has what we call the \ssp{}.
Second, we show that the \smip{} can be enforced for arbitrary decoding functions $D$ 
via randomization of the syndrome subsystem $S$.
\begin{defn}[\Ssp{}]
  \label{def_ssp}
  Consider a QEC cycle as described in Sec.~\ref{sec_prelim_noiseless} 
  with syndrome encoding and decoding functions $E$ and $D$, respectively. 
  If these functions satisfy $D(E(s) + \theta^\prime) = s + D(\theta^\prime)$ 
  for all bitstrings $s$ of $S$ and all bitstrings $\theta^\prime$ of $O$,
  then the QEC cycle has the \emph{\ssp{}}.
\end{defn}
We now establish that the \ssp{} implies the \smip{}. This is formalized in the following Proposition.
\begin{prop}[Syndrome-decoding symmetry implies syndrome marginal independence]
  Consider a QEC cycle as described in Sec.~\ref{sec_prelim_noiseless} 
  with Pauli-stochastic noise as described in Sec.~\ref{sec_noisy_qec},
  so that Prop.~\ref{prop_qec_final} holds. If the QEC cycle has the \ssp{} (see Def.~\ref{def_ssp}), then 
  it has the \smip{} (see Def.~\ref{def_smip}).
   \label{prop_ssp_smip}
\end{prop}

Intuitively, Prop.~\ref{prop_ssp_smip} holds because the \ssp{} ensures that,
 for each term in the sum in Eq.~\ref{eq:rho out naive}, the dependence on $\sinput$ is eliminated.
The full proof appears in App.~\ref{app_ssp_smip}.

We expect most pairs of syndrome encoding and decoding functions $E$ and $D$ to be compatible with the \ssp{}. 
For example, if $E$ and $D$ are, respectively, 
an encoder and minimum-weight decoder for a classical linear code with \emph{odd} code distance $d$,
then the QEC cycle has the \ssp{}. 
This can be verified by first noting that, in such circumstances, 
any bitstring $\theta'$ of the system $O$ can be expressed as $\theta' = E(s') + \theta''$, 
where $s'$ is a bitstring of system $S$ and $\theta''$ is a bistring of system $O$, where the latter has a Hamming weight less than $d$.
Then, the minimum-weight decoder satisfies $D(\theta') = s'$, 
because $\theta'$ is closest to $E(s')$ in terms of the Hamming distance.
Because the encoder $E$ for any classical linear code satisfies $E(s_1 + s_2) = E(s_1) + E(s_2)$ for all bitstrings $s_1$ and $s_2$ of system $S$,
it follows that 
\begin{equation}
  D( E(s) + \theta') = D( E (s+s') + \theta'') = s + s' = s + D(\theta') \, ,~
\end{equation}
which is the \ssp{} (see Def.~\ref{def_ssp}). 
Additionally, in many QEC contexts, syndromes are fault-tolerantly extracted using classical repetition codes~\cite{gottesman2016surviving}, 
which is a special case of the argument above for linear codes.

We now explain how the \smip{} can be enforced by a modification of the QEC cycle 
that we call syndrome-state randomization, which does not require 
the \ssp{}.
We define syndrome-state randomization below, and show how it leads to the \smip{} in Prop.~\ref{prop_syndrome_rand_smip}.

\begin{defn}[Syndrome-state randomization]

  Given a QEC cycle as described in Sec.~\ref{sec_prelim_noiseless}, 
  \emph{syndrome-state randomization} is a modification to the QEC cycle. Prior to Step~\ref{qec_step_C},
  a bitstring $s_r$ of the subsystem $S$ is drawn uniformly at random. One then applies a Pauli superoperator to $S$ 
  that updates the joint input state on $L$ and $S$ according to
  $\rho_L \otimes \ZProj{\sinput} \mapsto \rho_L \otimes \ZProj{\sinput + s_r}$. 
  Importantly, knowledge of the bitstring $s_r$ 
  is not used for sydrome reset (Step~\ref{qec_step_reset}), but is used for logical recovery (Step~\ref{qec_step_EC}) by modifying $\smeas$ to $\smeas + s_r$.
  \label{defn_rand_synd}
\end{defn}

Intuitively, Def.~\ref{defn_rand_synd} ensures that the state of $S$ after the ``Reset'' operation in Diagram~\ref{circ_qec_step_perfect}
is uncorrelated with the input syndrome state $\sinput$, as the randomization removes all dependence on $\sinput$. 
We formalize this intuition in Prop.~\ref{prop_syndrome_rand_smip}, below.
\begin{prop}[Syndrome randomization implies syndrome marginal independence]
  Consider a QEC cycle as defined in Sec.~\ref{sec_prelim_noiseless}. 
  If the input syndrome state is randomized (see Def.~\ref{defn_rand_synd}),
  then the modified QEC cycle has the \smip{} for \emph{any} pair of syndrome encoding and decoding functions $E$ and $D$.
  \label{prop_syndrome_rand_smip}
\end{prop}
The proof of Prop.~\ref{prop_syndrome_rand_smip} appears in App.~\ref{app_syndrome_rand_proof} 
and closely parallels the proof of Prop.~\ref{prop_qec_final}.
Intuitively, the \smip{} holds because the input syndrome bitstring $\sinput$ is used only for 
logical recovery (Step~\ref{qec_step_EC} of the QEC cycle), and \emph{not} for syndrome reset (Step~\ref{qec_step_reset}). 
As a result, the syndrome bitstring is effectively randomized in a manner that depends only on  $\soutput$.
In the remainder of this paper, we assume that the QEC cycle has the \smip{}, either via the \ssp{} 
or via syndrome randomization.

\section{Logical Markovian models for low error rate}
\label{sec_eff_markovian}

\subsection{Requirement for state preparation}
\label{sec_state_prep_gamma}

To simplify the construction of approximate logical Markovian models (see Def.~\ref{def_markov_approx}), 
we now impose an additional constraint on state preparation, which enters Theorem~\ref{thm_main}.
In particular, we require that the distribution $\gamma^{(\text{prep})}_{s_0}$ in the model of noisy state preparation~\eqref{eq_noisy_prep} 
is equal to the syndrome marginal distribution $\gamma^{\,}_{s_0}$ defined in Sec.~\ref{sec_smip}.
For example, if state preparation is implemented fault tolerantly, it will generally include a QEC cycle;
if it includes the \emph{same} QEC cycle that defines the QEC experiment (see Def.~\ref{defn_QEC_experiment}),
then the requirement $\gamma^{(\text{prep})}_{s_0} = \gamma^{\,}_{s_0}$ is automatically satisfied.
The constraint $\gamma^{(\text{prep})}_{s_0} = \gamma^{\,}_{s_0}$ can be enforced generically by simply absorbing the first QEC cycle into state preparation,
at the cost of changing $\ncyc$ to $\ncyc-1$.

\subsection{Changing frames for Pauli channels}
\label{sec_frame_change}

We now introduce a useful modification to the channels $\Lambda^{(\text{prep})}_{s_0}$, 
$\Lambda^{\vpp}_{\soutput,\sinput}$, and $\Lambda^{(\text{meas})}_{s_{\ncyc}}$ associated with
state preparation~\eqref{eq_noisy_prep}, QEC cycles~\eqref{eq_qec_gamma_lambda}, and logical measurement~\eqref{eq_Lambda_meas}, respectively. 
In particular, it is desirable for these channels to be close to the identity channel when the total error rate is low. 
As previously defined, these channels do not have this property: even in the low-error regime, 
they may include correctable logical errors that will be successfully corrected in the next QEC cycle, 
or logical corrections from prior cycles. We address this by (\emph{i}) appending to the state-preparation channels and QEC channels 
a unitary channel $\Cor$ conditioned on the true state of system $S$ that perfectly corrects these errors and (\emph{ii})
 prepending to the subsequent nominal channel the inverse unitary channel $\Cor^{-1}$, for all terms in $\Lambda^{(\ncyc)}_{\text{tot}}$~\eqref{eq_lambda_total}.
We use tildes to distinguish the modified channels, which are given by
\begin{subequations}
  \label{eq_cor_frame_channels}
  \begin{align}
    \widetilde{\Lambda}^{(\text{prep})}_{s_0} &= \Cor^{\vpd}_{s_0} \circ {\Lambda}^{(\text{prep})}_{s_0} \label{eq_cor_frame_prep_chan} \\
    \widetilde{\Lambda}^{\vpp}_{\soutput,\sinput} &= \Cor^{\vpd}_{\soutput} \circ \Lambda^{\vpp}_{\soutput,\sinput} \circ \Cor^{\dagger}_{\sinput} \label{eq_cor_frame_qec_chan} \\
    \widetilde{\Lambda}^{(\text{meas})}_{s_{\ncyc}} &= {\Lambda}^{(\text{meas})}_{s_{\ncyc}} \circ \Cor^{\dagger}_{s_{\ncyc}} \label{eq_cor_frame_meas_chan}
  \end{align}
\end{subequations}
where the unitary error-correction channel $\Cor$ is defined in Step~\ref{qec_step_EC} of the QEC cycle (see Eq.~\ref{eq:EC step}).
In terms of these channels, the total logical error channel $\Lambda_{\text{tot}}$~\eqref{eq_lambda_total} is
\begin{align}
  \Lambda^{(\ncyc)}_{\text{tot}} (\sigma_L) = \hspace{-2mm} \sum_{s_0,\dots,s_\ncyc} \hspace{-0.7mm}
    \gamma^{\vpp}_{s_{\ncyc}} \hspace{-0.8mm} \cdots \gamma^{\vpp}_{s_1} 
    \gamma^{\vpp}_{s_0}~
    \widetilde{\Lambda}^{\text{(meas)}}_{s_\ncyc} \hspace{-0.3mm} \circ \widetilde{\Lambda}^{\vpp}_{s_{\ncyc},s_{\ncyc-1}} \hspace{-0.8mm} \circ 
    \cdots \circ \widetilde{\Lambda}^{\vpp}_{s_1,s_0}\hspace{-0.5mm}
    \circ  \widetilde{\Lambda}^{\text{(prep)}}_{s_0} ( \sigma_L)
    \label{eq_qec_cor_frame_total} \, ,~~
\end{align}
which resembles Eq.~\ref{eq_lambda_total} with $\Lambda \to \widetilde{\Lambda}$. 
We note that we have also used the constraint $\gamma^{\text{prep}}_{s_0} = \gamma_{s_0}.$.
 As a reminder, 
the utility of the modified channels is that they approach the identity channel in the limit of low logical error rate.

\subsection{High-quality QEC cycles}
\label{sec_high_quality}

Here we introduce the notion of high-quality QEC cycle, which we associate with low logical error rate. 
This definition is important to the main result, Theorem~\ref{thm_main}. Below, we define the notion of a 
high-quality QEC cycle in terms of the entanglement fidelity~\cite{Nielsen_2002}.
\begin{defn}[Single-cycle logical entanglement fidelity]
  Consider a QEC cycle with 
  the \smip{} (see Def.~\ref{def_smip}), let $\gamma_s$ be the syndrome marginal distribution, and define
  $\widetilde{\Lambda}$ as in Eq.~\ref{eq_cor_frame_qec_chan}.
  Define the \emph{single-cycle logical error channel} as
  \begin{equation}
    \label{eq_Lambda_1}
    \Lambda_1 \coloneqq \sum\limits_{\soutput,\sinput} \, \gamma_{\soutput} \, \widetilde{\Lambda}_{\soutput,\sinput} \, \gamma_{\sinput} \, .~~
  \end{equation}
  Then, the \emph{single-cycle logical entanglement fidelity}
  $f_1$ is the entanglement fidelity of 
  $\Lambda_1$.
  \label{defn_log_fid}
\end{defn}

A standard formula for the entanglement fidelity is
\begin{equation}
  \label{eq_f1}
  f_1 \coloneqq \frac{1}{D^3} \sum\limits_{P \in \PauliSet*{L}} \trace \left( P \, \Lambda_1 ( P )   \right)\, , ~~
\end{equation}
where $D = 2^{\nlog}$ is the dimension of the logical Hilbert space~\cite{Nielsen_2002}.
Furthermore, because $\Lambda_1$ is a Pauli channel, 
$f_1$ is equal to the probability of the
identity superoperator in the mixture of Pauli superoperators that makes up $\Lambda_1$.
As noted in Sec.~\ref{sec_frame_change}, if all errors are successfully corrected, 
the logical error channels $\widetilde{\Lambda}^{\,}_{\soutput,\sinput}$~\eqref{eq_cor_frame_qec_chan} are each the identity and thus $f_1 = 1$. 
Moreover, $f_1$ can be interpreted as the probability of no logical error occuring, averaged over input syndrome bitstrings drawn from the syndrome marginal distribution.
Thus, we believe that it is reasonable to require that a ``high quality'' QEC cycle has $f_1 \geq 1- 1/64$, as we do in Theorem~\ref{thm_main}.
The particular definition of $\Lambda_1$ is also motivated by the fact that this 
quantity is experimentally accessible in certain settings, as we describe in App.~\ref{sec_discuss_defn_ss_lf}. 
We also describe in App.~\ref{sec_discuss_defn_ss_lf} how, when the criteria of Theorem~\ref{thm_main} hold, 
the Pauli eigenvalues of $\Lambda_1$ can be interpreted as a first-order approximation of the Pauli eigenvalues of the superoperator $\Phi_{\text{QEC}}$ that 
partly defines the approximate logical Markovian model $\mathcal{M}$ (see Def.~\ref{def_markov_approx}).

\subsection{Main results}
\label{main_thm_statement}

\begin{figure}[h]
  \centering
  \includegraphics[width=0.8\textwidth]{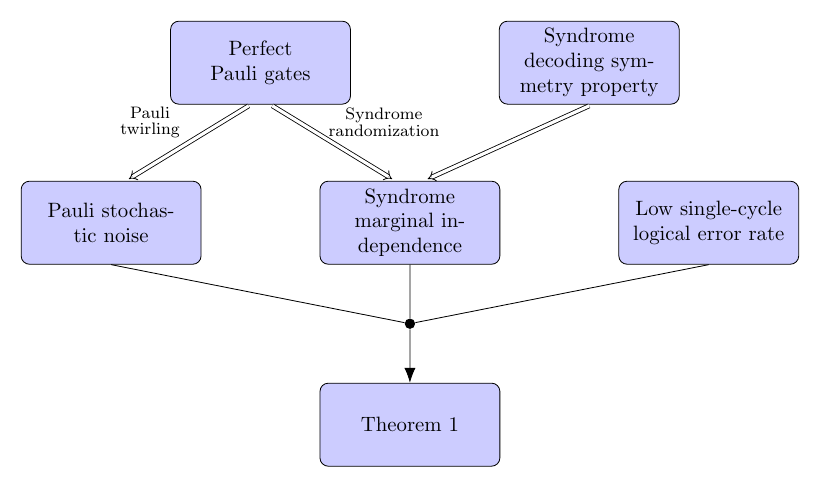} 
  \caption{Diagram showing the key assumptions of Theorem~\ref{thm_main} and how they may be enforced in practice.
  Each of the three blocks in the middle row corresponds to a requirement of Theorem~\ref{thm_main}, 
  as indicated by the regular arrows that jointly point to Theorem~\ref{thm_main}.
  Further details about Pauli-stochastic noise appear in Sec.~\ref{sec_noisy_qec}; 
  further details of the \smip{} appear in Sec.~\ref{sec_smip}; 
  and further details of the notion of ``low single-cycle logical error rate'' appear in Sec.~\ref{sec_high_quality}.
  The blocks in the upper row are assumptions that can be used to enforce the requirements in the middle row, 
  with implications indicated by double arrows. 
  Concretely, perfect Pauli gates can be used to enforce Pauli-stochastic noise via Pauli twirling (see Sec.~\ref{sec_noisy_qec}), 
  and can also be used to enforce the \smip{} via syndrome randomization (see Sec.~\ref{sec_smip}).
  Furthermore, the \ssp{} can also be used independently to enforce the \smip{} if needed (see Sec.~\ref{sec_smip}).}
  \label{fig_assumptions}
\end{figure}

\begin{mainthm}
  Consider a fixed QEC cycle as defined in Sec.~\ref{sec_prelim_noiseless} 
  with Pauli-stochastic noise as defined in Secs.~\ref{sec_Pauli_stochastic}--\ref{sec_Pauli_noise_meas}. 
  Suppose that the QEC cycle has the \smip{} as defined in Sec.~\ref{sec_smip} and that the single-cycle logical entanglement
  fidelity (see Def.~\ref{defn_log_fid}) satisfies $f_1 \geq 1 - \epsilon_1$ for 
  some $\epsilon_1$ satisfying $0 < \epsilon_1 \leq 1/64$. 
  Let $\mathsf{Exp}$ be a set of experimental settings (see Def.~\ref{defn_exp_settings}).
  Define $\epsilon = 2(1+\sqrt{2})\sqrt{\epsilon_1}$, 
  $G' = 1 + 1/(1-\epsilon/(1-2\sqrt{\epsilon_1})) < 7$, 
  and $G = G'\sqrt{D}$, where $D = \operatorname{dim}(\hilbert_L) = 2^{\nlog}$ is the dimension of $L$.
  Then there exists a $(G,\epsilon)$-approximate logical Markovian model $\mathcal{M}$ 
  for 
  $\mathsf{Exp}$ (see Def.~\ref{def_markov_approx}). 
  \label{thm_main}
\end{mainthm}
The proof of Theorem~\ref{thm_main} is presented in App.~\ref{sec_proof_main_thm} and is constructive, 
in that it provides a method for computing the $(G,\epsilon)$-approximate logical Markovian model $\mathcal{M}$ 
based on the specifics of the QEC cycle and the noise model under consideration. 
In particular, $\mathcal{M}$ is constructed by approximating the Pauli eigenvalues of the total logical error channel
using Prop.~\ref{prop_exp_decay_approx_model} and Lemma~\ref{lem_vector_geo_sum}.

The various means of fulfilling the assumptions entering Theorem~\ref{thm_main} are summarized 
in Fig.~\ref{fig_assumptions}. We expect these requirements to be compatible with a broad range of
experiments involving QEC cycles. We note that the requirement of a low 
single-cycle logical error rate $\epsilon_1 \leq 1/64$ is included to ensure that the perturbing matrix
in Eq.~\ref{eq_perturb_norm_small} of the proof of Theorem~\ref{thm_main} has sufficiently small 
operator norm for all logical Paulis $P$. However, if the operator norm $\epsilon'$ of this 
perturbing matrix is known \emph{a priori} (e.g., from a theoretical calculation starting 
from a known noise model), then one can also replace the requirement $\epsilon_1 \leq 1/64$ 
in Theorem~\ref{thm_main} with the requirement that $\epsilon' \leq 1/4$ for all logical Paulis $P$. 
In this case, the $\epsilon$ that appears Theorem~\ref{thm_main} is replaced by 
$\epsilon = (1+\sqrt{2})\epsilon'$.

\section{Discussion and outlook}
\label{sec_discussion}

To summarize, in Theorem~\ref{thm_main} we proved that, under broad circumstances, 
there exists an approximate logical Markovian model describing repeated QEC cycles of a stabilizer code.
Thus, when the assumptions of Theorem~\ref{thm_main} are met, the
logical  behavior of repeated QEC cycles can be approximated  
by an error superoperator $\Phi_\text{QEC}$ acting solely on the logical subsystem.
The deviations of the behavior predicted by $\Phi_{\text{QEC}}$ from that predicted by the 
true logical error model $\Lambda^{(\ncyc)}_{p,m}$~\eqref{eq_lambda_total}
are exponentially suppressed in the number of QEC cycles $\ncyc$.
Furthermore, $\Phi_\text{QEC}$ is Pauli diagonal (see Sec.~\ref{sec_Pauli_eigs}), 
and can thus be inferred using standard Pauli-learning techniques~\cite{Pauli_learn}.
Altogether, this simplifies the analysis and interpretation of repeated QEC cycles, 
and provides an optimization metric for the design of QEC cycles.

As a concrete example,  consider a QEC memory experiment in which logical information is encoded into a stabilizer code 
(see Sec.~\ref{sec_Pauli_noise_state_prep}), maintained over $\ncyc$ QEC cycles (see Sec.~\ref{sec_noisy_qec_cycles}), 
and read out from a final logical measurement (see Sec.~\ref{sec_Pauli_noise_meas}).
Our results imply that, under broad circumstances, such a QEC memory experiment can be modeled by a 
logical Markovian model $\mathcal{M}$ for which $\Phi_\text{QEC}$ is Pauli diagonal. 
Consequently, when the final logical measurement is of the expectation of a logical Pauli operator, 
then the QEC memory experiment has an exponentially decaying probability of success as a function 
of $\ncyc$, to extremely good approximation.
This justifies the common technique of fitting exponential decays to the results of QEC memory experiments.
Furthermore, the rates of these exponential decays---which depend on the Pauli being measured---can be used as 
performance metrics to inform both the characterization of QEC memory experiments and the design of QEC cycles.

As a more nuanced
example, consider a logical Clifford randomized-benchmarking experiment, 
where a noisy QEC cycle is performed after each logical gate.
Assuming that the noise associated with the logical Cliffords is gate independent, 
so that the corresponding error channels can be 
absorbed into the errors associated with the subsequent QEC cycle.
The effect of the logical Clifford gates is to depolarize the logical error channels $\widetilde{\Lambda}_{\soutput,\sinput}$. 
Thus, our results apply as long as the requirements of Theorem~\ref{thm_main} are met.
Consequently, under broad circumstances, we expect such logical Clifford randomized-benchmarking experiments 
to have exponentially decaying success probabilities, with a decay rate determined by a depolarized version 
of the logical superoperator $\Phi_{\text{QEC}}$ that enters the approximate logical Markovian model $\mathcal{M}$.
We anticipate that future work will formalize this application and extend 
 it to the more realistic
 case of gate-dependent errors.

In addition to the direct application of our work to characterizing logical performance and optimizing the design of QEC cycles, we see several interesting 
avenues of future research based on our results. 
One important direction is to relax the assumptions entering Theorem~\ref{thm_main}. Of particular interest are experiments in which the QEC cycles do not have identical noise models, and where the noise is not Pauli stochastic. 
Another appealing direction is to derive similar results beyond the context of stabilizer QEC with qubits, including subsystem codes and Floquet codes with qubits, stabilizer codes with qudits, and bosonic codes. 
Another promising direction is to adapt our methods to QEC scenarios beyond memory experiments, including logical Clifford randomized-benchmarking experiments with gate-dependent errors and the characterization of logical computations involving magic gates. 
Each of these extensions presents its own challenges, but our results provide a first step toward achieving these goals.

\begin{acknowledgments}
  A.J.F. acknowledges support from the Professional Research Experience Program (PREP) operated jointly by NIST and the University of Colorado. This work includes contributions of the National Institute of
  Standards and Technology, which are not subject to U.S. copyright. 
  This research is based upon work supported in part by the Office of the Director of National Intelligence (ODNI), Intelligence Advanced Research Projects Activity (IARPA), specifically the ELQ program. The views and conclusions contained herein are those of the authors and should not be interpreted as necessarily representing the official policies or endorsements, either expressed or implied, of the ODNI, IARPA, or the U.S. Government. The U.S. Government is authorized to reproduce and distribute reprints for Governmental purposes notwithstanding any copyright annotation thereon.
  The authors thank Michael Gullans, Srilekha Gandhari, Robin Blume-Kohout, Kenneth Rudinger, Kevin Young, Mohammad Alhejji, and Ariel Shlosberg for helpful comments and discussions.
\end{acknowledgments}

\appendix

\renewcommand{\thesubsection}{\thesection.\arabic{subsection}}
\renewcommand{\thesubsubsection}{\thesubsection.\arabic{subsubsection}}
\renewcommand{\theequation}{\thesection.\arabic{equation}}

\section{Updates to the density matrix}
\label{app_derivations}

Here we derive the updates to the density matrix of the $L$, $S$, $A$, and $O$ systems associated with state preparation, QEC cycles, and logical measurement.
In particular, we derive analytic expressions for various probabilities, channels, and transition matrices that appear in Sec.~\ref{sec_noisy_qec}. 
In defining these quantities, it is convenient to define the following indicator function:
\begin{equation}
    \label{eq:Pauli bitstring indicator}
    \mathcal{I}^{\bvec{s}}_P \coloneqq \prod_{j=1}^{\abs{\bvec{s}}} 
    \left( \kron{s_j,0} \kron{\Pauli*{j},\ident} + \kron{s_j,0} \kron{\Pauli*{j},Z} + \kron{s_j,1} \kron{\Pauli*{j},X} + \kron{s_j,1} \kron{\Pauli*{j},Y} \right) \, .~~
\end{equation}
where $s_j$ is the $j$th digit of the bitstring $s$ with length $\abs{s}$ 
and $\Pauli*{j}$ is the single-qubit Pauli operator located in the $j$th tensor factor of the Pauli operator $P$.
For a fixed $P$ and $s$,  $\mathcal{I}^{\bvec{s}}_P$ is either zero or one, and for a fixed $P$ there is exactly one bitstring $s$ such that $\mathcal{I}^{\bvec{s}}_P$ is equal to one.
The relationship to the action of a Pauli $P$ on computational basis states is as follows
\begin{equation}
  P \ZProj{s} P = \sum_{s'} \ZProj{s+s'} \, \mathcal{I}^{s'}_P \, .~~
\end{equation}

\subsection{Noisy state preparation}
\label{app_noisy_state_prep}

State preparation with Pauli-stochastic noise results in a state $\rho^{\,}_0$~\eqref{eq_noisy_prep} of the form
\begin{equation}
  \label{eq_noisy_prep_explicit}
  \rho_0  = \ErrChan_{\text{prep}} ( \rho_{\text{enc}}) = \ErrChan_{\text{prep}} \left( \sigma_L \otimes \ZProj{s_*} \right) \, ,~~
\end{equation}
where $\rho_{\text{enc}}$ is the ideal state and 
$\ErrChan_{\text{prep}}$ is a Pauli channel~\eqref{eq:general Pauli channel} on 
 $L$ and $S$. 
In particular, 
\begin{equation}
  \label{eq_state_prep_error_chan}
  \ErrChan_{\text{prep}}( \rho_L \otimes \rho_S) = \sum\limits_{\Pauli*{\ell} \in \PauliSet*{L}} \sum\limits_{\Pauli*{s} \in \PauliSet*{S}}  
  \mu^{(\text{prep})}_{\ell,s} \, 
  \PauliChan{\ell} (\rho_L) \otimes \PauliChan{s} (\rho_S) \, ,~~
\end{equation}
where $\mu^{(\text{prep})}_{\ell,s}$ are probabilities and $\rho_{L}$ and $\rho_S$ are arbitrary states on $L$ and $S$, respectively. 

Given $\ErrChan_{\text{prep}}$ as defined in Eq.~\ref{eq_state_prep_error_chan}, we derive Eq.~\ref{eq_noisy_prep} using
\begin{subequations}
  \label{eq_state_prep_probs}
  \begin{align}
    \gamma_{s_0}^{(\text{prep})} &= \sum\limits_{\Pauli*{\ell} \in \PauliSet*{L}} \sum\limits_{\Pauli*{s} \in \PauliSet*{S}} 
    \mu^{(\text{prep})}_{\ell,s} \, \mathcal{I}^{s_0 + s_*}_{\Pauli*{s}}  \\
    \Lambda^{(\text{prep})}_{s_0} (\rho_L) &= \frac{1}{ \gamma_{s_0}^{(\text{prep})}} 
    \sum\limits_{\Pauli*{\ell} \in \PauliSet*{L}} \sum\limits_{\Pauli*{s} \in \PauliSet*{S}} 
    \mu^{(\text{prep})}_{\ell,s} \, \mathcal{I}^{s_0 + s_*}_{\Pauli*{s}} \PauliChan{\ell} (\rho_L)  \label{eq_state_prep_gamma}
    \, , ~~
  \end{align}
\end{subequations}
when $\gamma_{s_0}^{(\text{prep})} \neq 0$, and $\Lambda^{(\text{prep})}_{s_0} (\rho_L) = 0$ otherwise. 
The indicator $\mathcal{I}$ is defined in Eq.~\ref{eq:Pauli bitstring indicator}.


\subsection{Noisy syndrome extraction}
\label{app_QEC_update_C}

Step~\ref{qec_step_C} of each QEC cycle is the noisy application of the unitary $C$~\eqref{circ_cliff_readout}, 
which measures the syndrome subsystem $S$ and records the outcomes in $O$. Consider the input state
\begin{equation*}
\tag{\ref{eq:QEC input state}}
  \rho_{\text{in}} =\rho_L \otimes \ZProj{\sinput} \, , ~~
\end{equation*}
where general input states 
are convex sums of pure states of the form above. 
We then introduce the registers $A$ and $O$, 
so that the state on a $L$, $S$, $A$, and $O$~\eqref{circ_qec_step_perfect} is
\begin{equation}
  \varrho_{\text{in}} = \rho_L \otimes \ZProj{\sinput} \otimes \ZProj{a} \otimes \ZProj{0} \, , ~~
\end{equation}
where $\ZProj{a}$ and $\ZProj{0}$ are computational-basis product states on $A$ and $O$, respectively. 

We then consider the noisy implementation of $C$, which we denote by $C'$. 
The ideal implementation of $C$ is captured by Diagram~\ref{circ_cliff_readout}, and acts as
\begin{equation}
    \label{eq_c_unitary_ideal}
    C \left[ \rho_L \otimes \ZProj{s_{\text{in}}} \otimes \ZProj{a} \otimes \ZProj{0} \right] C^{\dagger} 
    = \rho_L \otimes \ZProj{s_{\text{in}}} \otimes \ZProj{b} \otimes \ZProj{E( s_{\text{in}}) }\, ,~~
\end{equation}
and, as defined in Sec.~\ref{sec_Pauli_stochastic}, the assumption of Pauli-stochastic noise means that $C'$ can be modeled as
the ideal unitary $C$~\eqref{eq_c_unitary_ideal} followed by a generic Pauli channel $\ErrChan$~\eqref{eq:general Pauli channel}, i.e.,
\begin{equation}
    C' (\varrho_{\text{in}}) \coloneqq 
    \ErrChan ( C \varrho_{\text{in}} C^\dagger )  
    = \sum\limits_{P  \in \PauliSet{\text{all}}} \mu^{(\ErrChan)}_P \, P  C \varrho_{\text{in}} C^{\dagger}  P \, , ~\label{eq:noisy C}
\end{equation}
where $\PauliSet*{\text{all}}$ is the set of Pauli operators~\eqref{eq:Pauli set} on $L$, $S$, $A$, and $O$, 
$\mu^{(\ErrChan)}_P$ encodes a probability distribution, 
and at this point, $\ErrChan$ may differ between QEC cycles.

Following the application of $C'$~\eqref{eq:noisy C}, the state of the full system is given by
\begin{align}
  \varrho &= 
  \sum\limits_{\Pauli*{\ell'} \in \PauliSet*{L}} 
  \sum\limits_{\Pauli*{s} \in \PauliSet*{S}} 
  \sum\limits_{\Pauli*{a} \in \PauliSet*{A}} 
  \sum\limits_{\Pauli*{o} \in \PauliSet*{O}} 
  \mu^{(\ErrChan)}_{\Pauli*{\ell'} \otimes \Pauli*{s} \otimes \Pauli*{a} \otimes \Pauli*{o}} 
  \notag \\
  &\quad ~ \quad 
  \PauliChan{\ell'} (\rho_L ) 
  \otimes \PauliChan{s} ( \ZProj{\sinput} )
  \otimes \PauliChan{a} (\ZProj{b} ) 
  \otimes \PauliChan{o} (\ZProj{E(\sinput)} ) 
  \label{eq:full state after noisy C} \, .~~
\end{align}
We simplify this expression by tracing out 
(i.e., discarding) the ancilla register $A$ and replacing the Pauli superoperators on $S$ and $O$ with XOR of bitstrings~\eqref{eq:bit flips}, resulting in
\begin{align}
    \varrho &= 
    \sum\limits_{\Pauli*{\ell'} \in \PauliSet*{L}} 
    \sum\limits_{\Pauli*{s} \in \PauliSet*{S}} 
    \sum\limits_{\Pauli*{a} \in \PauliSet*{A}} 
    \sum\limits_{\Pauli*{o} \in \PauliSet*{O}} 
    \mu^{(\ErrChan)}_{\Pauli*{\ell'} \otimes \Pauli*{s} \otimes \Pauli*{a} \otimes \Pauli*{o}} 
    \sum\limits_{s'} 
    \mathcal{I}^{s'}_{\Pauli*{s}} \sum\limits_{\theta'} \mathcal{I}^{\theta'}_{\Pauli*{o}} \notag \\
    &\quad ~ \quad \PauliChan{\ell'} (\rho_L) 
    \otimes \ZProj{\sinput + s'} 
    \otimes \ZProj{E(\sinput) + \theta'} 
    \label{eq:reduced state after noisy C} \, ,
\end{align}
where $s' \in \Ints_2^{\abs{S}}$ runs over bitstrings on $S$ and 
$\theta' \in \Ints_2^{\abs{O}}$ runs over bitstrings on $O$. 

We next write this state in a more convenient form by defining 
a probability $\overline{p}^{\,}_{\ell', s', \theta'}$ associated with 
the logical Pauli channel $\PauliChan*{\ell'}$ and bit flips on the registers $S$ and $O$,
encoded via the bitstrings $s'$ and $\theta'$, respectively. 
The resulting state on $L$, $S$, and $O$ is given by
\begin{equation*}
    \varrho = \sum\limits_{\ell', s',\theta'} ~ \overline{p}^{\vpp}_{\ell', s', \theta'} ~ \PauliChan{\ell'} (\rho_L) 
    \otimes \ZProj{\sinput + s'} 
    \otimes \ZProj{E(\sinput) + \theta'} 
    \tag{\ref{eq_after_errchan}} \, ,~~
\end{equation*}
where $\overline{p}^{\,}_{\ell', s', \theta'}$ is defined in terms of 
the probabilities $\mu^{(\ErrChan)}_P$ that define
 $\ErrChan$~\eqref{eq:noisy C} via
\begin{align}
  \label{eq:QEC intermediate probabilities}
    \overline{p}^{\vpp}_{\ell', s', \theta'} &\coloneqq 
    \sum\limits_{\Pauli*{s} \in \PauliSet*{S}} 
    \sum\limits_{\Pauli*{a} \in \PauliSet*{A}} 
    \sum\limits_{\Pauli*{o} \in \PauliSet*{O}} 
    \mu^{(\ErrChan)}_{\Pauli*{\ell'} \otimes \Pauli*{s} \otimes \Pauli*{a} \otimes \Pauli*{o}} 
    \mathcal{I}^{s'}_{\Pauli*{s}} \mathcal{I}^{\theta'}_{\Pauli*{o}} \, . ~~
\end{align}

\subsection{Reset and error-correction operations}
\label{app_QEC_update_2}

Following noisy syndrome information readout (Step~\ref{qec_step_C} of the QEC cycle), the joint state on $L$, $S$, and $O$ 
is given by Eq.~\ref{eq_after_errchan}. Next, the error syndrome $\serr$~\eqref{eq_s_err} is 
recovered from $O$ via classical side processing (Step~\ref{qec_step_proc} of the QEC cycle), which we assume to be perfect.
Due to noisy syndrome readout, the inferred state of $S$ is given by
\begin{equation}
  \label{eq_s_meas}
  \smeas = D (\theta) = D \big( E (\sinput) + \theta' \big)  \, ~~
\end{equation}
where $\theta$ labels dephased states of $O$ and the error syndrome is $\serr = \smeas + s_*$~\eqref{eq_s_err}.

One next attempts to reset the syndrome subsystem $S$ to the default configuration $s_*$ (Step~\ref{qec_step_reset} of the QEC cycle)
using the error syndrome $\serr$~\eqref{eq_s_err}. 
In particular, one applies the Pauli $X$ to the $j$th qubit in $S$ in the Clifford frame~\eqref{eq:CliffordFrame} 
if the $j$th digit of $\serr$ is one, and does nothing otherwise. 
The resulting state on $L$, $S$, and $O$ is given by
\begin{align}
  \label{eq_rho_after_S_reset}
  \varrho &= \sum\limits_{\ell', s',\theta'} ~ \overline{p}^{\vpp}_{\ell',s',\theta'} ~ \PauliChan{\ell'} (\rho_L) 
    \otimes \ZProj{\sinput + s' + \serr} 
    \otimes \ZProj{\theta}   \, ,~~
\end{align}
where $\serr$ and $\theta$ both depend on $\sinput$ and $\theta'$. 
Crucially, in the absence of noise---i.e., when  $s',\theta'=0$---the state on $S$ 
after Step~\ref{qec_step_reset} is simply $\ZProj{s_*}$, as required. 

Finally, we apply the logical error-correction operation EC
(Step~\ref{qec_step_EC} of the QEC cycle
). Since the register $O$ is traced out and reset at the end of the QEC cycle~\eqref{circ_qec_step_perfect}, we now discard it. 
The final joint state of the subsystems $L$ and $S$ is given by
\begin{align*}
    \rho_{\text{out}} &=  \sum_{\ell', s', \theta'} 
    \overline{p}^{\vpp}_{\ell', s', \theta'} ~
    \Cor_{\serr} \circ \PauliChan{\ell'} (\rho_L) 
    \otimes 
    \ZProj{\sinput + s' + \serr} 
     \tag{\ref{eq:rho out naive}} \, ,~~
\end{align*}
where $s_{\text{err}}$ depends on $\sinput$ and $\theta'$ and the logical channel $\Cor_{\serr}$ is defined in Eq.~\ref{eq:EC step}.

For convenience, we next define the output bitstring on $S$, given by
\begin{align}
  \label{eq_s_out}
    \soutput =  \sinput + s' + \serr = \sinput + s' + s_* + D \big ( E (\sinput) + \theta' \big)  \, ,~~
\end{align}
which we use to simplify Eq.~\ref{eq:rho out naive} by replacing $s'$ with 
$\soutput$~\eqref{eq_s_out} as a sum variable, giving
\begin{align}
   \rho_{\text{out}} &= \sum\limits_{\Pauli*{\ell'} \in \PauliSet*{L}} \sum\limits_{\Pauli*{s} \in \PauliSet*{S}} 
  \sum\limits_{\Pauli*{a} \in \PauliSet*{A}} \sum\limits_{\Pauli*{o} \in \PauliSet*{O}} 
  \mu^{(\ErrChan)}_{\Pauli*{\ell'} \otimes \Pauli*{s} \otimes \Pauli*{a} \otimes \Pauli*{o}} \sum\limits_{\soutput,\theta'} 
  \mathcal{I}^{\sinput + \soutput + s_* + D\left( E(\sinput)+ \theta'\right)}_{\Pauli*{s}} \, 
  \mathcal{I}^{\theta'}_{\Pauli*{0}} 
  \notag \\
  &~~\quad ~~~~~~\Cor_{E(\sinput)+\theta'} \circ \PauliChan*{\ell'} (\rho_L) \otimes \ZProj{\soutput} \, , ~~
  \label{eq:rho out nicer}
\end{align}
where we used the fact that $\serr = s_* + D\left( E(\sinput)+\theta' \right)$ to write $\Cor^{\vpp}_{\serr}$ as 
$\Cor^{\vpp}_{E(\sinput)+\theta'}$.

Finally, we combine the two logical Pauli channels into one, so that
\begin{align*}
  \rho_{\text{out}} &= \sum\limits_{\ell,\soutput} p^{\vpp}_{\ell,\soutput} \PauliChan{\ell} (\rho_L) \otimes \ZProj{\soutput} \, ,~~
  \tag{\ref{eq:rho out nice}}
\end{align*}
where we have implicitly defind the probability
\begin{align}
  p^{\vpp}_{\ell,\soutput} &\coloneqq \sum\limits_{\Pauli*{\ell'} \in \PauliSet*{L}} \sum\limits_{\Pauli*{s} \in \PauliSet*{S}} 
  \sum\limits_{\Pauli*{a} \in \PauliSet*{A}} \sum\limits_{\Pauli*{o} \in \PauliSet*{O}} 
  \mu^{(\ErrChan)}_{\Pauli*{\ell'} \otimes \Pauli*{s} \otimes \Pauli*{a} \otimes \Pauli*{o}} \notag \\
  &~~~\sum\limits_{\theta \in \Ints_2^{\abs{O}}} \mathcal{I}^{\sinput + \soutput + s_* + D\left( E(\sinput)+ \theta'\right)}_{\Pauli*{s}} \, 
  \mathcal{I}^{\theta'}_{\Pauli*{0}} \, \delta \big( \PauliChan{\ell} , \Cor_{E(\sinput)+\theta'} \circ \PauliChan*{\ell'}) 
  \label{eq:QEC nice probabilities} \, , ~~
\end{align}
where $\delta(\PauliChan{\ell} , \PauliChan{\ell'})$ evaluates to unity 
when the Pauli superoperators $\PauliChan{\ell}$ and $\PauliChan{\ell'}$ are equal,
and vanishes otherwise. The indicator function $\mathcal{I}$ is defined in Eq.~\ref{eq:Pauli bitstring indicator}.

\subsection{Proof of Prop.~\ref{prop_qec_final}.}
\label{app_gamma_and_gamma}

\begin{proof}
  The proof proceeds via induction. First, consider the base case. 
  The input state to the first QEC cycle is the output of noisy state preparation,
  \begin{equation*}
    \tag{\ref{eq_noisy_prep}}
    \rho^{\vpp}_0 = \sum_{s_0}
    \gamma_{s_0}^{(\text{prep})} \,  \Lambda_{s_0}^{(\text{prep})} (\sigma_L) \otimes \ZProj{s_0} \, ,~~
  \end{equation*}
  on the registers $L$ and $S$, and the output state of the first QEC cycle is%
  \begin{align}
    \rho^{\vpp}_{1} &= \sum_{s_0} \gamma_{s_0}^{(\text{prep})} \sum_{\ell,s',\theta'} \sum\limits_{\Pauli*{s},\Pauli*{a},\Pauli*{o}} 
    \mu^{(\ErrChan_1)}_{\Pauli*{\ell} \otimes \Pauli*{s} \otimes \Pauli*{a} \otimes \Pauli*{o}}
    \, \mathcal{I}^{s'}_{\Pauli*{s}} \, \mathcal{I}^{\theta'}_{\Pauli*{o}} \notag \\
    &~~~~~~ \Cor_{E(s_0)+\theta'} \circ \PauliChan{\ell} \circ \Lambda_{s_0}^{(\text{prep})} (\sigma_L) 
    \otimes \ZProj{s_0 + s' + s_* + D( E(s_0)+\theta')}
    \label{eq:QEC1 output full} \, ,~~
  \end{align}
  where we label $\Cor$ as in App.~\ref{app_QEC_update_2}. Also as in App.~\ref{app_QEC_update_2}, 
  we simplify the foregoing expression by defining $s_1$, in analogy to $\soutput$~\eqref{eq_s_out}. 
  The resulting density matrix is
  \begin{align}
    \rho^{\vpp}_{1} &= \sum_{s_0,s_1} \gamma_{s_0}^{(\text{prep})}  \sum\limits_{\Pauli*{\ell},\Pauli*{s},\Pauli*{a},\Pauli*{o}} 
    \mu^{(\ErrChan_1)}_{\Pauli*{\ell} \otimes \Pauli*{s} \otimes \Pauli*{a} \otimes \Pauli*{o}}
    \sum_{\theta'} \mathcal{I}^{s_0 + s_1 + s_* + D\left( E(s_0)+ \theta'\right)}_{\Pauli*{s}} \, 
    \mathcal{I}^{\theta'}_{\Pauli*{0}}  \notag \\
    &~~~~~~~~~ \Cor_{E(s_0)+\theta'} \circ \PauliChan{\ell} \circ \Lambda_{s_0}^{(\text{prep})} (\sigma_L) 
    \otimes \ZProj{s_1}
    \label{eq:QEC output nice} \, .~~
  \end{align}
  Next, we define the Markov process $\gamma^{(1)}_{s_1,s_0}$ on $S$:
  \begin{align}
    \gamma^{(1)}_{s_1,s_0} &\coloneqq \sum\limits_{\Pauli*{\ell},\Pauli*{s},\Pauli*{a},\Pauli*{o}} 
    \mu^{(\ErrChan_1)}_{\Pauli*{\ell} \otimes \Pauli*{s} \otimes \Pauli*{a} \otimes \Pauli*{o}}
    \sum_{\theta'} \mathcal{I}^{s_0 + s_1 + s_* + D\left( E(s_0)+ \theta'\right)}_{\Pauli*{s}} \, 
    \mathcal{I}^{\theta'}_{\Pauli*{0}}  
    \label{eq:QEC1 gamma def} \, ,~~
  \end{align}
  and we also define the following logical channel for all $\rho_L$:
  \begin{align}
    \Lambda^{(1)}_{s_1,s_0} (\rho_L) &\coloneqq \frac{1}{\gamma^{(1)}_{s_1,s_0}} \sum\limits_{\Pauli*{\ell},\Pauli*{s},\Pauli*{a},\Pauli*{o}} \hspace{-1mm}
    \mu^{(\ErrChan_1)}_{\Pauli*{\ell} \otimes \Pauli*{s} \otimes \Pauli*{a} \otimes \Pauli*{o}}
    \sum_{\theta'} \mathcal{I}^{s_0 + s_1 + s_* + D\left( E(s_0)+ \theta'\right)}_{\Pauli*{s}} \, 
    \mathcal{I}^{\theta'}_{\Pauli*{0}} \notag \\
    &~~~~~~~\times \Cor_{E(s_0)+\theta'} \circ \PauliChan{\ell} (\rho_L)
    \label{eq:QEC1 Lambda def} \, ,
  \end{align}
  when $\gamma^{(1)}_{s_1,s_0} \neq 0$, with $\Lambda^{(1)}_{s_1,s_0}=0$ otherwise. With these definitions,
  \begin{align}
    \label{eq:QEC1 output state final}
    \rho_1 = \sum\limits_{s_0,s_1} \gamma^{(1)}_{s_1,s_0} \gamma_{s_0}^{(\text{prep})} ~ \Lambda^{(1)}_{s_1,s_0} \circ \Lambda_{s_0}^{(\text{prep})} (\sigma_L) \otimes \ZProj{s_1} \, ,~~
  \end{align}
  is the state output by the first QEC cycle given the input state $\rho_0$~\eqref{eq_noisy_prep}.

  Next, we show that $\cycind \implies \cycind + 1$. Suppose the state after QEC cycle $\cycind-1$ is
  \begin{align}
    \label{eq:QECk Ansatz output}
    \rho_{\cycind-1} = \sum\limits_{s_0,\dots,s_{\cycind-1}} \gamma^{(\cycind-1)}_{s_{\cycind-1},s_{\cycind-2}} \cdots \gamma_{s_0}^{(\text{prep})} ~ \Lambda^{(\cycind-1)}_{s_{\cycind-1},s_{\cycind-2}} \circ \cdots \circ \Lambda_{s_0}^{(\text{prep})} (\sigma_L) \otimes \ZProj{s_{\cycind-1}} \, ,~~
  \end{align}
  which is the input state to the $\cycind$th QEC cycle. Following the $\cycind$th QEC cycle, the state is
  \begin{align}
    \rho^{\vpp}_{\cycind} &= \sum_{s_0,\dots,s_{\cycind}} \gamma_{s_{\cycind-1},s_{\cycind-2}}^{(\cycind-1)} \cdots \gamma_{s_0}^{(\text{prep})}  
    \sum\limits_{\Pauli*{\ell},\Pauli*{s},\Pauli*{a},\Pauli*{o}} 
    \mu^{(\ErrChan_1)}_{\Pauli*{\ell} \otimes \Pauli*{s} \otimes \Pauli*{a} \otimes \Pauli*{o}}
    \sum_{\theta'} \mathcal{I}^{s_{\cycind-1}+s_{\cycind}+s_*+D\left( E (s_{\cycind-1})+\theta'\right)}_{\Pauli*{s}} \, \mathcal{I}^{\, \theta'}_{\Pauli*{o}} \notag \\
    &~~~~~~~~~ \Cor_{E(s_{\cycind-1})+\theta'} \circ \PauliChan{\ell} \circ \Lambda_{s_{\cycind-1},s_{\cycind-2}}^{(\cycind-1)} \circ \cdots \circ \Lambda_{s_0}^{(\text{prep})} (\sigma_L) 
    \otimes \ZProj{s_1}
    \label{eq:QEC1 output nicer} \, .~~
  \end{align}
  We next define the associated Markov process $\gamma^{(\cycind)}_{s_{\cycind},s_{\cycind-1}}$ on $S$ as 
  \begin{align}
    \gamma^{(\cycind)}_{s_{\cycind},s_{\cycind-1}} &\coloneqq \sum\limits_{\Pauli*{\ell} \in \PauliSet*{L}} 
    \sum\limits_{\Pauli*{s} \in \PauliSet*{S}} \sum\limits_{\Pauli*{a} \in \PauliSet*{A}} \sum\limits_{\Pauli*{o} \in \PauliSet*{O}} 
    \mu^{(\ErrChan_{\cycind})}_{\Pauli*{\ell} \otimes \Pauli*{s} \otimes \Pauli*{a} \otimes \Pauli*{o}} \sum_{\theta'} 
    \mathcal{I}^{s_{\cycind-1}+s_{\cycind}+s_*+D\left( E (s_{\cycind-1})+\theta'\right)}_{\Pauli*{s}} \, \mathcal{I}^{\, \theta'}_{\Pauli*{o}}
    \label{eq:QEC gamma def} \, ,~~
  \end{align}
  and we define the corresponding logical channel $\Lambda^{(\cycind)}_{s_{\cycind},s_{\cycind-1}} (\rho_L)$ for all $\rho_L$ as
  \begin{align}
    \Lambda^{(\cycind)}_{s_{\cycind},s_{\cycind-1}} (\rho_L) &\coloneqq \frac{1}{\gamma^{(\cycind)}_{s_{\cycind},s_{\cycind-1}}} \, 
    \sum\limits_{\Pauli*{\ell} \in \PauliSet*{L}} \sum\limits_{\Pauli*{s} \in \PauliSet*{S}} \sum\limits_{\Pauli*{a} \in \PauliSet*{A}} \sum\limits_{\Pauli*{o} \in \PauliSet*{O}}
    \mu^{(\ErrChan_{\cycind})}_{\Pauli*{\ell} \otimes \Pauli*{s} \otimes \Pauli*{a} \otimes \Pauli*{o}} \notag \\
    &~~~~~\sum_{\theta'} \mathcal{I}^{s_{\cycind-1}+s_{\cycind}+s_*+D\left( E (s_{\cycind-1})+\theta'\right)}_{\Pauli*{s}} \, \mathcal{I}^{\, \theta'}_{\Pauli*{o}}
    \Cor_{E(s_{\ncyc-1})+\theta'} \circ \PauliChan{\ell}  (\rho_L)
    \label{eq:QEC Lambda def} \, ,
  \end{align}
  when $\gamma^{(\cycind)}_{s_{\cycind},s_{\cycind-1}} \neq 0$, with $\Lambda^{(\cycind)}_{s_{\cycind},s_{\cycind-1}} (\rho_L) = 0$ otherwise.
  for all $\rho_L$.  Using the above, we write
  \begin{align*}
    \tag{\ref{eq_qec_final}}
    \rho^{\vpp}_{\cycind} = \sum\limits_{s_0,\dots,s_{\cycind}} 
    \gamma^{(\cycind)}_{s_{\cycind},s_{\cycind-1}} \dots \gamma^{(1)}_{s_{1},s_{0}}
     \gamma^{(\text{prep})}_{s_0} ~
    \Lambda^{(\cycind)}_{s_{\cycind},s_{\cycind-1}} \circ \dots \circ \Lambda^{(1)}_{s_{1},s_{0}} 
    \circ \Lambda^{(\text{prep})}_{s_0} (\sigma_L) \otimes \ZProj{s_{\cycind}} \, , ~~
  \end{align*}
  for all $\cycind \geq 1$. This also establishes the update for a single QEC cycle applied to a 
  product input state $\rho_{\text{in}} = \rho_L \otimes \ZProj{\sinput}$~\eqref{eq:QEC input state}, 
  as completing the proof.
\end{proof}

\subsection{Logical measurement}
\label{app_logical_meas}

We consider destructive stabilizer measurements of $L$ captured by the Clifford unitary $C_{\text{meas}}$ in Diagram~\ref{circ_qec_meas_full}.
Such measurements are equivalent---up to a logical Clifford---to computational-basis measurements of $L$;
the corresponding outcomes are enumerated by the bitstrings $x$ of the subsystem $L$. For convenience, 
we now restrict to computational-basis measurements, where the Clifford unitary 
$C_{\text{meas}}$ acts on product states of the four registers as
\begin{align}
  \label{eq_C_POVM}
  C_{\text{meas}} \ket{x_L} \otimes \ket{\soutput} \otimes \ket{a} \otimes \ket{0} = \ket{x_L} \otimes \ket{\soutput} \otimes \ket{b} \otimes \ket{E_{\POVM}(x_L,\soutput)} \, ,~~
\end{align}
where $a$ and $b$ are arbitrary computational-basis states of $A$,
the initial reference state $0$ of $O$ is similarly arbitrary,
$E_{\POVM}$ is a function that represents the encoding of the outcomes of measuring $S$ and $L$, typically in some classical error-correcting code,
and $x_L$ and $\soutput$ are bitstrings of the systems $L$ and $S$, respectively. 
The function $E_{\POVM}$ is specified implicitly as part of the unitary $C_{\text{meas}}$
, and maps the bitstrings $x_L$ and $\soutput$ to a bitstring $\vartheta$ of system $O$. 
The measurement outcomes are extracted via the function $D_{\POVM}$ that decodes $\vartheta$. 

As noted in Sec.~\ref{sec_Pauli_noise_meas}, Pauli-stochastic noise is defined to mean that
the noisy implementation of $C^{\vpp}_{\text{meas}}$~\eqref{eq_C_POVM} can be written as
$C_{\text{meas}}' = \ErrChan^{\vpp}_{\text{meas}} \circ C_{\text{meas}}^{\vpp}$, where $\ErrChan^{\vpp}_{\text{meas}}$ is a Pauli channel~\eqref{eq:general Pauli channel}.
In practice, $\ErrChan^{\vpp}_{\text{meas}}$ can be converted into a Pauli channel via Pauli twirling of $C^{\vpp}_{\text{meas}}$. 

We now derive the measurement channel $\Lambda_{\soutput}^{(\text{meas})}(\rho_L)$~\eqref{eq_Lambda_meas} using Eq.~\ref{eq_meas_err_state}.
To do so, we compute the probability $\operatorname{Prob}(\elt)$ associated with the logical bitstring $x_{\elt}$ (and POVM element $\POVM_{\elt}$). 
The full state of the system following the application of $C_{\text{meas}}'$ is
\begin{align}
  \varrho_{\POVM} &= \sum\limits_{\Pauli*{\ell'} \in \PauliSet*{L}} \sum\limits_{\Pauli*{s} \in \PauliSet*{S}} \sum\limits_{\Pauli*{a} \in \PauliSet*{A}} \sum_{\Pauli*{o} \in \PauliSet*{O}}
  \mu^{(\text{meas})}_{\Pauli*{\ell'} \otimes \Pauli*{s} \otimes \Pauli*{a} \otimes \Pauli*{o}} \sum\limits_{\vartheta'} \mathcal{I}^{\vartheta'}_{\Pauli*{o}} \, \times \notag \\
  &~~~~ \PauliChan{\ell'} (\ZProj{x_L}) \otimes \PauliChan{s} ( \ZProj{\soutput}) \otimes \PauliChan{a} ( \ZProj{b}) \otimes \ZProj{E_{\POVM} (x_L,\soutput)+\vartheta'} 
  \label{eq_full_postmeas_state}\, ,~~
\end{align}
where $\vartheta'$ runs over bitstrings of $O$, $\mu^{(\text{meas})}_{P}$ is the probability of the Pauli $P \in \PauliSet*{\text{all}}$ 
in the mixture $\ErrChan^{\vpp}_{\text{meas}}$, and $\mathcal{I}$ is defined in Eq.~\ref{eq:Pauli bitstring indicator}.
At this point, the outcomes are stored in $O$, and the other registers may be discarded.
The reduced state on $O$ is
\begin{align}
  \label{eq_rho_data}
  \rho_{\text{data}} &= \sum\limits_{\vartheta'} \overline{p}^{(\text{meas})}_{x_L,\soutput,\vartheta'} ~ \ZProj{E_{\POVM} ( x_L,\soutput)+\vartheta'} \, ,~~
\end{align}
where we have implicitly defined the probabilities
\begin{align}
  \label{eq_intermediate_meas_probs}
  \overline{p}^{(\text{meas})}_{x_L,\soutput,\vartheta'} =  \sum\limits_{\Pauli*{\ell} \in \PauliSet*{L}} \sum\limits_{\Pauli*{s} \in \PauliSet*{S}} \sum\limits_{\Pauli*{a} \in \PauliSet*{A}} \sum_{\Pauli*{o} \in \PauliSet*{O}}
  \mu^{(\text{meas})}_{\Pauli*{\ell'} \otimes \Pauli*{s} \otimes \Pauli*{a} \otimes \Pauli*{o}} \mathcal{I}^{\vartheta'}_{\Pauli*{o}}  \, ,~~
\end{align}
where the labels $x_L$ and $\soutput$ are fixed by the initial input state in Eq.~\ref{eq_C_POVM},
which is also an eigenstate in the measurement basis. Considering the analogous expression for density matrices,
we note that only incoherent superpositions in the measurement 
basis---i.e., terms in $\rho_L$ of the form $\ZProj{x}$---contribute 
to $\operatorname{Prob}(\elt)$~\eqref{eq_meas_err_state}. 

Next, we evaluate the probability of recovering the outcome $x_{\elt} = x_L + x'$,
\begin{align}
  \label{eq_meas_prob_actual}
  \operatorname{Prob}(x_L + x') &= \sum\limits_{\vartheta'} \overline{p}^{(\text{meas})}_{x_L,\soutput,\vartheta'} \, 
  \delta \Big( D_{\POVM} \big( E_{\POVM} ( x_L,\soutput)) + \vartheta' \big) , x_L + x' \Big) \, ,~~
\end{align}
for fixed $x_L$ and $\soutput$, where $\delta ( a,b)$ is one if $a=b$ and vanishes otherwise.
Note that the expression for $\operatorname{Prob}(\elt)$~\eqref{eq_meas_err_state} recovers upon taking $x_L + x' \to x_{\elt}$.

Suppose that the joint state of $L$ and $S$ prior to measurement takes the product form $\rho = \ZProj{x_L} \otimes \ZProj{\soutput}$.
Our task is to identify a logical channel $\Lambda_{\soutput}^{(\text{meas})}$ satisfying
\begin{align}
  \operatorname{Prob}(\elt) &= \tr{\POVM_{\elt} \, \Lambda_{\soutput}^{(\text{meas})}(\ZProj{x_L})}  =  
  \sum\limits_{\vartheta'} \overline{p}^{(\text{meas})}_{x_L,\soutput,\vartheta'} \, 
  \delta \Big( D_{\POVM} \big( E_{\POVM} ( x_L,\soutput)) + \vartheta' \big) , x_{\elt} \Big)\, . ~~
  \label{eq_meas_relation}
\end{align}
This is fulfilled by the choice
\begin{equation*}
  \tag{\ref{eq_Lambda_meas}}
  \Lambda^{(\text{meas})}_{\soutput} 
  ( \rho_L ) = \sum\limits_{\Pauli*{\ell} \in \PauliSet*{L}} p^{(\text{meas})}_{\ell,\soutput} \, \PauliChan{\ell} 
  ( \rho_L ) \, ,~~
\end{equation*}
for any computational-basis state $\rho_L = \ZProj{x_L}$, where we define
\begin{align}
  \label{eq_meas_Lambda_probs}
  p^{(\text{meas})}_{\ell,\soutput} &= \sum\limits_{x'} \mathcal{I}^{\, x'}_{\Pauli*{\ell}} \, 
  \sum\limits_{\vartheta'} \, \overline{p}^{(\text{meas})}_{x_L,\soutput,\vartheta'} 
\end{align}
where $x'$ runs over bitstrings of $L$ and $\overline{p}^{(\text{meas})}_{x_L,\soutput,\vartheta'}$ is defined in Eq.~\ref{eq_meas_prob_actual}.

Although we restricted our consideration to joint states of $L$ and $S$ of the form $\ZProj{x_L} \otimes \ZProj{\soutput}$, 
we now explain why the choice of $\Lambda^{(\text{meas})}_{\soutput}$ in Eq.~\ref{eq_Lambda_meas} holds for generic $\rho_L$.
Importantly, the action of $\Lambda^{(\text{meas})}_{\soutput}$ on any logical states 
that are dephased in the measurement (computational) basis follows directly from Eqs.~\ref{eq_meas_prob_actual} and Eq.~\ref{eq_meas_relation}.
This is because any such state is a linear combination of states of the form $\ZProj{x_L} \otimes \ZProj{\soutput}$.
Next, we must consider the off-diagonal elements of $\rho_L$ in the computational basis. 
According to Eq.~\ref{eq_meas_prob_actual}, 
the contribution of these terms to $\operatorname{Prob}(\elt)$~\eqref{eq_meas_err_state} must vanish.
We then observe that $\Lambda^{(\text{meas})}_{\soutput}$~\eqref{eq_Lambda_meas} maps off-diagonal terms in $\rho_L$
to off-diagonal terms, which are then annihilated upon taking the trace with $\POVM_{\elt}$ in Eq.~\ref{eq_meas_relation}. 
Thus, Eq.~\ref{eq_Lambda_meas} correctly reproduces the probabiltiies $\operatorname{Prob}(\elt)$ associated with 
the noisy measurement of $\POVM$ for any logical input state.

\section{Syndrome marginal independence}
\label{app_smip_proofs}

\subsection{Proof of Prop.~\ref{prop_ssp_smip}}
\label{app_ssp_smip}

\begin{proof}
  The existence of the Markov process $\gamma^{\,}_{\soutput,\sinput}$ is guaranteed by Prop.~\ref{prop_qec_final}, with 
  $\gamma^{\,}_{\soutput,\sinput}$ defined in Eq.~\ref{eq:QEC gamma def}. 
  Crucially, $\gamma^{\,}_{\soutput,\sinput}$ depends on $\sinput$ only through the term
  \begin{equation*}
    \mathcal{I}^{\sinput + \soutput + s_*+D\left( E (\sinput)+\theta'\right)}_{\Pauli*{s}} \, .~~
  \end{equation*}
  When the \ssp{} holds, we have that $D \left( E (\sinput) + \theta' \right) = \sinput + D(\theta')$; 
  because $\sinput + \sinput = 0$, $\gamma^{\,}_{\soutput,\sinput}$ does not depend on $\sinput$, 
  and we write $\gamma^{\,}_{\soutput,\sinput} \to \gamma^{\,}_{\soutput}$. 
\end{proof}

\subsection{Proof of Prop.~\ref{prop_syndrome_rand_smip}}
\label{app_syndrome_rand_proof}

\begin{proof}
  As in App.~\ref{app_QEC_update_C}, consider a joint input state on $L$ and $S$ of the form 
  $\rho_L \otimes \ZProj{\sinput}$~\eqref{eq:QEC input state}. Applying syndrome randomization (see Def.~\ref{defn_rand_synd})
  and initializing the registers $A$ and $O$ in their default states leads to the initial state
  \begin{equation}
    \label{eq_rand_init_state}
    \varrho_{\text{in}} = \frac{1}{2^{\nphys-\nlog}} \sum\limits_{s_r} \rho_L \otimes \ZProj{\sinput + s_r} \otimes \ZProj{a} \otimes \ZProj{0} \, .~~
  \end{equation}
  Applying the ideal unitary $C$~\eqref{eq_c_unitary_ideal} to this state results in
  \begin{equation}
    \varrho' = \frac{1}{2^{\nphys-\nlog}} \sum\limits_{s_r} \rho_L \otimes \ZProj{\sinput + s_r} 
    \otimes \ZProj{a} \otimes \ZProj{E(\sinput + s_r)} \, ,~~
  \end{equation}
  and applying the noisy channel $\ErrChan$~\eqref{eq:noisy C} associated with $C$ and discarding 
  (i.e., tracing out and resetting) the state of $A$ results in the joint state
  \begin{align}
    \varrho' &= \frac{1}{2^{\nphys-\nlog}}  \sum\limits_{\Pauli*{\ell'} \in \PauliSet*{L}} \sum\limits_{s_r} \sum\limits_{s'} \sum\limits_{\theta'}
     \overline{p}^{\vpp}_{\ell',s',\theta'} \, \PauliChan{\ell'} (\rho_L) \otimes \ZProj{\sinput + s_r + s'} \otimes \ZProj{E(\sinput + s_r) + \theta'} \, ,~~
   \end{align}
   on $L$, $S$, and $O$, where $\overline{p}^{\vpp}_{\ell',s',\theta'} $ is unchanged by syndrome randomization (see Eq.~\ref{eq:QEC intermediate probabilities}).

   Next, we perform classical side processing (Step~\ref{qec_step_proc}) to determine the error syndrome $\serr$. 
   This is used to reset the syndrome subsystem $S$ (Step~\ref{qec_step_reset}), resulting in the joint state
   \begin{align}
    \varrho' &= \frac{1}{2^{\nphys-\nlog}}  \sum\limits_{\Pauli*{\ell'} \in \PauliSet*{L}} \sum\limits_{s_r} \sum\limits_{s'} \sum\limits_{\theta'}
     \overline{p}^{\vpp}_{\ell',s',\theta'} \, \PauliChan{\ell'} (\rho_L) \notag \\
     &~~~~\otimes \ZProj{\sinput + s_r + s'+s_* + D(E(\sinput + s_r) + \theta')} \otimes \ZProj{E(\sinput + s_r) + \theta'} \, ,~~
   \end{align}
   where the syndrome reset is conditioned on configuration $D\left( E (\sinput + s_r) + \theta'\right)$ inferred from system $O$, 
   which includes the effects of randomization. Finally, we apply the logical recovery operation EC~\eqref{eq:EC step} in Step~\ref{qec_step_EC},
   where we condition $\Cor$~\eqref{eq:EC step} on $\serr + s_r$ so that syndrome randomization does not introduce logical errors. 
   The resulting state is
   \begin{align}
    \rho_{\text{out}} &= \frac{1}{2^{\nphys-\nlog}}  \sum\limits_{\Pauli*{\ell'} \in \PauliSet*{L}} \sum\limits_{s_r} \sum\limits_{s'} \sum\limits_{\theta'}
     \overline{p}^{\vpp}_{\ell',s',\theta'} \, \Cor^{\vpd}_{s_r + D(E(\sinput + s_r) + \theta')} \circ \PauliChan{\ell'} (\rho_L) \notag \\
     &~~~~~\otimes \ZProj{\sinput + s_r + s'+s_* + D(E(\sinput + s_r) + \theta')} \, , ~~
   \end{align}
   after tracing out and reseting subsystem $O$. Finally, we define the syndrome bitstring $\sigma = \sinput + s_r$, 
   and change summation variables from $s_r$ to $\sigma$, so that 
   \begin{align}
    \rho_{\text{out}} &= \frac{1}{2^{\nphys-\nlog}}  \sum\limits_{\Pauli*{\ell'} \in \PauliSet*{L}} \sum\limits_{\sigma} \sum\limits_{s'} \sum\limits_{\theta'}
     \overline{p}^{\vpp}_{\ell',s',\theta'} \, \Cor^{\vpd}_{\sinput + \sigma + D(E(\sigma) + \theta')} \circ \PauliChan{\ell'} (\rho_L) \notag \\
     &~~~~~\otimes \ZProj{\sigma + s'+s_* + D(E(\sigma) + \theta')} \, ,~~
     \label{eq_rand_output_state}
   \end{align}
  and we observe that the distribution over $\soutput$ is the same for any input state $\sinput$,
  and the only dependence on $\sinput$ is in $\Cor$, which only affects $\Lambda$~\eqref{eq:QEC Lambda def}. We find that
  \begin{align}
    \Lambda^{(\text{rand})}_{\soutput,\sinput} (\rho_L) &\coloneqq \frac{1}{\gamma^{(\text{rand})}_{\soutput,\sinput}} \, 
    \frac{1}{2^{\nphys-\nlog}} \sum\limits_{\sigma} \sum\limits_{\Pauli*{\ell} \in \PauliSet*{L}} \sum\limits_{\Pauli*{s} \in \PauliSet*{S}} 
    \sum\limits_{\Pauli*{a} \in \PauliSet*{A}} \sum\limits_{\Pauli*{o} \in \PauliSet*{O}} \mu^{(\ErrChan_{\cycind})}_{\Pauli*{\ell} \otimes \Pauli*{s} \otimes \Pauli*{a} \otimes \Pauli*{o}} \notag \\
    &~~\sum_{\theta'} \mathcal{I}^{\soutput + \sigma + s_* + D\left( E (\sigma)+\theta'\right)}_{\Pauli*{s}} \, \mathcal{I}^{\, \theta'}_{\Pauli*{o}}
    ~\Cor_{\sinput + \sigma + D\left( E(\sigma)+\theta'\right)} \circ \PauliChan{\ell}  (\rho_L)
    \label{eq_syn_randomized_Lambda} \, ,
  \end{align} 
  where the Markov process $\gamma^{(\text{rand})}$~\eqref{eq:QEC gamma def} with syndrome randomization is given by
   \begin{align}
    \gamma^{(\text{rand})}_{\soutput,\sinput} = \frac{1}{2^{\nphys-\nlog}} \sum\limits_{\sigma} \sum\limits_{\Pauli*{\ell} \in \PauliSet*{L}} 
    \sum\limits_{\Pauli*{s} \in \PauliSet*{S}} \sum\limits_{\Pauli*{a} \in \PauliSet*{A}} \sum\limits_{\Pauli*{o} \in \PauliSet*{O}} 
    \mu^{(\ErrChan_{\cycind})}_{\Pauli*{\ell} \otimes \Pauli*{s} \otimes \Pauli*{a} \otimes \Pauli*{o}} \sum_{\theta'} 
    \mathcal{I}^{\soutput + \sigma + s_* + D\left( E (\sigma)+\theta'\right)}_{\Pauli*{s}} \, \mathcal{I}^{\, \theta'}_{\Pauli*{o}} \, ,~~
    \label{eq_syn_randomized_gamma}
  \end{align}
  which is independent of $\sinput$, so that the \smip{} holds. 
\end{proof}

\section{Results used in the proof of Theorem~\ref{thm_main}}
\label{app_lemmas}

\subsection{Lemma \ref{prop_perturbation_norm_bound}}
\label{app_prop_perturb_norm_bound}

\begin{lem}
  Let $\sigma$ be a finite-dimensional vector of length $D$ whose components realize a discrete probability distribution
  and let $\Sigma$ and $E$ be real $D \times D$ matrices where 
  $\Sigma = \operatorname{diag}(\sigma)$ and $E_{ij} \in [0,1]$.
  Let $\norm{x}^{\,}_{\mathrm{op}} $ be the \emph{operator norm} of an operator (or matrix) $x$.
  If $\sigma^T E \sigma \leq \varepsilon$ for some real $\varepsilon \geq 0$, 
  then $\norm{\ssig E \ssig}_{\mathrm{op}} \leq \varepsilon^{1/2}$.
  \label{prop_perturbation_norm_bound}
\end{lem}
\begin{proof}
  A convenient expression for the operator norm is
  \begin{equation}
     \norm{x}^{\,}_{\mathrm{op}} = \left( \sup \operatorname{spec} ( x^\dagger x) \right)^{1/2} \, ,~~
     \label{eq_op_norm}
  \end{equation}
  which is also known as the spectral radius of $x$. The Frobenius norm of $x$ is given by
  \begin{equation}
    \label{eq_Frobenius_norm}
    \norm{x}^{\vpp}_{\mathrm{F}} \coloneqq  \trace \left( x^\dagger x \right)^{1/2} \, ,~~
  \end{equation}
  which is the sum over the entire spectrum of $x^\dagger x$. 
  Because $x^\dagger x \geq 0$ is a positive operator for any $x$, its spectrum is positive semi-definite. 
  Hence, $\norm{x}_{\mathrm{F}}^2 \geq \norm{x}_{\mathrm{op}}^2$, and so the Frobenius norm upper bounds the operator norm.
  Next, we compute the Frobenius norm of $\ssig E \ssig$, finding
  \begin{equation}
    \label{eq_Fro_norm_ssig}
    \norm{\ssig E \ssig}^{\vpp}_{F} = \sqrt{\sum_{i,j} \sigma_i \sigma_{i} E^2_{ij}} \, .~~
  \end{equation}
  Now, because all entries of $E$ lie in the interval $[0,1]$, we have that $E_{ij}^2 \leq E_{ij}$ for any $i$ and $j$.
  Accordingly, we have that $\sum_{i,j} \sigma_i \sigma_j E^2_{ij} \leq \sum_{i,j} \sigma_i \sigma_j E^{\,}_{ij} = \sigma^T E \sigma$.
  As a result,  $\sigma^T E \sigma$  upper bounds the square of the Frobenius norm of $\ssig E \ssig$ (see Eq.~\ref{eq_Fro_norm_ssig}), 
  which in turn upper bounds the operator norm of $\ssig E \ssig$. Taking square roots completes the proof.
\end{proof}

\subsection{Lemmas for Pauli-diagonal superoperators}
\label{app_lemma_pauli_diag}

We make several assumptions that enter all of the Lemmas below. 
Let $\hilbert$ be an $\nphys$-qubit Hilbert space with $\operatorname{dim}(\hilbert)=2^{\nphys}$. 
Let $\operatorname{End}(\hilbert)$ denote the Hilbert space of linear operators on $\hilbert$, 
equipped with the Hilbert-Schmidt inner product $\expval{A,B} = 2^{-\nphys} \trace ( A^\dagger B)$ 
for all $A,B \in \operatorname{End}(\hilbert)$.
For all $A \in \operatorname{End}(\hilbert)$, we denote by $\norm{A}_{\mathrm{op}}$ 
the standard operator norm of $A$---i.e., the smallest number $c>0$ such that $\norm{Ax} \leq c \, \norm{x}$ for 
all $x \in \hilbert$, where $\norm{x}=\expval{x,x}^{1/2}$.
We denote by $\norm{A}_{\mathrm{HS}} = \expval{A,A}^{1/2}$ the Hilbert-Schmidt norm of $A$, 
which is the norm of $A$ as an element of $\operatorname{End}(\hilbert)$.
Similarly, if $\Phi$ is a \emph{super}operator on $\hilbert$---i.e., a linear operator acting on $\operatorname{End}(\hilbert)$---we 
denote by $\norm{\Phi}_{\mathrm{op}}$ the operator norm of $\Phi$---i.e., the smallest number $c>0$ such 
that $\norm{\Phi(A)}_{\mathrm{HS}} \leq c\, \norm{A}_{\mathrm{HS}}$ for all $A \in \operatorname{End}(\hilbert)$. 
Finally, we note that the operator Hilbert space $\operatorname{End}(\hilbert)$ is spanned 
by the Pauli operators $P \in \PauliSet*{\nphys}$~\eqref{eq:Pauli set}.
We begin by defining Pauli-diagonal superoperators.

\begin{defn}[Pauli-diagonal superoperator]
  A superoperator $\Phi$ on $\hilbert$ is \emph{Pauli diagonal} if there exists a family 
  $\{ \alpha_{P} \}_P \subset \Comps$ such that $\Phi(P) = \alpha_P \, P$ for all $P \in \PauliSet*{\nphys}$.
  \label{defn_Pauli_diagonal}
\end{defn}

All Pauli superoperators $\PauliChan*{i}$~\eqref{eq:Pauli superoperator}---and thus, 
all Pauli channels $\Lambda$~\eqref{eq:general Pauli channel}---are Pauli diagonal. 
Below, we state and prove a useful equivalence between the maximum-magnitude Pauli eigenvalue of a Pauli-diagonal 
superoperator $\Phi$ and the operator norm of $\Phi$.
\begin{lem}
  \label{lem_norm_equiv}
  Let $\Phi$ be a Pauli-diagonal superoperator on $\hilbert$ 
  with associated Pauli eigenvalues $\{ \alpha^{\,}_P \}^{\,}_P$ (see Def.~\ref{defn_Pauli_diagonal}). 
  Then the norm of $\Phi$ is given by
  \begin{equation}
    \norm{\Phi}^{\vpp}_{\mathrm{op}} = \max_{P \in \PauliSet*{\nphys}} \, \abs{\alpha^{\vpp}_P} \, .~~
    \label{eq_norm_equiv}
  \end{equation}
\end{lem}
\begin{proof}
The Pauli operators form an orthogonal basis for $\operatorname{End}(\hilbert)$. 
They are an eigenbasis for $\Phi$ in the sense that $\Phi(P) = \alpha_P P$. 
Since the Paulis satisfy $\norm{P}_{\mathrm{HS}}=1$, $\norm{\Phi}_{\mathrm{op}}$ is simply the maximum over 
$P \in \PauliSet*{\nphys}$ of $\norm{\Phi(P)}_{\mathrm{HS}}$. This immediately leads to Eq.~\ref{eq_norm_equiv}.
\end{proof}

\begin{lem}
  Consider an $\nphys$-qubit system. Let $\rho$ be a density operator and let $\Pi$ be an operator with rank $r$
  such that both $\Pi$ and $\ident - \Pi$ are positive. Then, 
  for all Pauli-diagonal superoperators $\Phi$ with Pauli eigenvalues $\{ \alpha^{\,}_P\}^{\,}_{P}$ 
  (see Def.~\ref{defn_Pauli_diagonal}),
  \begin{equation}
    \abs{ \trace \left( \Pi^\dagger \, \Phi (\rho) \right) }^2  \leq 
    \abs{ \trace \left( \Pi \right)  }^2 \, \max\limits_{P \in \PauliSet*{\nphys}} \, \abs{\alpha^{\vpp}_P}^2 \leq 
    r \, 
    \max\limits_{P \in \PauliSet*{\nphys}} \, \abs{\alpha^{\vpp}_P}^2 \, .~~
    \label{eq_Pauli_diag_inequality}
  \end{equation}
  \label{lem_pauli_diag}
\end{lem}

\begin{proof}
  Applying the Cauchy-Schwarz inequality for the Hilbert space $\operatorname{End}(\hilbert)$ implies that 
  \begin{align}
    \abs{ \trace \left( \Pi^\dagger \, \Phi (\rho) \right) }^2  &= 4^{\nphys} \left| \expval{\Pi,\Phi(\rho)} \right|^2 \notag \\
    &\leq 4^{\nphys} \abs{ \expval{\Pi,\Pi} } \, \abs{ \expval{\Phi(\rho),\Phi(\rho)}} \notag \\
    &= 2^{\nphys} \trace \left( \Pi^2 \right) \, \norm{\Phi(\rho)}^2_{\mathrm{HS}} \, , ~~ 
    \label{eq_Pauli_diag_inequality_1} \\
    \intertext{where we used the definition of the Hilbert-Schmidt norm of the operator $\Phi(\rho) \in \operatorname{End}(\hilbert)$. 
    Next, we invoke the standard result that $\norm{Ax} \leq \norm{A} \, \norm{x}$; for superoperators, 
    the analogous inequality is $\norm{\Phi(\rho)}^{\,}_{\mathrm{HS}} \leq \norm{\Phi}^{\,}_{\mathrm{op}} \, \norm{\rho}^{\,}_{\mathrm{HS}}$.
    The resulting expression is}
    \abs{ \trace \left( \Pi^\dagger \, \Phi (\rho) \right) }^2  &\leq 2^{\nphys} \trace \left( \Pi^2 \right) \, 
    \norm{\rho}^{2}_{\mathrm{HS}} \, \norm{\Phi}^2_{\mathrm{op}} \notag \\
    &= \trace \left( \Pi^2 \right) \, \trace \left( \rho^\dagger \rho \right) \,
     \max\limits_{P \in \PauliSet*{\nphys}} \, \abs{\alpha^{\vpp}_P}^2 \, , ~~
     \label{eq_Pauli_diag_inequality_2}
  \end{align}
  where we used the fact that $\rho$ and $\Pi$ are Hermitian, 
  and we used Eq.~\ref{eq_norm_equiv} from Lemma~\ref{lem_norm_equiv} to recover the expression in terms of the Pauli eigenvalues.  
  We recover the first inequality in Eq.~\ref{eq_Pauli_diag_inequality} by 
  noting that $\trace (\rho^2) \leq \trace (\rho) = 1$, so that Eq.~\ref{eq_Pauli_diag_inequality_2} becomes
  \begin{equation}
    \abs{ \trace \left( \Pi^\dagger \, \Phi (\rho) \right) }^2  \leq \trace \left( \Pi^2 \right) 
    \, \max\limits_{P \in \PauliSet*{\nphys}} \, \abs{\alpha^{\vpp}_P}^2
    \, , ~~ \label{eq_Pauli_diag_inequality_first}
  \end{equation}
  and the second inequality in Eq.~\ref{eq_Pauli_diag_inequality} follows from the observation that
  any rank-$r$ operator $\Pi$ with $0 \leq \Pi \leq \ident$ satisfies
  $\trace (\Pi^2) \leq \trace (\Pi)  \leq  r$. 
  Applying this to Eq.~\ref{eq_Pauli_diag_inequality_first} leads to
  \begin{equation}
    \trace \left( \Pi^2 \right) \, \max\limits_{P \in \PauliSet*{\nphys}} \, \abs{\alpha^{\vpp}_P}^2 \leq 
    r \, \max\limits_{P \in \PauliSet*{\nphys}} \, \abs{\alpha^{\vpp}_P}^2
    \, . ~~ \label{eq_Pauli_diag_inequality_second}
  \end{equation}
\end{proof}

\begin{lem}
  Consider an $\nphys$-qubit system with dimension $D=2^{\nphys}$, a density operator $\rho$, 
  and a POVM element $\POVM_{\elt}$. Let $\Lambda$ be a Pauli channel~\eqref{eq:general Pauli channel} with 
  Pauli eigenvalues $\{\lambda^{\,}_P\}^{\,}_P$ and let $\Phi$ be a Pauli-diagonal superoperator 
  (see Def.~\ref{defn_Pauli_diagonal}) with Pauli eigenvalues $\{\lambda^{\,}_P + \epsilon^{\,}_P\}^{\,}_P$. Then 
  \begin{equation}
    \delta = \left| \trace \big( \POVM_{\elt} \, \Phi (\rho) \big) - \trace \big( \POVM_{\elt} \, \Lambda (\rho) \big) \right| \leq 
    D^{1/2} \, \max_P \left\{ \abs{\epsilon^{\,}_P} \right\}
    \label{eq_delta_inequality} \, .~~
  \end{equation}
  \label{lem_approx_model_eigenvalues}
\end{lem}
\begin{proof}
  Define the superoperator $\Delta = \Phi - \Lambda$. Because both $\Lambda$ and $\Phi$ are Pauli diagonal (see Def.~\ref{defn_Pauli_diagonal}),
  $\Delta$ is also Pauli diagonal, with Pauli eigenvalue $\epsilon^{\,}_P$ for the Pauli $P \in \PauliSet*{\nphys}$. 
  Since $\operatorname{rank}(\POVM_{\elt}) \leq D$, applying Lemma~\ref{lem_pauli_diag} to 
  $\delta = \abs{ \trace ( \POVM_{\elt} \, \Delta (\rho) )}$ leads to Eq.~\ref{eq_delta_inequality}.
\end{proof}

\subsection{Lemmas concerning matrix eigenvalues}
\label{app_matrix_lemmas}
For the lemmas below, let $V$ be a finite-dimensional inner-product space
and let $\ket{\psi} \in V$ be a unit vector.
Unless otherwise specified, all vectors are elements of $V$ and all operators and matrices act on $V$.
\begin{lem}
   Let $T$ be a real matrix 
  such that $T = \ketbra{\psi} + E$, where $\norm{E}^{\,}_{\mathrm{op}} \leq \epsilon < 1/2$.
  Then, $T$ has exactly one eigenvalue $\lambda$ satisfying $\abs{1-\lambda} \leq \epsilon$.
  Furthermore, $\lambda$ is real and is the largest-magnitude eigenvalue of $T$.
  \label{lem_perturbed_eigenvalue}
\end{lem}
\begin{proof}
  We first apply the Bauer-Fike theorem (see Equation 6.3.2 in Ref.~\citenum{HornJohnson2013}). 
  Given a normal matrix $A$ and a perturbation $E$ with $\norm{E}^{\,}_{\mathrm{op}} \leq \epsilon$, 
  the Bauer-Fike theorem states that, if $\lambda$ is an eigenvalue of $T = A + E$, 
  then there exists an eigenvalue $a$ of $A$ such that $\abs{\lambda-a} \leq \epsilon$. 
  Because $\ketbra{\psi}$ is an orthogonal rank-one projector, it is Hermitian and has a single eigenvalue equal to one, with all others zero.  
  Thus, by the Bauer-Fike theorem, every eigenvalue of $T = \ketbra{\psi} + E$ satisfies either 
  $\abs{\lambda} \leq \epsilon$  or $\abs{1-\lambda} \leq \epsilon$.

  Next, we use the continuity of eigenvalues of a matrix under  perturbations to show that there is 
  exactly one eigenvalue $\lambda$ of $T = \ketbra{\psi} + E$ satisfying $\abs{1-\lambda} \leq \epsilon$. 
  Consider a family of matrices $\{ T_{\alpha} \coloneqq \ketbra{\psi} + \alpha E \}_{\alpha}$ where 
  $\alpha \in [0,1] \subset \Reals$. Denote by $\Delta \subseteq [0,1]$ the set of values of $\alpha$ 
  such that $T_{\alpha}$ has exactly one eigenvalue $\lambda_{\alpha} \in \operatorname{spec}(T_{\alpha})$ 
  satisfying $\abs{1-\lambda_\alpha} \leq \epsilon$. First, we observe that, when $\alpha = 0$,
  the matrix $T_0 = \ketbra{\psi}$, has a single eigenvalue $\lambda_0 = 1$ (with all others zero), 
  meaning that $\alpha=0$ is in $\Delta$.
  We now use the rigorous notion of eigenvalue continuity in Theorem D2 of Ref.~\citenum{HornJohnson2013} to show that $\Delta = [0,1]$. 
  That theorem states that, given two $D \times D$ matrices $A$ and $B$ with eigenvalues 
  $a^{\,}_1, \dots, a^{\,}_D$ and $b^{\,}_1, \dots, b^{\,}_D$, respectively, 
  there exists a permutation $\pi$ of the labels $1,\dots,D$ such that
  \begin{equation}
    \label{eq_theorem_D2}
    \max_{1 \leq i \leq D} \left| a^{\vpp}_i - b^{\vpp}_{\pi(i)} \right|  \leq
     2^{2-1/D} \, \big(  \norm{A}^{\vpp}_{\text{op}} + \norm{B}^{\vpp}_{\text{op}} \big)^{1 - 1/D}
    \norm{A-B}^{1/D}_{\text{op}} \, . ~~
  \end{equation}
  Let $\alpha \in \Delta$ and let $A = T_\alpha$. Because the family $\listb{T_\alpha}_\alpha$ is uniformly bounded 
  in operator norm and $D$ is a constant, for any $\varepsilon$, 
  there exists a $\delta'$ such that, for all $\delta \leq \delta'$, 
  setting $B = T_{\alpha+\delta}$ in Eq.~\ref{eq_theorem_D2} ensures that the left-hand side is smaller than $\varepsilon$.
  Thus, for any $\alpha \in \Delta$ there is a neighborhood
  of $\alpha$ whose intersection with $[0,1]$ is contained entirely within $\Delta$. 
  Additionally, by the same argument above, the complement of $\Delta$ in $[0,1]$ (which we denote $\Delta^{\text{c}}$) 
  has the property that, for any $\beta \in \Delta^{\text{c}}$, there is a neighborhood of $\beta$ 
  whose intersection with $[0,1]$ is contained entirely within $\beta \in \Delta^{\text{c}}$. 
  As a result, the supremum of $\Delta$ cannot lie in the interior of $[0,1]$; 
  since $0 \in \Delta$, the supremum of $\Delta$ must be strictly larger than $0$ by the previous neighborhood argument, and it thus follows that $\Delta = [0,1]$. 

  Since all other eigenvalues must have magnitude less than or equal to $\epsilon$, $\lambda$ is the largest-magnitude 
  eigenvalue of $T$. Finally, because $T$ is real, its eigenvalues must come in 
  complex-conjugate pairs. Because $\lambda$ is the only eigenvalue with magnitude within $\epsilon$ of unity, 
  it must be its own complex conjugate. Hence, $\lambda$ is real.
\end{proof}

\begin{lem}
  Let $T$ be a real matrix that can be expressed as $T = \ketbra{\psi} + E$, 
  where $E$ is an operator satisfying $\norm{E}^{\,}_{\mathrm{op}} \leq \epsilon \leq 1/4$. 
  Denote by $\lambda$ the real eigenvalue of $T$ satisfying $\lambda \geq 1-\epsilon$, 
  whose existence is guaranteed by Lemma~\ref{lem_perturbed_eigenvalue},
  and let $\ket{\lambda}$ be the associated eigenvector of $T$. Then, $T$ can be expressed as 
  $T = \lambda \ketbra{\lambda} + E'$, where the perturbing matrix $E'$ satisfies 
  $\norm{E'}^{\,}_{\mathrm{op}} \leq (1+\sqrt{2})\epsilon$ and $E' \ket{\lambda} = 0$.
  \label{lem_perturbation_norm_better}
\end{lem}
\begin{proof}
  Note that $E' = T - \lambda \ketbra\lambda$; because $\ket{\lambda}$ is an eigenvector of $T$ 
  with eigenvalue $\lambda$, it follows that $E'\ket{\lambda} = 0$. Next, we note that 
  $\norm{E'}^{\,}_{\mathrm{op}} = \max_{\ket{\phi}} \norm{E' \ket{\phi}}$, where the maximum is taken over 
  unit vectors $\ket{\phi} \in V$ satisfying $\norm{\ket{\phi}}=1$. We now bound $\norm{E' \ket{\phi}}$. 

  Without loss of generality, suppose that $\ket{\phi} = C_0 \ket{\lambda} + C_1 \ket{g}$, 
  where $\inprod{g}{g} = 1$, $\inprod{\lambda}{g} = 0$, and $\abs{C_0}^2 + \abs{C_1}^2 = 1$. 
  This implies that $\abs{C_{\ell}} \leq 1$ for $\ell=0,1$. Then, because 
  $E' \ket{\phi} = C_1 E'\ket{g} = C_1 (T - \lambda \ketbra{\lambda})\ket{g} = C_1 T\ket{g}$,
  we have that $\norm{E'\ket{\phi}} = \norm{ C_1 T \ket{g}} \leq {\norm{T \ket{g}}}$. 
  Substituting $T = E + \ketbra{\psi} $, we use the triangle inequality to recover 
  \begin{align}
    \norm{E' \ket{\phi}} &\leq \norm{ (E + \ketbra{\psi} ) \ket{g}} \notag \\
    &\leq  \norm{E \ket{g}} +  \norm{\inprod{\psi}{g} \ket{\psi}} \notag \\
    &\leq \epsilon + \abs{\inprod{\psi}{g}} \, . ~~\label{eq_E'_psi_g_bound}
  \end{align}
  To recover an upper bound on $\abs{\inprod{\psi}{g}}$, 
  we write $\ket\psi = \sqrt{1-\delta^2}\ket{\lambda} + \delta\ket{\delta}$, 
  where $\inprod{\delta}{\delta} = 1$, $\inprod{\delta}{\lambda} = 0$, and $0 \leq \delta \leq 1$.
  Then, it follows that $\inprod{\lambda}{\psi} = \sqrt{1-\delta^2}$ and 
  $\inprod{\delta}{\psi} = \delta$. Thus,
  \begin{equation}
    \label{eq_psi_g_bound}
    \abs{\inprod{\psi}{g}} = \delta \abs{\inprod{\delta}{g}} \leq \delta \, .~~
  \end{equation}
  It remains to bound $\delta$ from above. 
  Expanding $ \lambda \ket{\lambda} = (\ketbra{\psi} + E) \ket{\lambda}$ leads to
  \begin{equation}
    \label{eq_lem_perturb_eig_useful}
    \lambda\ket\lambda = \sqrt{1-\delta^2}\ket\psi + E\ket\lambda \, .~~
  \end{equation}
  Applying $\bra{\lambda}$ on the left and using $\inprod{\lambda}{\lambda}=1$, we find that
  \begin{equation*}
    \lambda = (1-\delta^2) + \matel{\lambda}{E}{\lambda} \leq 1  - \delta^2 + \epsilon \, ,~~
  \end{equation*}
  and so $1 - \delta^2 \geq \lambda - \epsilon$. Using the result $\lambda \geq 1 - \epsilon$ from 
  Lemma~\ref{lem_perturbed_eigenvalue} then implies that
  \begin{equation}
    1 - \delta^2 \geq 1 - 2 \epsilon \label{eq_lem_perturb_eig_delta_bound1} \, . ~~
  \end{equation}
  Next, applying $\bra{\delta}$ to Eq.~\ref{eq_lem_perturb_eig_useful} from the left, 
  we find that 
  \begin{equation*}
    0 = \delta \sqrt{1-\delta^2} +\matel{\delta}{E}{\lambda} 
    \implies \abs{\delta} \leq \frac{\abs{\matel{\delta}{E}{\lambda}}}{\sqrt{1-\delta^2}}
    \leq \frac{\epsilon}{\sqrt{1-\delta^2}} \, , ~~
  \end{equation*}
  where we used the fact that $\left| \matel{\delta}{E}{\lambda} \right| \leq \norm{{\ket{\delta}}} \, \norm{E}^{\,}_{\text{op}} \, \norm{\ket{\lambda}} = \norm{E}_{\text{op}} = \epsilon$, 
  since $\ket{\lambda}$ and $\ket{\delta}$ are unit vectors. 
  Combining the inequality above with 
  Eq.~\ref{eq_lem_perturb_eig_delta_bound1}, we find that
  \begin{equation}
    \delta \leq \frac{\epsilon}{\sqrt{1-\delta^2}} \leq \frac{\epsilon}{\sqrt{1-2\epsilon}}.
  \end{equation}
  For $\epsilon \leq 1/4$, we have that $(1-2\epsilon)^{-1/2} \leq \sqrt{2}$, and thus $\delta \leq \sqrt{2} \, \epsilon$.
  Combining the relation above with Eqs.~\ref{eq_E'_psi_g_bound} and \ref{eq_psi_g_bound} leads to $\norm{E'}^{\,}_{\mathrm{op}} \leq (1+\sqrt{2}) \epsilon$.

\end{proof}

\begin{lem}
  Let $\lambda \in (0,1]$ and $\epsilon \in [0,\lambda)$. Let $\ket{\lambda} \in V$ be a unit vector,  
  let $E$ be an operator on $V$ satisfying $\norm{E}^{\,}_{\mathrm{op}} \leq \epsilon$ and $E \ket{\lambda} = 0$, 
  and let  $T = \lambda \ketbra{\lambda} + E$. Then, there exists a vector \(\ket{f}\in V\) 
  with \(\norm{\ket{f}}\leq 1/(1-\epsilon/\lambda)\) such that, for all $\ncyc \in \Nats$ and all $\ket{a},\ket{b} \in V$,
  \begin{align}
    \left| \matel*{a}{T^{\ncyc}}{b} -
    \braket{a}{\lambda}\braket{f}{b}\lambda^{K}\right|
    &\leq
      \left(1+\frac{1}{1-\epsilon/\lambda}\right)
      \norm{\ket{a}}\norm{\ket{b}}\epsilon^{\ncyc} \, .~~
  \end{align}
\label{lem_vector_geo_sum}
\end{lem}

  \begin{proof}
    Because \(E\ket{\lambda}=0\),
    when expanding \(T^{\ncyc} = (\lambda\ketbra{\lambda} + E)^{\ncyc}\),
    only terms in the summand of the form \((\lambda\ketbra{\lambda})^{\cycind} E^{\ncyc-\cycind}\)
    survive. Therefore, 
    \begin{align}
      T^{\ncyc} &= \sum_{\cycind=0}^{\ncyc-1}\lambda^{\ncyc-\cycind}\ketbra{\lambda} E^{\cycind}+ E^{\ncyc}
      \notag\\
      &= \lambda^{\ncyc}\left(
        \sum_{\cycind=0}^{\infty}\ketbra{\lambda}E^{\cycind}/\lambda^{\cycind}
        - \sum_{\cycind=\ncyc}^{\infty}\ketbra{\lambda}E^{\cycind}/\lambda^{\cycind}
        \right) + E^{\ncyc} \, .~~
        \label{eq_lem23_a}
    \end{align}
    Let \(\bra{f_{\ell}}= \bra{\lambda} \sum_{\cycind=\ell}^{\infty}\left(E/\lambda\right)^{\cycind}\).
    Because \(\norm{E}^{\,}_{\mathrm{op}}\leq \epsilon\), the norm of the \(\cycind\)th term in the sum that defines 
    \(\bra{f_{\ell}}\) is upper bounded by \((\epsilon/\lambda)^{\cycind}\). Since \(\epsilon<\lambda\),
    the series converges, \(\ket{f_{\ell}}\) is well defined as a vector in 
    \(V\), and its norm satisfies 
    \(\norm{\ket{f_{\ell}}} \leq
    \sum_{\cycind=\ell}^{\infty}(\epsilon/\lambda)^{\cycind} =
    (\epsilon/\lambda)^{\ell}/(1-\epsilon/\lambda)\).

    Rearranging the terms in Eq.~\ref{eq_lem23_a} leads to
    \begin{align}
      T^{\ncyc}- \lambda^{\ncyc}\KBop{\lambda}{f_{0}}
      &= \lambda^{\ncyc}\KBop{\lambda}{f_{\ncyc}}+E^{\ncyc} \, .~~
        \label{eq_lem23_b}
    \end{align}
    To complete the proof, let \(\ket{f} \coloneqq \ket{f_{0}}\), and apply
    \(\bra{a}\) to Eq.~\ref{eq_lem23_b} from the left, apply \(\ket{b}\) from the right, 
    and take the absolute value to recover
    \begin{align}
      \left| \matel{a}{T^{\ncyc}}{b} - \braket{a}{\lambda} \braket{f}{b} \lambda^{\ncyc} \right|
      &= \left| \lambda^{\ncyc} \braket{a}{\lambda} \braket{f_{\ncyc}}{b} + \matel{a}{E^{\ncyc}}{b}\right|
        \notag \\
      &\leq
         \norm{\ket{a}} \, \norm{\ket{b}} \,\left( \lambda^{\ncyc} \norm{\ket{f_{\ncyc}}} 
        + \epsilon^{\ncyc} \right) \notag\\
      &\leq
        \left(1+\frac{1}{1-\epsilon/\lambda}\right)
        \norm{\ket{a}} \, \norm{\ket{b}} \, \epsilon^{\ncyc} \, ,~
    \end{align}
    where we used the triangle equality and the standard bound $\norm{\matel{a}{T}{b}} \leq \norm{\ket{a}} \, \norm{\ket{b}} \norm{T}^{\,}_{\text{op}}$.
  \end{proof}

\section{Proof of Theorem~\ref{thm_main}}
\label{sec_proof_main_thm}

Unless otherwise specified, the norm of any vector is taken to be the Euclidean ($L^2$) norm,
which induces a corresponding operator norm for matrices (where $\norm{A}^{\,}_{\text{op}}$ is the smallest number $c>0$ such 
that $\norm{A v} \leq c \norm{v}$ for all vectors $v$).

Overall, the proof strategy is to show that, for each QEC experimental setting $(p,m) \in \mathsf{Exp}$
and for each logical Pauli $P \in \PauliSet*{L}$,
the Pauli eigenvalue $\lambda_{p,m,P}^{(\ncyc)}$~\eqref{eq_qec_lambda_total_eig} of $\Lambda_{p,m}^{(\ncyc)}$~\eqref{eq_lambda_total} 
for the logical Pauli $P$ is well approximated by a function that decays exponentially in $\ncyc$.
The approximate exponential decay is constructed such that both the decay rate and the extent to which the decay 
approximates $\lambda_{p,m,P}^{(\ncyc)}$ are independent of $(p,m)$.
Then, Prop.~\ref{prop_exp_decay_approx_model} guarantees the existence of a $(G,\epsilon)$-approximate logical Markovian model 
for $\mathsf{Exp}$ (see Def.~\ref{def_markov_approx}). 

\begin{proof}
The starting point is the expression for the total logical error channel $\Lambda_{p,m}^{(\ncyc)}$ (see Cor.~\ref{cor_lambda_tot}); 
because of the \smip{}, we have that
\begin{align*}
  \Lambda^{(\ncyc)}_{p,m} (\sigma_L) = \hspace{-2mm} \sum_{s_0,\dots,s_\ncyc} \hspace{-0.7mm}
  \gamma^{\vpp}_{s_{\ncyc}} \hspace{-0.8mm} \cdots \gamma^{\vpp}_{s_1} 
  \gamma^{\vpp}_{s_0}~
  \widetilde{\Lambda}^{\text{(meas)}}_{s_\ncyc} \hspace{-0.3mm} \circ \widetilde{\Lambda}^{\vpp}_{s_{\ncyc},s_{\ncyc-1}} \hspace{-0.8mm} \circ 
  \cdots \circ \widetilde{\Lambda}^{\vpp}_{s_1,s_0}\hspace{-0.5mm}
  \circ  \widetilde{\Lambda}^{\text{(prep)}}_{s_0} ( \sigma_L)
  \tag{\ref{eq_qec_cor_frame_total}} \, ,~~
\end{align*}
where $\gamma^{\,}_{s_{\cycind}}$ is the syndrome marginal distribution (see Sec.~\ref{sec_smip}).

We now express each Pauli eigenvalue $\lambda_{p,m,P}^{(\ncyc)}$~\eqref{eq_qec_lambda_total_eig} of 
$\Lambda^{(\ncyc)}_{p,m}$~\eqref{eq_qec_cor_frame_total} in terms of a product of square matrices 
with dimension equal to the number of bitstrings of system $S$. Doing so facilitates the application 
of results from matrix perturbation theory (see App.~\ref{app_matrix_lemmas}).
First, for each logical Pauli $P \in \PauliSet*{L}$, we define a matrix $F_P$ with components 
\begin{equation}
  \left( F_P \right)_{\soutput,\sinput} \coloneqq 
  \widetilde{\lambda}^{\vpp}_{\soutput,\sinput,P} 
  \label{eq_F_def} \, ,~~
\end{equation}
where $\widetilde{\lambda}^{\vpp}_{\soutput,\sinput,P}$ is the Pauli eigenvalue of the logical error channel 
$\widetilde{\Lambda}^{\vpp}_{\soutput,\sinput}$~\eqref{eq_cor_frame_qec_chan} for the logical Pauli $P$ 
(see also Eq.~\ref{eq_qec_lambda_total_eig}). Next, we define a diagonal matrix associated 
with the syndrome marginal distribution $\gamma$, given by
\begin{equation}
  \label{eq_Sigma_def}
  \Sigma \coloneqq \operatorname{diag}(\gamma)  ~~ \implies ~~ \Sigma^{\vpp}_{\soutput,\sinput} = \kron{\soutput,\sinput} \gamma^{\vpp}_{\soutput} \, , ~~
\end{equation}
where $\kron*{s,s'}$ is the Kronecker delta. We also define a pair of vectors $\prepvec$ and $\measvec$, where
\begin{subequations}
  \label{eq_Markov_vectors}
  \begin{align}
    \left( \prepvec \right)_{s}  &\coloneqq \lambda^{(\text{prep})}_{s,P} 
    \label{eq_prep_vec} \\
    \left( \measvec \right)_{s}  &\coloneqq \lambda^{(\text{meas})}_{s,P} 
    \label{eq_meas_vec} \, ,~~
  \end{align}
\end{subequations}
for each logical Pauli $P \in \PauliSet*{L}$, where $\lambda^{(\text{prep})}_{s,P}$ and $\lambda^{(\text{meas})}_{s,P}$ 
are the Pauli eigenvalues of the channels $\widetilde{\Lambda}^{(\text{prep})}_{s}$~\eqref{eq_cor_frame_prep_chan} and 
$\widetilde{\Lambda}^{(\text{meas})}_{s}$~\eqref{eq_cor_frame_meas_chan}, respectively (see Eq.~\ref{eq_qec_lambda_total_eig}).
Combining the definitions above, the Pauli eigenvalue $\lambda^{(\ncyc)}_{p,m,P}$ of the total logical error channel 
$\Lambda_{p,m}^{(\ncyc)}$~\eqref{eq_qec_cor_frame_total} can be written
\begin{equation}
  \label{eq_Pauli_eig_matrix_naive}
  \lambda_{p,m,P}^{(\ncyc)} = \measvec^T (\Sigma F_P )^{\ncyc} \Sigma \, \prepvec \, ,~~
\end{equation}
for each logical Pauli $P \in \PauliSet*{L}$. 
It is convenient to rewrite the previous expression as
\begin{equation}
  \lambda_{p,m,P}^{(\ncyc)} = \measvec^T \,  \ssig (\ssig F_P \ssig)^{\ncyc} \ssig \, \prepvec \, , ~~
  \label{eq_qec_lambda_matrix_form}
\end{equation}
in order to highlight the fact that the behavior of $\lambda_{p,m,P}^{(\ncyc)} $ as a function of $\ncyc$ is heavily influenced by the eigenvalue structure of of $\ssig F_P \ssig$.

Next, we express $F_P$ as a perturbation of the matrix $N$ of all ones 
(i.e., $N^{\,}_{\soutput,\sinput}=1$ for all $\sinput$ and $\soutput$). In particular, we define the matrix $E_P$ such that
\begin{equation}
  F_P = N  - 2E_P \, ,~~
  \label{eq_approxmarkov_def_E}
\end{equation}
so that the entries of $E_P$  lie in the interval $[0,1]$ (see Eq.~\ref{eq_F_def}). 
We now use the condition $\epsilon_1 \leq 1/64$ of the theorem to show that the operator norm of 
$2\ssig E_P \ssig $ is small.

As discussed below Prop.~\ref{prop_ssp_smip}, 
the probability of the identity in the mixture of Pauli superoperators that makes up $\Lambda_1$ is equal to 
$f_1$ and the total probability of non-identity superoperators is $\epsilon_{1}$. For any Pauli $P$, we denote the corresponding Pauli eigenvalue of $\Lambda_1$ by $\lambda_{1,P}$.
The most negative contribution a non-identity superoperator can make to $\lambda_{1,P}$ occurs when it maps $P$ to $-P$.
Thus, for any $P$ the corresponding Pauli eigenvalue satisfies

\begin{equation}
  \label{eq_lambda1_bound}
  \lambda^{\vpp}_{1,P} \geq 1 - 2\epsilon_{1} \, . ~~
\end{equation} 
%
This bound is saturated in the worst-case scenario in which all nonidentity superoperators in $\Lambda_1$ 
associated with $\epsilon_1 > 0$ have Pauli eigenvalue $-1$ for the logical Pauli $P$.

Next, we bound the norm of $E_P$~\eqref{eq_approxmarkov_def_E} by relating $\lambda^{\,}_{1,P}$ to $F_P$.
In particular, for any Pauli $P$, we have that $\lambda^{\,}_{1,P} = \gamma^T F_P \gamma$. Because $\gamma$ is a probability 
distribution and $N$ is a matrix of ones, we have that $\gamma^T N \gamma = 1$. Combining this with the bound on  
$\lambda^{\,}_{1,P}$~\eqref{eq_lambda1_bound} leads to 
\begin{equation}
  \label{eq_lambda_EP_bound}
  \lambda^{\vpp}_{1,P} = \gamma^T F_P \gamma = 1 - 2 \gamma^T E_P \gamma \geq 1 - 2 \epsilon_1  \implies \gamma^T E_P \gamma \leq \epsilon_1 \,, ~~
\end{equation}
and combining the foregoing result with Prop.~\ref{prop_perturbation_norm_bound} leads to the bound
\begin{equation}
  \abs{2\ssig E_P \ssig}_\text{op} \leq 2\sqrt{\epsilon_1} \, . ~~
  \label{eq_perturb_norm_small}
\end{equation}
To make further progress, we observe that $\ssig N \ssig$ is a rank-one projector. To see this, 
it is convenient to write $\ssig N \ssig = \ketbra{\sqrt{\gamma}}$, where we define the unit vector
\begin{equation}
  \label{eq_sqrt_gamma_vec}
  \ket{\sqrt{\gamma}} = \sum\limits_{s} \gamma^{1/2}_s \, \ket{s} \, ,~~
\end{equation}
where $s$ runs over bitstrings of $S$, so that $\inprod{\sqrt{\gamma}}{\sqrt{\gamma}} = \sum_s \gamma^{\vpp}_s = 1$. 
It is straightforward to check that $\ketbra{\sqrt{\gamma}}$ is idempotent and has unit trace. 
Using Eq.~\ref{eq_sqrt_gamma_vec}, we write 
\begin{equation}
  \label{eq_FP_expansion}
  \ssig F_P \ssig = \ketbra{\sqrt\gamma} + 2\ssig E_P \ssig \,  ,~~
\end{equation}
which we use to show that $\lambda_{p,m,P}^{(\ncyc)}$~\eqref{eq_qec_lambda_total_eig} is 
well approximated by a function that decays exponentially in $\ncyc$.
We do so by applying Lemma~\ref{lem_perturbation_norm_better} with $T = \ssig F_P \ssig$,
$\ketbra{\psi} = \ketbra{\sqrt{\gamma}}$, and $E = 2 \ssig E_P \ssig$. The Lemma's 
requirement that $\norm{E}^{\,}_{\text{op}} \leq 1/4$ is automatically fulfilled by 
Eq.~\ref{eq_perturb_norm_small} and the assumption $\epsilon_1 \leq 1/64$ entering the theorem. 
As a result of Lemma~\ref{lem_perturbation_norm_better}, 
\begin{equation}
  \label{eq_SigFSig_decomp1}
  \ssig F_P \ssig = \lambda_P \ketbra{\lambda_P} + E' \, ,~~
\end{equation}
for a unit vector $\ket{\lambda_P}$ and a matrix $E'$ satisfying $\norm{E'}^{\,}_{\text{op}} \leq 2 ( 1 + \sqrt{2}) \sqrt{\epsilon_1}$ and $E'\ket{\lambda_P} = 0$, where
$\lambda_P \geq 1 - 2 \sqrt{\epsilon_1}$ by Lemma~\ref{lem_perturbed_eigenvalue}. Because $\epsilon_1 \leq 1/64$, 
we have that $\lambda_P \geq 3/4$ and $\norm{E'}^{\,}_{\text{op}} \leq (1 + \sqrt{2})/4$; 
as a result, $\norm{E'}^{\,}_{\text{op}} \leq \lambda_P$.

Next, we apply Lemma~\ref{lem_vector_geo_sum} to bound the extent to which $\lambda^{(\ncyc)}_{p,m,P}$ can be approximated by a single exponential decay as a function of $K$. 
We take $T = \ssig F_P \ssig $, with $\lambda_P$ and $\ketbra{\lambda_P}$ as in Eq.~\ref{eq_SigFSig_decomp1}, and set the operator $E$ in Lemma~\ref{lem_vector_geo_sum} to be $E'$ in Eq.~\ref{eq_SigFSig_decomp1}, where $\norm{E'}_\text{op} \leq \epsilon$ with $\epsilon \coloneqq 2(1+\sqrt{2})\sqrt{\epsilon_1}$.
We define the vectors 
$\ket{a} = \ssig \, \measvec$ and $\ket{b} = \ssig \, \prepvec$, where
\begin{equation}
  \norm{\ket{a}} = \norm{\ssig \, \measvec} 
  = \Big( \sum_s \gamma^{\vpp}_s \, \big| \lambda^{(\text{meas})}_{s,P} \big|^2 \Big)^{1/2} 
  \leq \Big( \sum_s \gamma^{\vpp}_s \Big)^{1/2} = 1 \, ,~~
\end{equation}
and we similarly have $\norm{\ket{b}} \leq 1$, since 
$\lvert \lambda^{(\text{meas})}_{s,P} \rvert, \lvert \lambda^{(\text{prep})}_{s,P} \rvert \in [0,1]$, 
as discussed above Eq.~\ref{eq_lambda1_bound}.
In terms of $\ket{a}$ and $\ket{b}$, $\lambda^{(\ncyc)}_{p,m,P} = \bra{b} T^K \ket{a}$.
Then,
Lemma~\ref{lem_vector_geo_sum} guarantees the existence of a vector $\ket{f}$ 
such that, after defining $C^{(\text{prep})}_{p,P} = \inprod{a}{\lambda_P}$ and $C^{(\text{meas})}_{m,P} = \inprod{f}{b}$ 
we have that
\begin{align}
  \left| \lambda^{(\ncyc)}_{p,m,P} - C^{(\text{prep})}_{p,P} C^{(\text{meas})}_{m,P} \lambda^{~\ncyc}_{P} \right| 
  &\leq \Big(1 + \frac{1}{1-\epsilon/\lambda}\Big) \, \epsilon^{\ncyc}\, ,~~
  \label{eq_we_did_it}
\end{align}
where all dependence on the experiment setting $(p,m)$ and the Pauli $P$ is shown explicitly.

Finally, we define a real number
\begin{equation}
  G' = 1 + \frac{1}{1-\epsilon/\lambda} \leq 1 + \frac{1}{1-\epsilon/(1-2\sqrt{\epsilon_1})} < 7 \, , ~~
\end{equation}
where the two inequalities above follow from $\lambda \geq 1-2\sqrt{\epsilon_1}$ and $\epsilon_1 \leq 1/64$, respectively.

The vector $\ket{\lambda}$ is real valued in the standard basis because both $T$ and $E$ are real matrices and $\lambda$ is real. 
As a result, the vector $\ket{f}$ in Lemma~\ref{lem_vector_geo_sum} is also real valued.
Moreover, for each logical Pauli $P$, $\lambda$ plays the role of $\chi^{\,}_P$~\eqref{eq_Phis_from_Paulis} in Prop.~\ref{prop_exp_decay_approx_model}, so that the latter is real valued.
Additionally, both $\ket{a}$ and $\ket{b}$ are real because $\lambda^{(\text{prep})}_{s,P}$ and $\lambda^{(\text{prep})}_{s,P}$ are the Pauli eigenvalues of a Pauli channel (see Eq.~\ref{eq_Markov_vectors}). Thus, both $C^{(\text{prep})}_{p,P}$ and $C^{(\text{meas})}_{m,P}$ are also real for all experimental settings $(p,m) \in \mathsf{Exp}$ and all logical Pauli operators $P$.

To conclude the proof we apply Prop.~\ref{prop_exp_decay_approx_model} along with Eq.~\ref{eq_we_did_it}.
This ensures the existence of a ($G,\epsilon$)-approximate logical Markovian model (see Def.~\ref{def_markov_approx}) for $\mathsf{Exp}$ with
$G = G'\sqrt{D}$ and $\epsilon = 2(1+\sqrt{2})\sqrt{\epsilon_1}$, 
where $D = 2^{\nlog} = \operatorname{dim}(\hilbert^{\,}_L)$ is the logical dimension.

\end{proof}

\section{Discussion of Definition~\ref{defn_log_fid}}
\label{sec_discuss_defn_ss_lf}

Here we discuss a concrete scenario in which the total logical error channel 
$\Lambda^{(\ncyc)}_{\text{tot}}$~\eqref{eq_lambda_total} is equal to the single-cycle logical error channel
$\Lambda_1$~\eqref{eq_Lambda_1} in Def.~\ref{defn_log_fid}. 
Thus, the scenario we consider gives an experimental means of probing $\Lambda_1$ directly, in principle.

In particular, we consider a QEC experiment (see Def.~\ref{defn_QEC_experiment}) involving a single QEC cycle (i.e., $\ncyc = 1$).
The main additional requirement is that the state-preparation procedure (see Sec.~\ref{sec_Pauli_noise_state_prep}) 
involves only correctable errors, and that the output state is thus of the from
\begin{equation}
  \label{eq_Lambda1_init_state}
  \rho_0 = \sigma_L \otimes \sum\limits_s \gamma^{\,}_s \, \ZProj{s} \, ,~~
\end{equation}
where $\gamma^{\,}_s$ is the syndrome marginal distribution associated with a QEC cycle 
(see the discussion near Def.~\ref{def_smip}). 
We then apply a single noisy QEC cycle to $\rho_0$, as in Sec.~\ref{sec_noisy_qec_cycles}, 
leading to
\begin{equation*}
  \tag{\ref{eq_Lambda_1}}
  \Lambda_1 (\sigma_L)= \sum\limits_{\soutput,\sinput} \, \gamma_{\soutput} \, \gamma_{\sinput}\, \widetilde{\Lambda}_{\soutput,\sinput} (\sigma_L) \, .~~
\end{equation*}
The single-cycle entanglement fidelity $f_1$ can be determined experimentally by making measurements of this state for different values of $\sigma_L$, provided that Eq.~\ref{eq_Lambda1_init_state} holds.

To provide further insight into 
$\Lambda_1$, we show that it is a first-order approximation to the logical superoperator $\Phi^{\,}_{\text{QEC}}$ 
in the Markovian model $\mathcal{M}$ constructed in the proof of Theorem~\ref{thm_main} (see App.~\ref{sec_proof_main_thm}).
First, we note that Eq.~\ref{eq_approxmarkov_def_E} gives an expression for the matrix $F_P$ whose largest eigenvalue 
determines the Pauli eigenvalue $\lambda^{\,}_P$ of the logical superoperator $\Phi^{\,}_{\text{QEC}}$ in Prop.~\ref{prop_exp_decay_approx_model}. 
Second, in Eq.~\ref{eq_FP_expansion}, we write $F_P$ as a perturbation about the rank-one projector $\ketbra{\sqrt{\gamma}}$,
whose largest eigenvalue is one. 
The perturbing operator in Eq.~\ref{eq_FP_expansion} is $2\ssig E_P \ssig$; 
to first order, the perturbation to the eigenvalue $\lambda=1$ of $\ketbra{\sqrt{\gamma}}$ is 
\begin{equation}
  -\matel{\sqrt{\gamma}}{2\ssig E_P \ssig}{\sqrt{\gamma}} = - 2 \gamma^T E_P \gamma \, , ~~
\end{equation}
where $\gamma$ is a vector whose components give the syndrome marginal distribution (see Sec.~\ref{sec_smip}).
Hence, the first-order approximation to the eigenvalue of $F_P$ is given by $\lambda = 1 - 2 \gamma^T E_P \gamma$, 
which is equal to $\gamma^T F_P \gamma$. This quantity is, in turn, equal to $\lambda^{\,}_{1,P}$~\eqref{eq_lambda_EP_bound}, 
the Pauli eigenvalue of $\Lambda_1$~\eqref{eq_Lambda_1} for the Pauli $P$. 
Hence, the channel $\Lambda_1$ determines a first-order approximation of the dominant eigenvalue of $F_P$, which in turn determines $\Phi^{\,}_{\text{QEC}}$ in Prop.~\ref{prop_exp_decay_approx_model}.

\bibliography{ref.bib}

\begin{thebibliography}{34}%
\makeatletter
\providecommand \@ifxundefined [1]{%
 \@ifx{#1\undefined}
}%
\providecommand \@ifnum [1]{%
 \ifnum #1\expandafter \@firstoftwo
 \else \expandafter \@secondoftwo
 \fi
}%
\providecommand \@ifx [1]{%
 \ifx #1\expandafter \@firstoftwo
 \else \expandafter \@secondoftwo
 \fi
}%
\providecommand \natexlab [1]{#1}%
\providecommand \enquote  [1]{``#1''}%
\providecommand \bibnamefont  [1]{#1}%
\providecommand \bibfnamefont [1]{#1}%
\providecommand \citenamefont [1]{#1}%
\providecommand \href@noop [0]{\@secondoftwo}%
\providecommand \href [0]{\begingroup \@sanitize@url \@href}%
\providecommand \@href[1]{\@@startlink{#1}\@@href}%
\providecommand \@@href[1]{\endgroup#1\@@endlink}%
\providecommand \@sanitize@url [0]{\catcode `\\12\catcode `\$12\catcode
  `\&12\catcode `\#12\catcode `\^12\catcode `\_12\catcode `\%12\relax}%
\providecommand \@@startlink[1]{}%
\providecommand \@@endlink[0]{}%
\providecommand \url  [0]{\begingroup\@sanitize@url \@url }%
\providecommand \@url [1]{\endgroup\@href {#1}{\urlprefix }}%
\providecommand \urlprefix  [0]{URL }%
\providecommand \Eprint [0]{\href }%
\providecommand \doibase [0]{https://doi.org/}%
\providecommand \selectlanguage [0]{\@gobble}%
\providecommand \bibinfo  [0]{\@secondoftwo}%
\providecommand \bibfield  [0]{\@secondoftwo}%
\providecommand \translation [1]{[#1]}%
\providecommand \BibitemOpen [0]{}%
\providecommand \bibitemStop [0]{}%
\providecommand \bibitemNoStop [0]{.\EOS\space}%
\providecommand \EOS [0]{\spacefactor3000\relax}%
\providecommand \BibitemShut  [1]{\csname bibitem#1\endcsname}%
\let\auto@bib@innerbib\@empty
\bibitem [{\citenamefont {{Google Quantum AI and
  Collaborators}}(2025)}]{acharyaQuantumErrorCorrection2025}%
  \BibitemOpen
  \bibfield  {author} {\bibinfo {author} {\bibnamefont {{Google Quantum AI and
  Collaborators}}},\ }\bibfield  {title} {\bibinfo {title} {Quantum error
  correction below the surface-code threshold},\ }\href
  {https://doi.org/10.1038/s41586-024-08449-y} {\bibfield  {journal} {\bibinfo
  {journal} {Nature}\ }\textbf {\bibinfo {volume} {638}},\ \bibinfo {pages}
  {920} (\bibinfo {year} {2025})}\BibitemShut {NoStop}%
\bibitem [{\citenamefont {Bluvstein}\ \emph {et~al.}(2024)\citenamefont
  {Bluvstein}, \citenamefont {Evered}, \citenamefont {Geim}, \citenamefont
  {Li}, \citenamefont {Zhou}, \citenamefont {Manovitz}, \citenamefont {Ebadi},
  \citenamefont {Cain}, \citenamefont {Kalinowski}, \citenamefont {Hangleiter},
  \citenamefont {Bonilla~Ataides}, \citenamefont {Maskara}, \citenamefont
  {Cong}, \citenamefont {Gao}, \citenamefont {Sales~Rodriguez}, \citenamefont
  {Karolyshyn}, \citenamefont {Semeghini}, \citenamefont {Gullans},
  \citenamefont {Greiner}, \citenamefont {Vuleti{\'c}},\ and\ \citenamefont
  {Lukin}}]{bluvsteinLogicalQuantumProcessor2024}%
  \BibitemOpen
  \bibfield  {author} {\bibinfo {author} {\bibfnamefont {D.}~\bibnamefont
  {Bluvstein}}, \bibinfo {author} {\bibfnamefont {S.~J.}\ \bibnamefont
  {Evered}}, \bibinfo {author} {\bibfnamefont {A.~A.}\ \bibnamefont {Geim}},
  \bibinfo {author} {\bibfnamefont {S.~H.}\ \bibnamefont {Li}}, \bibinfo
  {author} {\bibfnamefont {H.}~\bibnamefont {Zhou}}, \bibinfo {author}
  {\bibfnamefont {T.}~\bibnamefont {Manovitz}}, \bibinfo {author}
  {\bibfnamefont {S.}~\bibnamefont {Ebadi}}, \bibinfo {author} {\bibfnamefont
  {M.}~\bibnamefont {Cain}}, \bibinfo {author} {\bibfnamefont {M.}~\bibnamefont
  {Kalinowski}}, \bibinfo {author} {\bibfnamefont {D.}~\bibnamefont
  {Hangleiter}}, \bibinfo {author} {\bibfnamefont {J.~P.}\ \bibnamefont
  {Bonilla~Ataides}}, \bibinfo {author} {\bibfnamefont {N.}~\bibnamefont
  {Maskara}}, \bibinfo {author} {\bibfnamefont {I.}~\bibnamefont {Cong}},
  \bibinfo {author} {\bibfnamefont {X.}~\bibnamefont {Gao}}, \bibinfo {author}
  {\bibfnamefont {P.}~\bibnamefont {Sales~Rodriguez}}, \bibinfo {author}
  {\bibfnamefont {T.}~\bibnamefont {Karolyshyn}}, \bibinfo {author}
  {\bibfnamefont {G.}~\bibnamefont {Semeghini}}, \bibinfo {author}
  {\bibfnamefont {M.~J.}\ \bibnamefont {Gullans}}, \bibinfo {author}
  {\bibfnamefont {M.}~\bibnamefont {Greiner}}, \bibinfo {author} {\bibfnamefont
  {V.}~\bibnamefont {Vuleti{\'c}}},\ and\ \bibinfo {author} {\bibfnamefont
  {M.~D.}\ \bibnamefont {Lukin}},\ }\bibfield  {title} {\bibinfo {title}
  {Logical quantum processor based on reconfigurable atom arrays},\ }\href
  {https://doi.org/10.1038/s41586-023-06927-3} {\bibfield  {journal} {\bibinfo
  {journal} {Nature}\ }\textbf {\bibinfo {volume} {626}},\ \bibinfo {pages}
  {58} (\bibinfo {year} {2024})}\BibitemShut {NoStop}%
\bibitem [{\citenamefont {Paetznick}\ \emph {et~al.}(2024)\citenamefont
  {Paetznick}, \citenamefont {da~Silva}, \citenamefont {{Ryan-Anderson}},
  \citenamefont {{Bello-Rivas}}, \citenamefont {III}, \citenamefont
  {Chernoguzov}, \citenamefont {Dreiling}, \citenamefont {Foltz}, \citenamefont
  {Frachon}, \citenamefont {Gaebler}, \citenamefont {Gatterman}, \citenamefont
  {{Grans-Samuelsson}}, \citenamefont {Gresh}, \citenamefont {Hayes},
  \citenamefont {Hewitt}, \citenamefont {Holliman}, \citenamefont {Horst},
  \citenamefont {Johansen}, \citenamefont {Lucchetti}, \citenamefont
  {Matsuoka}, \citenamefont {Mills}, \citenamefont {Moses}, \citenamefont
  {Neyenhuis}, \citenamefont {Paz}, \citenamefont {Pino}, \citenamefont
  {Siegfried}, \citenamefont {Sundaram}, \citenamefont {Tom}, \citenamefont
  {Wernli}, \citenamefont {Zanner}, \citenamefont {Stutz},\ and\ \citenamefont
  {Svore}}]{paetznickDemonstrationLogicalQubits2024}%
  \BibitemOpen
  \bibfield  {author} {\bibinfo {author} {\bibfnamefont {A.}~\bibnamefont
  {Paetznick}}, \bibinfo {author} {\bibfnamefont {M.~P.}\ \bibnamefont
  {da~Silva}}, \bibinfo {author} {\bibfnamefont {C.}~\bibnamefont
  {{Ryan-Anderson}}}, \bibinfo {author} {\bibfnamefont {J.~M.}\ \bibnamefont
  {{Bello-Rivas}}}, \bibinfo {author} {\bibfnamefont {J.~P.~C.}\ \bibnamefont
  {III}}, \bibinfo {author} {\bibfnamefont {A.}~\bibnamefont {Chernoguzov}},
  \bibinfo {author} {\bibfnamefont {J.~M.}\ \bibnamefont {Dreiling}}, \bibinfo
  {author} {\bibfnamefont {C.}~\bibnamefont {Foltz}}, \bibinfo {author}
  {\bibfnamefont {F.}~\bibnamefont {Frachon}}, \bibinfo {author} {\bibfnamefont
  {J.~P.}\ \bibnamefont {Gaebler}}, \bibinfo {author} {\bibfnamefont {T.~M.}\
  \bibnamefont {Gatterman}}, \bibinfo {author} {\bibfnamefont {L.}~\bibnamefont
  {{Grans-Samuelsson}}}, \bibinfo {author} {\bibfnamefont {D.}~\bibnamefont
  {Gresh}}, \bibinfo {author} {\bibfnamefont {D.}~\bibnamefont {Hayes}},
  \bibinfo {author} {\bibfnamefont {N.}~\bibnamefont {Hewitt}}, \bibinfo
  {author} {\bibfnamefont {C.}~\bibnamefont {Holliman}}, \bibinfo {author}
  {\bibfnamefont {C.~V.}\ \bibnamefont {Horst}}, \bibinfo {author}
  {\bibfnamefont {J.}~\bibnamefont {Johansen}}, \bibinfo {author}
  {\bibfnamefont {D.}~\bibnamefont {Lucchetti}}, \bibinfo {author}
  {\bibfnamefont {Y.}~\bibnamefont {Matsuoka}}, \bibinfo {author}
  {\bibfnamefont {M.}~\bibnamefont {Mills}}, \bibinfo {author} {\bibfnamefont
  {S.~A.}\ \bibnamefont {Moses}}, \bibinfo {author} {\bibfnamefont
  {B.}~\bibnamefont {Neyenhuis}}, \bibinfo {author} {\bibfnamefont
  {A.}~\bibnamefont {Paz}}, \bibinfo {author} {\bibfnamefont {J.}~\bibnamefont
  {Pino}}, \bibinfo {author} {\bibfnamefont {P.}~\bibnamefont {Siegfried}},
  \bibinfo {author} {\bibfnamefont {A.}~\bibnamefont {Sundaram}}, \bibinfo
  {author} {\bibfnamefont {D.}~\bibnamefont {Tom}}, \bibinfo {author}
  {\bibfnamefont {S.~J.}\ \bibnamefont {Wernli}}, \bibinfo {author}
  {\bibfnamefont {M.}~\bibnamefont {Zanner}}, \bibinfo {author} {\bibfnamefont
  {R.~P.}\ \bibnamefont {Stutz}},\ and\ \bibinfo {author} {\bibfnamefont
  {K.~M.}\ \bibnamefont {Svore}},\ }\href
  {https://doi.org/10.48550/arXiv.2404.02280} {\bibinfo {title} {Demonstration
  of logical qubits and repeated error correction with better-than-physical
  error rates}} (\bibinfo {year} {2024}),\ \Eprint
  {https://arxiv.org/abs/2404.02280} {2404.02280 [quant-ph]} \BibitemShut
  {NoStop}%
\bibitem [{\citenamefont {Reichardt}\ \emph {et~al.}(2024)\citenamefont
  {Reichardt}, \citenamefont {Aasen}, \citenamefont {Chao}, \citenamefont
  {Chernoguzov}, \citenamefont {van Dam}, \citenamefont {Gaebler},
  \citenamefont {Gresh}, \citenamefont {Lucchetti}, \citenamefont {Mills},
  \citenamefont {Moses}, \citenamefont {Neyenhuis}, \citenamefont {Paetznick},
  \citenamefont {Paz}, \citenamefont {Siegfried}, \citenamefont {da~Silva},
  \citenamefont {Svore}, \citenamefont {Wang},\ and\ \citenamefont
  {Zanner}}]{reichardtDemonstrationQuantumComputation2024}%
  \BibitemOpen
  \bibfield  {author} {\bibinfo {author} {\bibfnamefont {B.~W.}\ \bibnamefont
  {Reichardt}}, \bibinfo {author} {\bibfnamefont {D.}~\bibnamefont {Aasen}},
  \bibinfo {author} {\bibfnamefont {R.}~\bibnamefont {Chao}}, \bibinfo {author}
  {\bibfnamefont {A.}~\bibnamefont {Chernoguzov}}, \bibinfo {author}
  {\bibfnamefont {W.}~\bibnamefont {van Dam}}, \bibinfo {author} {\bibfnamefont
  {J.~P.}\ \bibnamefont {Gaebler}}, \bibinfo {author} {\bibfnamefont
  {D.}~\bibnamefont {Gresh}}, \bibinfo {author} {\bibfnamefont
  {D.}~\bibnamefont {Lucchetti}}, \bibinfo {author} {\bibfnamefont
  {M.}~\bibnamefont {Mills}}, \bibinfo {author} {\bibfnamefont {S.~A.}\
  \bibnamefont {Moses}}, \bibinfo {author} {\bibfnamefont {B.}~\bibnamefont
  {Neyenhuis}}, \bibinfo {author} {\bibfnamefont {A.}~\bibnamefont
  {Paetznick}}, \bibinfo {author} {\bibfnamefont {A.}~\bibnamefont {Paz}},
  \bibinfo {author} {\bibfnamefont {P.~E.}\ \bibnamefont {Siegfried}}, \bibinfo
  {author} {\bibfnamefont {M.~P.}\ \bibnamefont {da~Silva}}, \bibinfo {author}
  {\bibfnamefont {K.~M.}\ \bibnamefont {Svore}}, \bibinfo {author}
  {\bibfnamefont {Z.}~\bibnamefont {Wang}},\ and\ \bibinfo {author}
  {\bibfnamefont {M.}~\bibnamefont {Zanner}},\ }\href
  {https://doi.org/10.48550/arXiv.2409.04628} {\bibinfo {title} {Demonstration
  of quantum computation and error correction with a tesseract code}} (\bibinfo
  {year} {2024}),\ \Eprint {https://arxiv.org/abs/2409.04628} {2409.04628
  [quant-ph]} \BibitemShut {NoStop}%
\bibitem [{\citenamefont {Postler}\ \emph {et~al.}(2024)\citenamefont
  {Postler}, \citenamefont {Butt}, \citenamefont {Pogorelov}, \citenamefont
  {Marciniak}, \citenamefont {Heu\ss{}en}, \citenamefont {Blatt}, \citenamefont
  {Schindler}, \citenamefont {Rispler}, \citenamefont {M\"uller},\ and\
  \citenamefont {Monz}}]{postler_monz_24_steane}%
  \BibitemOpen
  \bibfield  {author} {\bibinfo {author} {\bibfnamefont {L.}~\bibnamefont
  {Postler}}, \bibinfo {author} {\bibfnamefont {F.}~\bibnamefont {Butt}},
  \bibinfo {author} {\bibfnamefont {I.}~\bibnamefont {Pogorelov}}, \bibinfo
  {author} {\bibfnamefont {C.~D.}\ \bibnamefont {Marciniak}}, \bibinfo {author}
  {\bibfnamefont {S.}~\bibnamefont {Heu\ss{}en}}, \bibinfo {author}
  {\bibfnamefont {R.}~\bibnamefont {Blatt}}, \bibinfo {author} {\bibfnamefont
  {P.}~\bibnamefont {Schindler}}, \bibinfo {author} {\bibfnamefont
  {M.}~\bibnamefont {Rispler}}, \bibinfo {author} {\bibfnamefont
  {M.}~\bibnamefont {M\"uller}},\ and\ \bibinfo {author} {\bibfnamefont
  {T.}~\bibnamefont {Monz}},\ }\bibfield  {title} {\bibinfo {title}
  {Demonstration of fault-tolerant {S}teane quantum error correction},\ }\href
  {https://doi.org/10.1103/PRXQuantum.5.030326} {\bibfield  {journal} {\bibinfo
   {journal} {PRX Quantum}\ }\textbf {\bibinfo {volume} {5}},\ \bibinfo {pages}
  {030326} (\bibinfo {year} {2024})}\BibitemShut {NoStop}%
\bibitem [{\citenamefont {Lacroix}\ \emph {et~al.}(2024)\citenamefont
  {Lacroix}, \citenamefont {Bourassa}, \citenamefont {Heras}, \citenamefont
  {Zhang}, \citenamefont {Bausch}, \citenamefont {Senior}, \citenamefont
  {Edlich}, \citenamefont {Shutty}, \citenamefont {Sivak}, \citenamefont
  {Bengtsson}, \citenamefont {McEwen}, \citenamefont {Higgott}, \citenamefont
  {Kafri}, \citenamefont {Claes}, \citenamefont {Morvan}, \citenamefont {Chen},
  \citenamefont {Zalcman}, \citenamefont {Madhuk}, \citenamefont {Acharya},
  \citenamefont {Beni}, \citenamefont {Aigeldinger}, \citenamefont {Alcaraz},
  \citenamefont {Andersen}, \citenamefont {Ansmann}, \citenamefont {Arute},
  \citenamefont {Arya}, \citenamefont {Asfaw}, \citenamefont {Atalaya},
  \citenamefont {Babbush}, \citenamefont {Ballard}, \citenamefont {Bardin},
  \citenamefont {Bilmes}, \citenamefont {Blackwell}, \citenamefont {Bovaird},
  \citenamefont {Bowers}, \citenamefont {Brill}, \citenamefont {Broughton},
  \citenamefont {Browne}, \citenamefont {Buchea}, \citenamefont {Buckley},
  \citenamefont {Burger}, \citenamefont {Burkett}, \citenamefont {Bushnell},
  \citenamefont {Cabrera}, \citenamefont {Campero}, \citenamefont {Chang},
  \citenamefont {Chiaro}, \citenamefont {Chih}, \citenamefont {Cleland},
  \citenamefont {Cogan}, \citenamefont {Collins}, \citenamefont {Conner},
  \citenamefont {Courtney}, \citenamefont {Crook}, \citenamefont {Curtin},
  \citenamefont {Das}, \citenamefont {Demura}, \citenamefont {Lorenzo},
  \citenamefont {Paolo}, \citenamefont {Donohoe}, \citenamefont {Drozdov},
  \citenamefont {Dunsworth}, \citenamefont {Eickbusch}, \citenamefont {Elbag},
  \citenamefont {Elzouka}, \citenamefont {Erickson}, \citenamefont {Ferreira},
  \citenamefont {Burgos}, \citenamefont {Forati}, \citenamefont {Fowler},
  \citenamefont {Foxen}, \citenamefont {Ganjam}, \citenamefont {Garcia},
  \citenamefont {Gasca}, \citenamefont {Élie Genois}, \citenamefont {Giang},
  \citenamefont {Gilboa}, \citenamefont {Gosula}, \citenamefont {Dau},
  \citenamefont {Graumann}, \citenamefont {Greene}, \citenamefont {Gross},
  \citenamefont {Ha}, \citenamefont {Habegger}, \citenamefont {Hansen},
  \citenamefont {Harrigan}, \citenamefont {Harrington}, \citenamefont {Heslin},
  \citenamefont {Heu}, \citenamefont {Hiltermann}, \citenamefont {Hilton},
  \citenamefont {Hong}, \citenamefont {Huang}, \citenamefont {Huff},
  \citenamefont {Huggins}, \citenamefont {Jeffrey}, \citenamefont {Jiang},
  \citenamefont {Jin}, \citenamefont {Joshi}, \citenamefont {Juhas},
  \citenamefont {Kabel}, \citenamefont {Kang}, \citenamefont {Karamlou},
  \citenamefont {Kechedzhi}, \citenamefont {Khaire}, \citenamefont {Khattar},
  \citenamefont {Khezri}, \citenamefont {Kim}, \citenamefont {Klimov},
  \citenamefont {Kobrin}, \citenamefont {Korotkov}, \citenamefont {Kostritsa},
  \citenamefont {Kreikebaum}, \citenamefont {Kurilovich}, \citenamefont
  {Landhuis}, \citenamefont {Lange-Dei}, \citenamefont {Langley}, \citenamefont
  {Laptev}, \citenamefont {Lau}, \citenamefont {Ledford}, \citenamefont {Lee},
  \citenamefont {Lester}, \citenamefont {Guevel}, \citenamefont {Li},
  \citenamefont {Li}, \citenamefont {Lill}, \citenamefont {Livingston},
  \citenamefont {Locharla}, \citenamefont {Lucero}, \citenamefont {Lundahl},
  \citenamefont {Lunt}, \citenamefont {Maloney}, \citenamefont {Mandrà},
  \citenamefont {Martin}, \citenamefont {Martin}, \citenamefont {Maxfield},
  \citenamefont {McClean}, \citenamefont {Meeks}, \citenamefont {Megrant},
  \citenamefont {Miao}, \citenamefont {Molavi}, \citenamefont {Molina},
  \citenamefont {Montazeri}, \citenamefont {Movassagh}, \citenamefont {Neill},
  \citenamefont {Newman}, \citenamefont {Nguyen}, \citenamefont {Nguyen},
  \citenamefont {Ni}, \citenamefont {Niu}, \citenamefont {Oas}, \citenamefont
  {Oliver}, \citenamefont {Orosco}, \citenamefont {Ottosson}, \citenamefont
  {Pizzuto}, \citenamefont {Potter}, \citenamefont {Pritchard}, \citenamefont
  {Quintana}, \citenamefont {Ramachandran}, \citenamefont {Reagor},
  \citenamefont {Resnick}, \citenamefont {Rhodes}, \citenamefont {Roberts},
  \citenamefont {Rosenberg}, \citenamefont {Rosenfeld}, \citenamefont {Rossi},
  \citenamefont {Roushan}, \citenamefont {Sankaragomathi}, \citenamefont
  {Schurkus}, \citenamefont {Shearn}, \citenamefont {Shorter}, \citenamefont
  {Shvarts}, \citenamefont {Small}, \citenamefont {Smith}, \citenamefont
  {Springer}, \citenamefont {Sterling}, \citenamefont {Suchard}, \citenamefont
  {Szasz}, \citenamefont {Sztein}, \citenamefont {Thor}, \citenamefont
  {Tomita}, \citenamefont {Torres}, \citenamefont {Torunbalci}, \citenamefont
  {Vaishnav}, \citenamefont {Vargas}, \citenamefont {Vdovichev}, \citenamefont
  {Vidal}, \citenamefont {Heidweiller}, \citenamefont {Waltman}, \citenamefont
  {Waltz}, \citenamefont {Wang}, \citenamefont {Ware}, \citenamefont {Weidel},
  \citenamefont {White}, \citenamefont {Wong}, \citenamefont {Woo},
  \citenamefont {Woodson}, \citenamefont {Xing}, \citenamefont {Yao},
  \citenamefont {Yeh}, \citenamefont {Ying}, \citenamefont {Yoo}, \citenamefont
  {Yosri}, \citenamefont {Young}, \citenamefont {Zhang}, \citenamefont {Zhu},
  \citenamefont {Zobrist}, \citenamefont {Neven}, \citenamefont {Kohli},
  \citenamefont {Davies}, \citenamefont {Boixo}, \citenamefont {Kelly},
  \citenamefont {Jones}, \citenamefont {Gidney},\ and\ \citenamefont
  {Satzinger}}]{lacroix2024scalinglogiccolorcode}%
  \BibitemOpen
  \bibfield  {author} {\bibinfo {author} {\bibfnamefont {N.}~\bibnamefont
  {Lacroix}}, \bibinfo {author} {\bibfnamefont {A.}~\bibnamefont {Bourassa}},
  \bibinfo {author} {\bibfnamefont {F.~J.~H.}\ \bibnamefont {Heras}}, \bibinfo
  {author} {\bibfnamefont {L.~M.}\ \bibnamefont {Zhang}}, \bibinfo {author}
  {\bibfnamefont {J.}~\bibnamefont {Bausch}}, \bibinfo {author} {\bibfnamefont
  {A.~W.}\ \bibnamefont {Senior}}, \bibinfo {author} {\bibfnamefont
  {T.}~\bibnamefont {Edlich}}, \bibinfo {author} {\bibfnamefont
  {N.}~\bibnamefont {Shutty}}, \bibinfo {author} {\bibfnamefont
  {V.}~\bibnamefont {Sivak}}, \bibinfo {author} {\bibfnamefont
  {A.}~\bibnamefont {Bengtsson}}, \bibinfo {author} {\bibfnamefont
  {M.}~\bibnamefont {McEwen}}, \bibinfo {author} {\bibfnamefont
  {O.}~\bibnamefont {Higgott}}, \bibinfo {author} {\bibfnamefont
  {D.}~\bibnamefont {Kafri}}, \bibinfo {author} {\bibfnamefont
  {J.}~\bibnamefont {Claes}}, \bibinfo {author} {\bibfnamefont
  {A.}~\bibnamefont {Morvan}}, \bibinfo {author} {\bibfnamefont
  {Z.}~\bibnamefont {Chen}}, \bibinfo {author} {\bibfnamefont {A.}~\bibnamefont
  {Zalcman}}, \bibinfo {author} {\bibfnamefont {S.}~\bibnamefont {Madhuk}},
  \bibinfo {author} {\bibfnamefont {R.}~\bibnamefont {Acharya}}, \bibinfo
  {author} {\bibfnamefont {L.~A.}\ \bibnamefont {Beni}}, \bibinfo {author}
  {\bibfnamefont {G.}~\bibnamefont {Aigeldinger}}, \bibinfo {author}
  {\bibfnamefont {R.}~\bibnamefont {Alcaraz}}, \bibinfo {author} {\bibfnamefont
  {T.~I.}\ \bibnamefont {Andersen}}, \bibinfo {author} {\bibfnamefont
  {M.}~\bibnamefont {Ansmann}}, \bibinfo {author} {\bibfnamefont
  {F.}~\bibnamefont {Arute}}, \bibinfo {author} {\bibfnamefont
  {K.}~\bibnamefont {Arya}}, \bibinfo {author} {\bibfnamefont {A.}~\bibnamefont
  {Asfaw}}, \bibinfo {author} {\bibfnamefont {J.}~\bibnamefont {Atalaya}},
  \bibinfo {author} {\bibfnamefont {R.}~\bibnamefont {Babbush}}, \bibinfo
  {author} {\bibfnamefont {B.}~\bibnamefont {Ballard}}, \bibinfo {author}
  {\bibfnamefont {J.~C.}\ \bibnamefont {Bardin}}, \bibinfo {author}
  {\bibfnamefont {A.}~\bibnamefont {Bilmes}}, \bibinfo {author} {\bibfnamefont
  {S.}~\bibnamefont {Blackwell}}, \bibinfo {author} {\bibfnamefont
  {J.}~\bibnamefont {Bovaird}}, \bibinfo {author} {\bibfnamefont
  {D.}~\bibnamefont {Bowers}}, \bibinfo {author} {\bibfnamefont
  {L.}~\bibnamefont {Brill}}, \bibinfo {author} {\bibfnamefont
  {M.}~\bibnamefont {Broughton}}, \bibinfo {author} {\bibfnamefont {D.~A.}\
  \bibnamefont {Browne}}, \bibinfo {author} {\bibfnamefont {B.}~\bibnamefont
  {Buchea}}, \bibinfo {author} {\bibfnamefont {B.~B.}\ \bibnamefont {Buckley}},
  \bibinfo {author} {\bibfnamefont {T.}~\bibnamefont {Burger}}, \bibinfo
  {author} {\bibfnamefont {B.}~\bibnamefont {Burkett}}, \bibinfo {author}
  {\bibfnamefont {N.}~\bibnamefont {Bushnell}}, \bibinfo {author}
  {\bibfnamefont {A.}~\bibnamefont {Cabrera}}, \bibinfo {author} {\bibfnamefont
  {J.}~\bibnamefont {Campero}}, \bibinfo {author} {\bibfnamefont {H.-S.}\
  \bibnamefont {Chang}}, \bibinfo {author} {\bibfnamefont {B.}~\bibnamefont
  {Chiaro}}, \bibinfo {author} {\bibfnamefont {L.-Y.}\ \bibnamefont {Chih}},
  \bibinfo {author} {\bibfnamefont {A.~Y.}\ \bibnamefont {Cleland}}, \bibinfo
  {author} {\bibfnamefont {J.}~\bibnamefont {Cogan}}, \bibinfo {author}
  {\bibfnamefont {R.}~\bibnamefont {Collins}}, \bibinfo {author} {\bibfnamefont
  {P.}~\bibnamefont {Conner}}, \bibinfo {author} {\bibfnamefont
  {W.}~\bibnamefont {Courtney}}, \bibinfo {author} {\bibfnamefont {A.~L.}\
  \bibnamefont {Crook}}, \bibinfo {author} {\bibfnamefont {B.}~\bibnamefont
  {Curtin}}, \bibinfo {author} {\bibfnamefont {S.}~\bibnamefont {Das}},
  \bibinfo {author} {\bibfnamefont {S.}~\bibnamefont {Demura}}, \bibinfo
  {author} {\bibfnamefont {L.~D.}\ \bibnamefont {Lorenzo}}, \bibinfo {author}
  {\bibfnamefont {A.~D.}\ \bibnamefont {Paolo}}, \bibinfo {author}
  {\bibfnamefont {P.}~\bibnamefont {Donohoe}}, \bibinfo {author} {\bibfnamefont
  {I.}~\bibnamefont {Drozdov}}, \bibinfo {author} {\bibfnamefont
  {A.}~\bibnamefont {Dunsworth}}, \bibinfo {author} {\bibfnamefont
  {A.}~\bibnamefont {Eickbusch}}, \bibinfo {author} {\bibfnamefont {A.~M.}\
  \bibnamefont {Elbag}}, \bibinfo {author} {\bibfnamefont {M.}~\bibnamefont
  {Elzouka}}, \bibinfo {author} {\bibfnamefont {C.}~\bibnamefont {Erickson}},
  \bibinfo {author} {\bibfnamefont {V.~S.}\ \bibnamefont {Ferreira}}, \bibinfo
  {author} {\bibfnamefont {L.~F.}\ \bibnamefont {Burgos}}, \bibinfo {author}
  {\bibfnamefont {E.}~\bibnamefont {Forati}}, \bibinfo {author} {\bibfnamefont
  {A.~G.}\ \bibnamefont {Fowler}}, \bibinfo {author} {\bibfnamefont
  {B.}~\bibnamefont {Foxen}}, \bibinfo {author} {\bibfnamefont
  {S.}~\bibnamefont {Ganjam}}, \bibinfo {author} {\bibfnamefont
  {G.}~\bibnamefont {Garcia}}, \bibinfo {author} {\bibfnamefont
  {R.}~\bibnamefont {Gasca}}, \bibinfo {author} {\bibnamefont {Élie Genois}},
  \bibinfo {author} {\bibfnamefont {W.}~\bibnamefont {Giang}}, \bibinfo
  {author} {\bibfnamefont {D.}~\bibnamefont {Gilboa}}, \bibinfo {author}
  {\bibfnamefont {R.}~\bibnamefont {Gosula}}, \bibinfo {author} {\bibfnamefont
  {A.~G.}\ \bibnamefont {Dau}}, \bibinfo {author} {\bibfnamefont
  {D.}~\bibnamefont {Graumann}}, \bibinfo {author} {\bibfnamefont
  {A.}~\bibnamefont {Greene}}, \bibinfo {author} {\bibfnamefont {J.~A.}\
  \bibnamefont {Gross}}, \bibinfo {author} {\bibfnamefont {T.}~\bibnamefont
  {Ha}}, \bibinfo {author} {\bibfnamefont {S.}~\bibnamefont {Habegger}},
  \bibinfo {author} {\bibfnamefont {M.}~\bibnamefont {Hansen}}, \bibinfo
  {author} {\bibfnamefont {M.~P.}\ \bibnamefont {Harrigan}}, \bibinfo {author}
  {\bibfnamefont {S.~D.}\ \bibnamefont {Harrington}}, \bibinfo {author}
  {\bibfnamefont {S.}~\bibnamefont {Heslin}}, \bibinfo {author} {\bibfnamefont
  {P.}~\bibnamefont {Heu}}, \bibinfo {author} {\bibfnamefont {R.}~\bibnamefont
  {Hiltermann}}, \bibinfo {author} {\bibfnamefont {J.}~\bibnamefont {Hilton}},
  \bibinfo {author} {\bibfnamefont {S.}~\bibnamefont {Hong}}, \bibinfo {author}
  {\bibfnamefont {H.-Y.}\ \bibnamefont {Huang}}, \bibinfo {author}
  {\bibfnamefont {A.}~\bibnamefont {Huff}}, \bibinfo {author} {\bibfnamefont
  {W.~J.}\ \bibnamefont {Huggins}}, \bibinfo {author} {\bibfnamefont
  {E.}~\bibnamefont {Jeffrey}}, \bibinfo {author} {\bibfnamefont
  {Z.}~\bibnamefont {Jiang}}, \bibinfo {author} {\bibfnamefont
  {X.}~\bibnamefont {Jin}}, \bibinfo {author} {\bibfnamefont {C.}~\bibnamefont
  {Joshi}}, \bibinfo {author} {\bibfnamefont {P.}~\bibnamefont {Juhas}},
  \bibinfo {author} {\bibfnamefont {A.}~\bibnamefont {Kabel}}, \bibinfo
  {author} {\bibfnamefont {H.}~\bibnamefont {Kang}}, \bibinfo {author}
  {\bibfnamefont {A.~H.}\ \bibnamefont {Karamlou}}, \bibinfo {author}
  {\bibfnamefont {K.}~\bibnamefont {Kechedzhi}}, \bibinfo {author}
  {\bibfnamefont {T.}~\bibnamefont {Khaire}}, \bibinfo {author} {\bibfnamefont
  {T.}~\bibnamefont {Khattar}}, \bibinfo {author} {\bibfnamefont
  {M.}~\bibnamefont {Khezri}}, \bibinfo {author} {\bibfnamefont
  {S.}~\bibnamefont {Kim}}, \bibinfo {author} {\bibfnamefont {P.~V.}\
  \bibnamefont {Klimov}}, \bibinfo {author} {\bibfnamefont {B.}~\bibnamefont
  {Kobrin}}, \bibinfo {author} {\bibfnamefont {A.~N.}\ \bibnamefont
  {Korotkov}}, \bibinfo {author} {\bibfnamefont {F.}~\bibnamefont {Kostritsa}},
  \bibinfo {author} {\bibfnamefont {J.~M.}\ \bibnamefont {Kreikebaum}},
  \bibinfo {author} {\bibfnamefont {V.~D.}\ \bibnamefont {Kurilovich}},
  \bibinfo {author} {\bibfnamefont {D.}~\bibnamefont {Landhuis}}, \bibinfo
  {author} {\bibfnamefont {T.}~\bibnamefont {Lange-Dei}}, \bibinfo {author}
  {\bibfnamefont {B.~W.}\ \bibnamefont {Langley}}, \bibinfo {author}
  {\bibfnamefont {P.}~\bibnamefont {Laptev}}, \bibinfo {author} {\bibfnamefont
  {K.-M.}\ \bibnamefont {Lau}}, \bibinfo {author} {\bibfnamefont
  {J.}~\bibnamefont {Ledford}}, \bibinfo {author} {\bibfnamefont
  {K.}~\bibnamefont {Lee}}, \bibinfo {author} {\bibfnamefont {B.~J.}\
  \bibnamefont {Lester}}, \bibinfo {author} {\bibfnamefont {L.~L.}\
  \bibnamefont {Guevel}}, \bibinfo {author} {\bibfnamefont {W.~Y.}\
  \bibnamefont {Li}}, \bibinfo {author} {\bibfnamefont {Y.}~\bibnamefont {Li}},
  \bibinfo {author} {\bibfnamefont {A.~T.}\ \bibnamefont {Lill}}, \bibinfo
  {author} {\bibfnamefont {W.~P.}\ \bibnamefont {Livingston}}, \bibinfo
  {author} {\bibfnamefont {A.}~\bibnamefont {Locharla}}, \bibinfo {author}
  {\bibfnamefont {E.}~\bibnamefont {Lucero}}, \bibinfo {author} {\bibfnamefont
  {D.}~\bibnamefont {Lundahl}}, \bibinfo {author} {\bibfnamefont
  {A.}~\bibnamefont {Lunt}}, \bibinfo {author} {\bibfnamefont {A.}~\bibnamefont
  {Maloney}}, \bibinfo {author} {\bibfnamefont {S.}~\bibnamefont {Mandrà}},
  \bibinfo {author} {\bibfnamefont {L.~S.}\ \bibnamefont {Martin}}, \bibinfo
  {author} {\bibfnamefont {O.}~\bibnamefont {Martin}}, \bibinfo {author}
  {\bibfnamefont {C.}~\bibnamefont {Maxfield}}, \bibinfo {author}
  {\bibfnamefont {J.~R.}\ \bibnamefont {McClean}}, \bibinfo {author}
  {\bibfnamefont {S.}~\bibnamefont {Meeks}}, \bibinfo {author} {\bibfnamefont
  {A.}~\bibnamefont {Megrant}}, \bibinfo {author} {\bibfnamefont {K.~C.}\
  \bibnamefont {Miao}}, \bibinfo {author} {\bibfnamefont {R.}~\bibnamefont
  {Molavi}}, \bibinfo {author} {\bibfnamefont {S.}~\bibnamefont {Molina}},
  \bibinfo {author} {\bibfnamefont {S.}~\bibnamefont {Montazeri}}, \bibinfo
  {author} {\bibfnamefont {R.}~\bibnamefont {Movassagh}}, \bibinfo {author}
  {\bibfnamefont {C.}~\bibnamefont {Neill}}, \bibinfo {author} {\bibfnamefont
  {M.}~\bibnamefont {Newman}}, \bibinfo {author} {\bibfnamefont
  {A.}~\bibnamefont {Nguyen}}, \bibinfo {author} {\bibfnamefont
  {M.}~\bibnamefont {Nguyen}}, \bibinfo {author} {\bibfnamefont {C.-H.}\
  \bibnamefont {Ni}}, \bibinfo {author} {\bibfnamefont {M.~Y.}\ \bibnamefont
  {Niu}}, \bibinfo {author} {\bibfnamefont {L.}~\bibnamefont {Oas}}, \bibinfo
  {author} {\bibfnamefont {W.~D.}\ \bibnamefont {Oliver}}, \bibinfo {author}
  {\bibfnamefont {R.}~\bibnamefont {Orosco}}, \bibinfo {author} {\bibfnamefont
  {K.}~\bibnamefont {Ottosson}}, \bibinfo {author} {\bibfnamefont
  {A.}~\bibnamefont {Pizzuto}}, \bibinfo {author} {\bibfnamefont
  {R.}~\bibnamefont {Potter}}, \bibinfo {author} {\bibfnamefont
  {O.}~\bibnamefont {Pritchard}}, \bibinfo {author} {\bibfnamefont
  {C.}~\bibnamefont {Quintana}}, \bibinfo {author} {\bibfnamefont
  {G.}~\bibnamefont {Ramachandran}}, \bibinfo {author} {\bibfnamefont {M.~J.}\
  \bibnamefont {Reagor}}, \bibinfo {author} {\bibfnamefont {R.}~\bibnamefont
  {Resnick}}, \bibinfo {author} {\bibfnamefont {D.~M.}\ \bibnamefont {Rhodes}},
  \bibinfo {author} {\bibfnamefont {G.}~\bibnamefont {Roberts}}, \bibinfo
  {author} {\bibfnamefont {E.}~\bibnamefont {Rosenberg}}, \bibinfo {author}
  {\bibfnamefont {E.}~\bibnamefont {Rosenfeld}}, \bibinfo {author}
  {\bibfnamefont {E.}~\bibnamefont {Rossi}}, \bibinfo {author} {\bibfnamefont
  {P.}~\bibnamefont {Roushan}}, \bibinfo {author} {\bibfnamefont
  {K.}~\bibnamefont {Sankaragomathi}}, \bibinfo {author} {\bibfnamefont
  {H.~F.}\ \bibnamefont {Schurkus}}, \bibinfo {author} {\bibfnamefont {M.~J.}\
  \bibnamefont {Shearn}}, \bibinfo {author} {\bibfnamefont {A.}~\bibnamefont
  {Shorter}}, \bibinfo {author} {\bibfnamefont {V.}~\bibnamefont {Shvarts}},
  \bibinfo {author} {\bibfnamefont {S.}~\bibnamefont {Small}}, \bibinfo
  {author} {\bibfnamefont {W.~C.}\ \bibnamefont {Smith}}, \bibinfo {author}
  {\bibfnamefont {S.}~\bibnamefont {Springer}}, \bibinfo {author}
  {\bibfnamefont {G.}~\bibnamefont {Sterling}}, \bibinfo {author}
  {\bibfnamefont {J.}~\bibnamefont {Suchard}}, \bibinfo {author} {\bibfnamefont
  {A.}~\bibnamefont {Szasz}}, \bibinfo {author} {\bibfnamefont
  {A.}~\bibnamefont {Sztein}}, \bibinfo {author} {\bibfnamefont
  {D.}~\bibnamefont {Thor}}, \bibinfo {author} {\bibfnamefont {E.}~\bibnamefont
  {Tomita}}, \bibinfo {author} {\bibfnamefont {A.}~\bibnamefont {Torres}},
  \bibinfo {author} {\bibfnamefont {M.~M.}\ \bibnamefont {Torunbalci}},
  \bibinfo {author} {\bibfnamefont {A.}~\bibnamefont {Vaishnav}}, \bibinfo
  {author} {\bibfnamefont {J.}~\bibnamefont {Vargas}}, \bibinfo {author}
  {\bibfnamefont {S.}~\bibnamefont {Vdovichev}}, \bibinfo {author}
  {\bibfnamefont {G.}~\bibnamefont {Vidal}}, \bibinfo {author} {\bibfnamefont
  {C.~V.}\ \bibnamefont {Heidweiller}}, \bibinfo {author} {\bibfnamefont
  {S.}~\bibnamefont {Waltman}}, \bibinfo {author} {\bibfnamefont
  {J.}~\bibnamefont {Waltz}}, \bibinfo {author} {\bibfnamefont {S.~X.}\
  \bibnamefont {Wang}}, \bibinfo {author} {\bibfnamefont {B.}~\bibnamefont
  {Ware}}, \bibinfo {author} {\bibfnamefont {T.}~\bibnamefont {Weidel}},
  \bibinfo {author} {\bibfnamefont {T.}~\bibnamefont {White}}, \bibinfo
  {author} {\bibfnamefont {K.}~\bibnamefont {Wong}}, \bibinfo {author}
  {\bibfnamefont {B.~W.~K.}\ \bibnamefont {Woo}}, \bibinfo {author}
  {\bibfnamefont {M.}~\bibnamefont {Woodson}}, \bibinfo {author} {\bibfnamefont
  {C.}~\bibnamefont {Xing}}, \bibinfo {author} {\bibfnamefont {Z.~J.}\
  \bibnamefont {Yao}}, \bibinfo {author} {\bibfnamefont {P.}~\bibnamefont
  {Yeh}}, \bibinfo {author} {\bibfnamefont {B.}~\bibnamefont {Ying}}, \bibinfo
  {author} {\bibfnamefont {J.}~\bibnamefont {Yoo}}, \bibinfo {author}
  {\bibfnamefont {N.}~\bibnamefont {Yosri}}, \bibinfo {author} {\bibfnamefont
  {G.}~\bibnamefont {Young}}, \bibinfo {author} {\bibfnamefont
  {Y.}~\bibnamefont {Zhang}}, \bibinfo {author} {\bibfnamefont
  {N.}~\bibnamefont {Zhu}}, \bibinfo {author} {\bibfnamefont {N.}~\bibnamefont
  {Zobrist}}, \bibinfo {author} {\bibfnamefont {H.}~\bibnamefont {Neven}},
  \bibinfo {author} {\bibfnamefont {P.}~\bibnamefont {Kohli}}, \bibinfo
  {author} {\bibfnamefont {A.}~\bibnamefont {Davies}}, \bibinfo {author}
  {\bibfnamefont {S.}~\bibnamefont {Boixo}}, \bibinfo {author} {\bibfnamefont
  {J.}~\bibnamefont {Kelly}}, \bibinfo {author} {\bibfnamefont
  {C.}~\bibnamefont {Jones}}, \bibinfo {author} {\bibfnamefont
  {C.}~\bibnamefont {Gidney}},\ and\ \bibinfo {author} {\bibfnamefont {K.~J.}\
  \bibnamefont {Satzinger}},\ }\href {https://arxiv.org/abs/2412.14256}
  {\bibinfo {title} {Scaling and logic in the color code on a superconducting
  quantum processor}} (\bibinfo {year} {2024}),\ \Eprint
  {https://arxiv.org/abs/2412.14256} {arXiv:2412.14256 [quant-ph]} \BibitemShut
  {NoStop}%
\bibitem [{\citenamefont {Gottesman}(2016)}]{gottesman2016surviving}%
  \BibitemOpen
  \bibfield  {author} {\bibinfo {author} {\bibfnamefont {D.}~\bibnamefont
  {Gottesman}},\ }\bibfield  {title} {\bibinfo {title} {Surviving as a quantum
  computer in a classical world},\ }\href
  {https://www.cs.umd.edu/class/spring2024/cmsc858G/QECCbook-2024-ch1-11.pdf}
  {\bibfield  {journal} {\bibinfo  {journal} {Textbook manuscript preprint}\ }
  (\bibinfo {year} {2016})}\BibitemShut {NoStop}%
\bibitem [{\citenamefont {Nielsen}\ and\ \citenamefont
  {Chuang}(2010)}]{QC_book}%
  \BibitemOpen
  \bibfield  {author} {\bibinfo {author} {\bibfnamefont {M.~A.}\ \bibnamefont
  {Nielsen}}\ and\ \bibinfo {author} {\bibfnamefont {I.~L.}\ \bibnamefont
  {Chuang}},\ }\href {https://doi.org/10.1017/CBO9780511976667} {\emph
  {\bibinfo {title} {{Quantum Computation and Quantum Information: 10th
  Anniversary Edition}}}}\ (\bibinfo  {publisher} {Cambridge University
  Press},\ \bibinfo {address} {Cambridge},\ \bibinfo {year} {2010})\BibitemShut
  {NoStop}%
\bibitem [{\citenamefont
  {Gottesman}(1997)}]{gottesman1997stabilizercodesquantumerror}%
  \BibitemOpen
  \bibfield  {author} {\bibinfo {author} {\bibfnamefont {D.}~\bibnamefont
  {Gottesman}},\ }\href {https://arxiv.org/abs/quant-ph/9705052} {\bibinfo
  {title} {Stabilizer codes and quantum error correction}} (\bibinfo {year}
  {1997}),\ \Eprint {https://arxiv.org/abs/quant-ph/9705052}
  {arXiv:quant-ph/9705052 [quant-ph]} \BibitemShut {NoStop}%
\bibitem [{\citenamefont {Iyer}\ and\ \citenamefont
  {Poulin}(2018)}]{iyer2018small}%
  \BibitemOpen
  \bibfield  {author} {\bibinfo {author} {\bibfnamefont {P.}~\bibnamefont
  {Iyer}}\ and\ \bibinfo {author} {\bibfnamefont {D.}~\bibnamefont {Poulin}},\
  }\bibfield  {title} {\bibinfo {title} {A small quantum computer is needed to
  optimize fault-tolerant protocols},\ }\href@noop {} {\bibfield  {journal}
  {\bibinfo  {journal} {QST}\ }\textbf {\bibinfo {volume} {3}},\ \bibinfo
  {pages} {030504} (\bibinfo {year} {2018})}\BibitemShut {NoStop}%
\bibitem [{\citenamefont {Rudinger}\ \emph {et~al.}(2023)\citenamefont
  {Rudinger}, \citenamefont {Ziyad}, \citenamefont {Morford-Oberst},
  \citenamefont {Campos}, \citenamefont {Seritan}, \citenamefont {Metodi},\
  and\ \citenamefont {Blume-Kohout}}]{rudinger2023probing}%
  \BibitemOpen
  \bibfield  {author} {\bibinfo {author} {\bibfnamefont {K.}~\bibnamefont
  {Rudinger}}, \bibinfo {author} {\bibfnamefont {J.}~\bibnamefont {Ziyad}},
  \bibinfo {author} {\bibfnamefont {M.}~\bibnamefont {Morford-Oberst}},
  \bibinfo {author} {\bibfnamefont {J.}~\bibnamefont {Campos}}, \bibinfo
  {author} {\bibfnamefont {S.}~\bibnamefont {Seritan}}, \bibinfo {author}
  {\bibfnamefont {T.}~\bibnamefont {Metodi}},\ and\ \bibinfo {author}
  {\bibfnamefont {R.}~\bibnamefont {Blume-Kohout}},\ }\bibfield  {title}
  {\bibinfo {title} {Probing logical error models with gate-set tomography},\
  }in\ \href@noop {} {\emph {\bibinfo {booktitle} {APS March Meeting
  Abstracts}}},\ Vol.\ \bibinfo {volume} {2023}\ (\bibinfo {year} {2023})\ pp.\
  \bibinfo {pages} {D72--011}\BibitemShut {NoStop}%
\bibitem [{\citenamefont {Beale}\ \emph {et~al.}(2018)\citenamefont {Beale},
  \citenamefont {Wallman}, \citenamefont {Gutiérrez}, \citenamefont {Brown},\
  and\ \citenamefont {Laflamme}}]{beale_quantum_2018}%
  \BibitemOpen
  \bibfield  {author} {\bibinfo {author} {\bibfnamefont {S.~J.}\ \bibnamefont
  {Beale}}, \bibinfo {author} {\bibfnamefont {J.~J.}\ \bibnamefont {Wallman}},
  \bibinfo {author} {\bibfnamefont {M.}~\bibnamefont {Gutiérrez}}, \bibinfo
  {author} {\bibfnamefont {K.~R.}\ \bibnamefont {Brown}},\ and\ \bibinfo
  {author} {\bibfnamefont {R.}~\bibnamefont {Laflamme}},\ }\bibfield  {title}
  {\bibinfo {title} {Quantum error correction decoheres noise},\ }\bibfield
  {journal} {\bibinfo  {journal} {Phys. Rev. Lett.}\ }\textbf {\bibinfo
  {volume} {121}},\ \href {https://doi.org/10.1103/physrevlett.121.190501}
  {10.1103/physrevlett.121.190501} (\bibinfo {year} {2018})\BibitemShut
  {NoStop}%
\bibitem [{\citenamefont {Rahn}\ \emph {et~al.}(2002)\citenamefont {Rahn},
  \citenamefont {Doherty},\ and\ \citenamefont {Mabuchi}}]{rahn2002exact}%
  \BibitemOpen
  \bibfield  {author} {\bibinfo {author} {\bibfnamefont {B.}~\bibnamefont
  {Rahn}}, \bibinfo {author} {\bibfnamefont {A.~C.}\ \bibnamefont {Doherty}},\
  and\ \bibinfo {author} {\bibfnamefont {H.}~\bibnamefont {Mabuchi}},\
  }\bibfield  {title} {\bibinfo {title} {Exact performance of concatenated
  quantum codes},\ }\href@noop {} {\bibfield  {journal} {\bibinfo  {journal}
  {Phys. Rev. A}\ }\textbf {\bibinfo {volume} {66}},\ \bibinfo {pages} {032304}
  (\bibinfo {year} {2002})}\BibitemShut {NoStop}%
\bibitem [{\citenamefont {Caesura}(2021)}]{caesura2021non}%
  \BibitemOpen
  \bibfield  {author} {\bibinfo {author} {\bibfnamefont {A.}~\bibnamefont
  {Caesura}},\ }\emph {\bibinfo {title} {Non-{M}arkovianity in Logical Fidelity
  Estimation}},\ \href@noop {} {Master's thesis},\ \bibinfo  {school}
  {University of Waterloo} (\bibinfo {year} {2021})\BibitemShut {NoStop}%
\bibitem [{\citenamefont {Ziyad}\ \emph {et~al.}(2025)\citenamefont {Ziyad},
  \citenamefont {Blume-Kohout}, \citenamefont {Metodi},\ and\ \citenamefont
  {Rudinger}}]{Ziyad2025Emergent}%
  \BibitemOpen
  \bibfield  {author} {\bibinfo {author} {\bibfnamefont {J.~A.}\ \bibnamefont
  {Ziyad}}, \bibinfo {author} {\bibfnamefont {R.}~\bibnamefont {Blume-Kohout}},
  \bibinfo {author} {\bibfnamefont {T.}~\bibnamefont {Metodi}},\ and\ \bibinfo
  {author} {\bibfnamefont {K.}~\bibnamefont {Rudinger}},\ }\bibfield  {title}
  {\bibinfo {title} {Emergent non-markovian dynamics in logical qubit
  systems},\ }in\ \href {https://schedule.aps.org/smt/2025/events/MAR-C35/6}
  {\emph {\bibinfo {booktitle} {APS March Meeting 2025}}}\ (\bibinfo {year}
  {2025})\ \bibinfo {note} {talk MAR-C35.00006}\BibitemShut {NoStop}%
\bibitem [{\citenamefont {Calderbank}\ and\ \citenamefont
  {Shor}(1996)}]{CalderbankGood}%
  \BibitemOpen
  \bibfield  {author} {\bibinfo {author} {\bibfnamefont {A.~R.}\ \bibnamefont
  {Calderbank}}\ and\ \bibinfo {author} {\bibfnamefont {P.~W.}\ \bibnamefont
  {Shor}},\ }\bibfield  {title} {\bibinfo {title} {Good quantum
  error-correcting codes exist},\ }\href
  {https://journals.aps.org/pra/abstract/10.1103/PhysRevA.54.1098} {\bibfield
  {journal} {\bibinfo  {journal} {Phys. Rev. A}\ }\textbf {\bibinfo {volume}
  {54}},\ \bibinfo {pages} {1098} (\bibinfo {year} {1996})}\BibitemShut
  {NoStop}%
\bibitem [{\citenamefont {Gottesman}(2009)}]{GottesmanIntro}%
  \BibitemOpen
  \bibfield  {author} {\bibinfo {author} {\bibfnamefont {D.}~\bibnamefont
  {Gottesman}},\ }\href {https://arxiv.org/abs/0904.2557} {\bibinfo {title} {An
  introduction to quantum error correction and fault-tolerant quantum
  computation}} (\bibinfo {year} {2009}),\ \Eprint
  {https://arxiv.org/abs/0904.2557} {arXiv:0904.2557 [quant-ph]} \BibitemShut
  {NoStop}%
\bibitem [{\citenamefont {Steane}(1996{\natexlab{a}})}]{SteaneQEC1}%
  \BibitemOpen
  \bibfield  {author} {\bibinfo {author} {\bibfnamefont {A.}~\bibnamefont
  {Steane}},\ }\bibfield  {title} {\bibinfo {title} {Multiple-particle
  interference and quantum error correction},\ }\href
  {https://royalsocietypublishing.org/doi/10.1098/rspa.1996.0136} {\bibfield
  {journal} {\bibinfo  {journal} {Proc. Roy. Soc. London. A}\ }\textbf
  {\bibinfo {volume} {452}},\ \bibinfo {pages} {2551} (\bibinfo {year}
  {1996}{\natexlab{a}})}\BibitemShut {NoStop}%
\bibitem [{\citenamefont {Steane}(1996{\natexlab{b}})}]{SteaneQEC2}%
  \BibitemOpen
  \bibfield  {author} {\bibinfo {author} {\bibfnamefont {A.~M.}\ \bibnamefont
  {Steane}},\ }\bibfield  {title} {\bibinfo {title} {Error-correcting codes in
  quantum theory},\ }\href {https://doi.org/10.1103/PhysRevLett.77.793}
  {\bibfield  {journal} {\bibinfo  {journal} {Phys. Rev. Lett.}\ }\textbf
  {\bibinfo {volume} {77}},\ \bibinfo {pages} {793} (\bibinfo {year}
  {1996}{\natexlab{b}})}\BibitemShut {NoStop}%
\bibitem [{\citenamefont {Shor}(1995)}]{Shor_1995}%
  \BibitemOpen
  \bibfield  {author} {\bibinfo {author} {\bibfnamefont {P.~W.}\ \bibnamefont
  {Shor}},\ }\bibfield  {title} {\bibinfo {title} {Scheme for reducing
  decoherence in quantum computer memory},\ }\href
  {https://doi.org/10.1103/PhysRevA.52.R2493} {\bibfield  {journal} {\bibinfo
  {journal} {Phys. Rev. A}\ }\textbf {\bibinfo {volume} {52}},\ \bibinfo
  {pages} {R2493} (\bibinfo {year} {1995})}\BibitemShut {NoStop}%
\bibitem [{\citenamefont {Lidar}\ and\ \citenamefont
  {Brun}(2013)}]{lidar_quantum_2013}%
  \BibitemOpen
  \bibinfo {editor} {\bibfnamefont {D.~A.}\ \bibnamefont {Lidar}}\ and\
  \bibinfo {editor} {\bibfnamefont {T.~A.}\ \bibnamefont {Brun}},\ eds.,\ \href
  {https://doi.org/10.1017/cbo9781139034807} {\emph {\bibinfo {title} {{Quantum
  Error Correction}}}},\ \bibinfo {edition} {1st}\ ed.\ (\bibinfo  {publisher}
  {Cambridge University Press},\ \bibinfo {year} {2013})\BibitemShut {NoStop}%
\bibitem [{\citenamefont {Combes}\ \emph {et~al.}(2017)\citenamefont {Combes},
  \citenamefont {Granade}, \citenamefont {Ferrie},\ and\ \citenamefont
  {Flammia}}]{combes2017logicalrandomizedbenchmarking}%
  \BibitemOpen
  \bibfield  {author} {\bibinfo {author} {\bibfnamefont {J.}~\bibnamefont
  {Combes}}, \bibinfo {author} {\bibfnamefont {C.}~\bibnamefont {Granade}},
  \bibinfo {author} {\bibfnamefont {C.}~\bibnamefont {Ferrie}},\ and\ \bibinfo
  {author} {\bibfnamefont {S.~T.}\ \bibnamefont {Flammia}},\ }\href
  {https://arxiv.org/abs/1702.03688} {\bibinfo {title} {Logical randomized
  benchmarking}} (\bibinfo {year} {2017}),\ \Eprint
  {https://arxiv.org/abs/1702.03688} {arXiv:1702.03688 [quant-ph]} \BibitemShut
  {NoStop}%
\bibitem [{\citenamefont {Stinespring}(1955)}]{Stinespring}%
  \BibitemOpen
  \bibfield  {author} {\bibinfo {author} {\bibfnamefont {W.~F.}\ \bibnamefont
  {Stinespring}},\ }\bibfield  {title} {\bibinfo {title} {Positive functions on
  {{C}}{\textsuperscript{*}}-algebras},\ }\href
  {https://doi.org/10.2307/2032342} {\bibfield  {journal} {\bibinfo  {journal}
  {Proc. Amer. Math. Soc.}\ }\textbf {\bibinfo {volume} {6}},\ \bibinfo {pages}
  {211} (\bibinfo {year} {1955})}\BibitemShut {NoStop}%
\bibitem [{\citenamefont {Hellwig}\ and\ \citenamefont
  {Kraus}(1969)}]{KrausMeas1969}%
  \BibitemOpen
  \bibfield  {author} {\bibinfo {author} {\bibfnamefont {K.-E.}\ \bibnamefont
  {Hellwig}}\ and\ \bibinfo {author} {\bibfnamefont {K.}~\bibnamefont
  {Kraus}},\ }\bibfield  {title} {\bibinfo {title} {Pure operations and
  measurements},\ }\href {https://doi.org/10.1007/BF01645807} {\bibfield
  {journal} {\bibinfo  {journal} {Commun. Math. Phys.}\ }\textbf {\bibinfo
  {volume} {11}},\ \bibinfo {pages} {214} (\bibinfo {year} {1969})}\BibitemShut
  {NoStop}%
\bibitem [{\citenamefont {Kraus}(1983)}]{KrausBook}%
  \BibitemOpen
  \bibfield  {author} {\bibinfo {author} {\bibfnamefont {K.}~\bibnamefont
  {Kraus}},\ }\href {https://doi.org/10.1007/3-540-12732-1} {\emph {\bibinfo
  {title} {{States, Effects, and Operations. Fundamental Notions of Quantum
  Theory}}}},\ Lecture Notes in Physics\ (\bibinfo  {publisher} {Springer
  Berlin},\ \bibinfo {address} {Heidelberg, Germany},\ \bibinfo {year}
  {1983})\BibitemShut {NoStop}%
\bibitem [{\citenamefont {Friedman}\ \emph {et~al.}(2022)\citenamefont
  {Friedman}, \citenamefont {Yin}, \citenamefont {Hong},\ and\ \citenamefont
  {Lucas}}]{SpeedLimit}%
  \BibitemOpen
  \bibfield  {author} {\bibinfo {author} {\bibfnamefont {A.~J.}\ \bibnamefont
  {Friedman}}, \bibinfo {author} {\bibfnamefont {C.}~\bibnamefont {Yin}},
  \bibinfo {author} {\bibfnamefont {Y.}~\bibnamefont {Hong}},\ and\ \bibinfo
  {author} {\bibfnamefont {A.}~\bibnamefont {Lucas}},\ }\href
  {https://arxiv.org/abs/2206.09929} {\bibinfo {title} {Locality and error
  correction in quantum dynamics with measurement}} (\bibinfo {year} {2022}),\
  \Eprint {https://arxiv.org/abs/2206.09929} {arXiv:2206.09929 [quant-ph]}
  \BibitemShut {NoStop}%
\bibitem [{\citenamefont {Barberena}\ and\ \citenamefont
  {Friedman}(2024)}]{DiegoMeasOverview}%
  \BibitemOpen
  \bibfield  {author} {\bibinfo {author} {\bibfnamefont {D.}~\bibnamefont
  {Barberena}}\ and\ \bibinfo {author} {\bibfnamefont {A.~J.}\ \bibnamefont
  {Friedman}},\ }\href@noop {} {\bibinfo {title} {Overview of projective
  quantum measurements}} (\bibinfo {year} {2024}),\ \Eprint
  {https://arxiv.org/abs/2404.05679} {arXiv:2404.05679 [quant-ph]} \BibitemShut
  {NoStop}%
\bibitem [{\citenamefont {Chao}\ and\ \citenamefont
  {Reichardt}(2018)}]{PhysRevLett.121.050502}%
  \BibitemOpen
  \bibfield  {author} {\bibinfo {author} {\bibfnamefont {R.}~\bibnamefont
  {Chao}}\ and\ \bibinfo {author} {\bibfnamefont {B.~W.}\ \bibnamefont
  {Reichardt}},\ }\bibfield  {title} {\bibinfo {title} {Quantum error
  correction with only two extra qubits},\ }\href
  {https://doi.org/10.1103/PhysRevLett.121.050502} {\bibfield  {journal}
  {\bibinfo  {journal} {Phys. Rev. Lett.}\ }\textbf {\bibinfo {volume} {121}},\
  \bibinfo {pages} {050502} (\bibinfo {year} {2018})}\BibitemShut {NoStop}%
\bibitem [{\citenamefont {Fowler}\ \emph {et~al.}(2012)\citenamefont {Fowler},
  \citenamefont {Mariantoni}, \citenamefont {Martinis},\ and\ \citenamefont
  {Cleland}}]{PhysRevA.86.032324}%
  \BibitemOpen
  \bibfield  {author} {\bibinfo {author} {\bibfnamefont {A.~G.}\ \bibnamefont
  {Fowler}}, \bibinfo {author} {\bibfnamefont {M.}~\bibnamefont {Mariantoni}},
  \bibinfo {author} {\bibfnamefont {J.~M.}\ \bibnamefont {Martinis}},\ and\
  \bibinfo {author} {\bibfnamefont {A.~N.}\ \bibnamefont {Cleland}},\
  }\bibfield  {title} {\bibinfo {title} {Surface codes: Towards practical
  large-scale quantum computation},\ }\href
  {https://doi.org/10.1103/PhysRevA.86.032324} {\bibfield  {journal} {\bibinfo
  {journal} {Phys. Rev. A}\ }\textbf {\bibinfo {volume} {86}},\ \bibinfo
  {pages} {032324} (\bibinfo {year} {2012})}\BibitemShut {NoStop}%
\bibitem [{\citenamefont {Knill}(2005)}]{Knill2005}%
  \BibitemOpen
  \bibfield  {author} {\bibinfo {author} {\bibfnamefont {E.}~\bibnamefont
  {Knill}},\ }\bibfield  {title} {\bibinfo {title} {Quantum computing with
  realistically noisy devices},\ }\href {https://doi.org/10.1038/nature03350}
  {\bibfield  {journal} {\bibinfo  {journal} {Nature}\ }\textbf {\bibinfo
  {volume} {434}},\ \bibinfo {pages} {39} (\bibinfo {year} {2005})}\BibitemShut
  {NoStop}%
\bibitem [{\citenamefont {Chamberland}\ \emph {et~al.}(2018)\citenamefont
  {Chamberland}, \citenamefont {Iyer},\ and\ \citenamefont
  {Poulin}}]{Chamberland_2018}%
  \BibitemOpen
  \bibfield  {author} {\bibinfo {author} {\bibfnamefont {C.}~\bibnamefont
  {Chamberland}}, \bibinfo {author} {\bibfnamefont {P.}~\bibnamefont {Iyer}},\
  and\ \bibinfo {author} {\bibfnamefont {D.}~\bibnamefont {Poulin}},\
  }\bibfield  {title} {\bibinfo {title} {Fault-tolerant quantum computing in
  the pauli or clifford frame with slow error diagnostics},\ }\href
  {https://doi.org/10.22331/q-2018-01-04-43} {\bibfield  {journal} {\bibinfo
  {journal} {Quantum}\ }\textbf {\bibinfo {volume} {2}},\ \bibinfo {pages} {43}
  (\bibinfo {year} {2018})}\BibitemShut {NoStop}%
\bibitem [{\citenamefont {Nielsen}(2002)}]{Nielsen_2002}%
  \BibitemOpen
  \bibfield  {author} {\bibinfo {author} {\bibfnamefont {M.~A.}\ \bibnamefont
  {Nielsen}},\ }\bibfield  {title} {\bibinfo {title} {A simple formula for the
  average gate fidelity of a quantum dynamical operation},\ }\href
  {https://doi.org/10.1016/s0375-9601(02)01272-0} {\bibfield  {journal}
  {\bibinfo  {journal} {Phys. Lett. A}\ }\textbf {\bibinfo {volume} {303}},\
  \bibinfo {pages} {249} (\bibinfo {year} {2002})}\BibitemShut {NoStop}%
\bibitem [{\citenamefont {Chen}\ \emph {et~al.}(2023)\citenamefont {Chen},
  \citenamefont {Liu}, \citenamefont {Otten}, \citenamefont {Seif},
  \citenamefont {Fefferman},\ and\ \citenamefont {Jiang}}]{Pauli_learn}%
  \BibitemOpen
  \bibfield  {author} {\bibinfo {author} {\bibfnamefont {S.}~\bibnamefont
  {Chen}}, \bibinfo {author} {\bibfnamefont {Y.}~\bibnamefont {Liu}}, \bibinfo
  {author} {\bibfnamefont {M.}~\bibnamefont {Otten}}, \bibinfo {author}
  {\bibfnamefont {A.}~\bibnamefont {Seif}}, \bibinfo {author} {\bibfnamefont
  {B.}~\bibnamefont {Fefferman}},\ and\ \bibinfo {author} {\bibfnamefont
  {L.}~\bibnamefont {Jiang}},\ }\bibfield  {title} {\bibinfo {title} {The
  learnability of {P}auli noise},\ }\href
  {https://doi.org/10.1038/s41467-022-35759-4} {\bibfield  {journal} {\bibinfo
  {journal} {Nat. Commun.}\ }\textbf {\bibinfo {volume} {14}},\ \bibinfo
  {pages} {52} (\bibinfo {year} {2023})}\BibitemShut {NoStop}%
\bibitem [{\citenamefont {Horn}\ and\ \citenamefont
  {Johnson}(2013)}]{HornJohnson2013}%
  \BibitemOpen
  \bibfield  {author} {\bibinfo {author} {\bibfnamefont {R.~A.}\ \bibnamefont
  {Horn}}\ and\ \bibinfo {author} {\bibfnamefont {C.~R.}\ \bibnamefont
  {Johnson}},\ }\href@noop {} {\emph {\bibinfo {title} {Matrix Analysis}}},\
  \bibinfo {edition} {2nd}\ ed.\ (\bibinfo  {publisher} {Cambridge University
  Press},\ \bibinfo {year} {2013})\BibitemShut {NoStop}%
\end{thebibliography}%
\end{document}